\documentclass[11pt]{article}
\usepackage{graphicx} 
 \usepackage{setspace}

 \usepackage{amssymb}
\usepackage{amsmath}
 \usepackage{bm}
 \usepackage{graphicx}
 \usepackage{mathrsfs}
  \usepackage{fancyhdr}
   \usepackage[colorlinks, linkcolor=blue]{hyperref}
    \usepackage{amsthm}
     \usepackage{booktabs}
 \usepackage{makecell}
 \usepackage{float}
 \usepackage{footmisc}
 \usepackage{tikz}
 \usetikzlibrary{external}
 \usetikzlibrary{decorations.pathreplacing}
 \usepackage{appendix}
\usepackage{pgfplots}
\pgfplotsset{compat=newest}
\usepackage{natbib}
\usepackage{listings}
\usepackage{tocloft}
 \usepackage{wasysym}
 \usepackage{ragged2e}
 \usepackage{hyperref}
 \usepackage{comment}
 \usepackage{algorithm2e}
 \usepackage{array}
 \usepackage{multirow}
 \usepackage{relsize}
\usepackage{verbatim}
\usepackage{caption}
\usepackage{accents}

 \usetikzlibrary{positioning}
\usetikzlibrary{calc}
\pgfplotsset{compat=1.18}

\usepackage[T1]{fontenc}   
\usepackage{kpfonts}

\hypersetup{
    citecolor=blue,
    colorlinks=true,
    linkcolor=blue,
    filecolor=magenta,      
    urlcolor=blue,
}

 \newtheorem{proposition}{Proposition}
 \newtheorem{example}{Example}
 \newtheorem{remark}{Remark}
  
  \newtheorem{definition}{Definition}
  \newtheorem{lemma}{Lemma}

  \newtheorem{corollary}{Corollary}

    \DeclareMathOperator*{\argmax}{argmax}
        \DeclareMathOperator*{\argmin}{argmin}
        
          \DeclareMathOperator*{\card}{card}
           \DeclareMathOperator*{\supp}{supp}
        
        \newcounter{savefootnote}
\newcounter{symfootnote}
\newcommand{\symfootnote}[1]{%
   \setcounter{savefootnote}{\value{footnote}}%
   \setcounter{footnote}{\value{symfootnote}}%
   \ifnum\value{footnote}>8\setcounter{footnote}{0}\fi%
   \let\oldthefootnote=\thefootnote%
   \renewcommand{\thefootnote}{\fnsymbol{footnote}}%
   \footnote{#1}%
   \let\thefootnote=\oldthefootnote%
   \setcounter{symfootnote}{\value{footnote}}%
   \setcounter{footnote}{\value{savefootnote}}%
}

\usepackage[margin=1in]{geometry}

\usepackage{microtype}

\allowdisplaybreaks

\begin{document}
\begin{spacing}{1.23}

\begin{titlepage}
    \centering
    \vspace*{1cm}
    
    \Large
    \textbf{Decomposing Common Agency} \symfootnote{I am indebted to Laura Doval for her continued guidance and support. I am grateful to Navin Kartik, Andreas Kleiner, Zihao Li, Qingmin Liu, Alessandro Pavan, Jacopo Perego, Rui Tang, and Yangfan Zhou for insightful comments and suggestions. }
    
    \vspace{0.5cm}
    \large
    Zhiming Feng \symfootnote{Department of Economics, Columbia University. Email: \href{zf2295@columbia.edu}{zf2295@columbia.edu}.}

    \vspace{0.5cm}
    \large
    \today
    
    \vspace{1.5cm}

    \normalsize
    \textbf{Abstract}
    \vspace{0.5cm}

\justifying

This paper develops a decomposition methodology for common agency
games in which each principal's payoff depends on her own outcome
and the agent's type, but not on rivals' outcomes. The key step
reduces each principal's best-response problem to a standard
screening problem defined over the agent's indirect utility---the
upper envelope of her payoff over rivals' offerings. Individually
best-responding mechanisms then assemble into a pure-menu perfect
Bayesian equilibrium when a compatibility condition
(utility-preserving recombination, \eqref{2+}) ensures aligned tie-breaking
across principals. Under a non-indifference condition,
the decomposition recovers all equilibria except those sustained
by menu items that no type of the agent actually selects but which
nevertheless discipline the rival's screening problem.
When principals' payoffs depend on the full allocation profile,
the decomposition adapts only under substantive regularity
conditions on the agent's off-path choice behavior, one of which
coincides with Luce's choice axiom. I apply the methodology to
two settings. In a quadratic-loss delegation model, equilibria
feature one principal offering a finite menu of discrete
``regimes'' while the other receives piecewise full delegation
within each regime. In a competitive bundling duopoly under
intrinsic common agency, the decomposition yields equilibria
exhibiting market splitting, in which firms specialize
in complementary bundles, and asymmetric equilibria with a
take-it-or-leave-it base contract paired with a nested or tree
menu of upgrades.
\\ \\ 
 
    \vfill

\end{titlepage}
\pagenumbering{gobble}
{\small
\tableofcontents
}

\newpage
\pagenumbering{arabic}
\setcounter{page}{1}

\section{Introduction}

Common agency---a setting in which multiple principals simultaneously contract with a single agent---arises naturally in many economic environments. Competing firms offering products to a shared consumer, multiple regulators overseeing a single entity, and rival interest groups lobbying a common policymaker all share this structure. A foundational result in this literature, the \emph{menu theorem} (\cite{Peters2001CommonAgency}, \cite{ms2002}; see also \cite{EpsteinPeters1999}), establishes that any equilibrium outcome of a general mechanism game can be replicated by an equilibrium in which each principal simply offers a menu of outcomes from which the agent selects. This result provides a powerful simplification: it reduces the strategic complexity of competing mechanisms to the analysis of menu offers.

Yet knowing that menus suffice is not the same as knowing which menus arise in equilibrium. The central difficulty is strategic interdependence: each principal's optimal menu depends on the menus offered by rivals, because the agent's willingness to accept any particular contract is shaped by her alternatives elsewhere. This paper develops a method for managing this circularity by reducing
each principal's problem to a standard screening problem, transforming
the search for equilibrium from the space of menus to the more
tractable space of direct mechanisms. The fixed-point structure is not
eliminated---each principal's screening problem is defined relative to
the rivals' allocation rules---but the transformed system can be
attacked using standard mechanism design tools, and candidate
solutions can be verified through a single compatibility check. The key construction is the \textbf{indirect utility function}: given the menus offered by rival principals, the agent's indirect utility in principal $i$ is defined as the upper envelope of her utility over all possible selections from the rivals' menus. This single function absorbs the competitive complexity---how rivals' offerings shape the agent's preferences---into the agent's ``type-dependent utility'' over principal $i$'s outcome. Once the indirect utility is in hand, each principal faces what is effectively a textbook mechanism design problem: choose an allocation rule to maximize payoff subject to incentive compatibility constraints defined by the indirect utility. The indirect utility construction thus serves as a bridge: it allows the powerful tools developed for single-principal settings---the envelope theorem, Lagrangian methods for delegation, virtual-surplus techniques for bundling---to be applied directly to multi-principal environments.

The decomposition works most cleanly when each principal's payoff
depends only on her own outcome and the agent's type, with no direct
dependence on rivals' outcomes. In this case---which I treat as the
core setting of the paper---the procedure is as follows. First, each principal solves her single-principal screening problem (program \eqref{P3}), taking the rivals' allocation rules as given; the solution pins down her best-response menu as the range of the optimal direct mechanism. A key step here is Lemma \ref{prop:separate}, which shows that optimality among direct mechanisms implies optimality among \emph{all} mechanisms---so no separate verification against indirect mechanisms is needed. Second, the analyst checks a compatibility condition I call \textbf{utility-preserving recombination} (\eqref{2+}). This condition addresses a subtlety absent in single-principal settings: when the agent faces the full menu profile, she may encounter ties---multiple outcome combinations that are equally attractive. In a single-principal world, breaking ties in favor of the principal is a standard assumption; with multiple principals, the tie-breaking rule must simultaneously favor \emph{all} principals. Condition \eqref{2+} ensures that such aligned tie-breaking is feasible. I show through an explicit example (Example \ref{e1+}) that this condition is not vacuous: there exist profiles of individually optimal allocation rules for which no agent strategy supports the profile as an equilibrium, precisely because the principals' interests at ties are irreconcilable.

I provide several sufficient conditions under which \eqref{2+} is automatically satisfied, so that any profile of solutions to \eqref{P3} directly yields a PBE. The sufficient conditions include a non-indifference condition (Corollary \ref{c3}): the agent is never indifferent between distinct outcome profiles, weak separability of the agent's utility across principals' outcomes in the sense of \cite{AttarMajumdarPiaserPorteriro2008} (Corollary \ref{c0}), and the structural condition that all but one principal offer singleton menus (Corollary \ref{c2}). Non-indifference plays an additional role in sharpening the converse characterization (Proposition \ref{prop:p3induced}) and the Pareto optimality results.

I also characterize the converse: which equilibria can be recovered through the decomposition. Under the non-indifference condition, Proposition \ref{prop:p3induced} shows that a PBE is \eqref{P3}-induced if and only if every item in every principal's menu is chosen by some type on the equilibrium path. In other words, the decomposition captures all equilibria except those with ``unused'' menu items. Such items play no role on the equilibrium path but can affect the rival's screening problem through the indirect utility---functioning as \emph{strategic shields} that deter profitable deviations. The decomposition, by working with minimal menus (the ranges of direct mechanisms), cannot capture equilibria sustained by such shields.

Beyond existence and characterization, I study the welfare properties of decomposition-induced equilibria. A natural concern is whether the decomposition, by restricting attention to minimal menus, might miss welfare-superior outcomes. Following the characterization result above, under non-indifference, any Pareto improvement over a PBE that is Pareto optimal among \eqref{P3}-induced PBE must come from equilibria with strategic shields---but the economic scope for such improvements is limited, as unused menu items can only tighten the rival's incentive constraints. Beyond non-indifference, I show that this concern can also be mitigated in several cases: when the type space is a singleton, and when one principal has a type-independent ideal point (Proposition \ref{prop:PO}). 

When each principal's payoff depends on the full allocation profile
rather than only on her own outcome, the agent's selection rule
enters each principal's objective, and the decomposition no longer
reduces each principal's problem to a self-contained screening
problem. I identify regularity conditions on the agent's choice
behavior---(stochastic) independence-of-irrelevant-alternatives
conditions, one of which coincides with \cite{luce1959}'s choice
axiom---under which a modified construction delivers best-response
characterizations. These conditions are substantive restrictions
on off-path choice rather than consequences of rational behavior,
so the corresponding results are more conditional than the
independent-payoff analysis.

Applying the decomposition to continuous-type settings requires the envelope theorem to hold for the indirect utility function, which is generically nondifferentiable. In Online Appendix \ref{s21}, I establish conditions under which the envelope formula remains valid for such functions and verify that these conditions are satisfied in both applications.

I demonstrate the framework through two applications. The
\textbf{first application is to delegation} with quadratic-loss
preferences and two principals. Each principal has a type-dependent
ideal point, and the agent's bliss point is the sum of the two
outcomes. The central result characterizes equilibria in which one
principal offers a finite menu of discrete ``regimes'' while the
other receives full delegation within each regime (Proposition
\ref{prop:delegationhard}). In the single-principal benchmark, full delegation arises under
alignment of preferences, and interval delegation arises more generally. With two competing principals, Proposition \ref{prop:delegationhard} shows that full delegation can be sustained
within segments of a partition of the type space, with one principal
selecting the segment through a coarse menu and the other adapting
freely within the selected segment. The structure
arises naturally in hierarchical economic environments: a franchiser
choosing among a small number of store formats while franchisees adapt
pricing and product mix locally; a telecommunications carrier
selecting a service tier while a device manufacturer tailors handset
configurations within each tier; a federal regulator setting a
regulatory regime while state agencies exercise local discretion. In
each case, one principal provides the broad framework and the other
provides the granular adaptation, and this division of screening
responsibilities is exactly what the decomposition produces. I also
identify conditions under which ``no compromise for both'' emerges,
when the principals' ideal points perfectly partition the agent's
bliss point (Proposition \ref{prop:delegationeasy}).

A \textbf{second application is to competitive bundling}, where I
extend the one-dimensional, non-additive-value framework of
\cite{ghili2023characterization}, \cite{yang2023nested}, and
\cite{feng2025} to a duopoly setting under intrinsic common agency.
The decomposition yields two classes of equilibria: one with
\textbf{market splitting}, where firms partition the
product space into complementary bundles and charge equal markups
tied to the marginal type's virtual surplus (Proposition \ref{t1});
and one with a \textbf{``base product plus complementary upgrades''}
structure, where one firm posts a singleton take-it-or-leave-it
contract while the rival screens consumer types through a menu of
upgrades, potentially with the nested or tree structure of
\cite{feng2025} (Proposition \ref{t2} and Corollary \ref{c1}).

\paragraph{Related literature.} This paper contributes to the theory of common agency and competing mechanisms. The closest predecessor is \cite{MartimortStole2001Discrete}, who
study intrinsic common agency with discrete types. They
observe that each principal's best-response problem can be analyzed
within the class of direct revelation mechanisms by the revelation
principle applied to the sub-problem. Lemma \ref{prop:separate}
formalizes this observation and makes it the foundation of a
decomposition framework: combined with the indirect utility function, each principal's sub-problem becomes a standard screening problem, and the compatibility condition
\eqref{2+} identifies when individually best-responding mechanisms
assemble into a PBE. My framework also accommodates continuous type
spaces, continuous outcome spaces, general utility structures, and
mixed agent strategies, whereas \cite{MartimortStole2001Discrete}
work with discrete types, pure agent strategies, and two principals
controlling discrete activities; their invocation of the revelation
principle would not extend directly to the mixed-strategy case, since
under mixed agent selection from the rival's menu, some payoffs
achievable with indirect mechanisms cannot be replicated by direct
mechanisms. In addition, \cite{MartimortStole2001Discrete}'s focus is on the
structure of the equilibrium set under complementarity and
substitutability conditions on the agent's objective; mine is on the
decomposition as a general-purpose method, with \eqref{2+} addressing
the off-path behavior that their equilibrium concept does not treat.

\cite{ms2009} also work closely with direct revelation mechanisms,
characterizing equilibria in delegated and intrinsic common agency
under the multi-principal \cite{mussa1978monopoly} model. Their
equilibrium notion specifies each principal's best response given the
\emph{equilibrium} menus of rivals, without requiring the agent's
strategy to be defined for arbitrary off-path menu profiles. The
perfect Bayesian equilibrium I study here requires the agent's
strategy to be specified for \emph{all} menu profiles, including
deviations---the stronger requirement that gives rise to the
tie-breaking complications captured by \eqref{2+}. A more recent
\cite{MartimortSemenovStole2018} contribution, also in the
\cite{mussa1978monopoly} setting, characterizes the \emph{complete}
set of equilibrium allocations via ``self-generating optimization
programs.'' Their approach and mine both reduce the multi-principal
problem to optimization programs, but with different objectives:
their programs characterize the allocation set, while mine construct
equilibrium menus and strategies directly. My framework also applies beyond the \cite{mussa1978monopoly} specification to settings with non-transferable utility and arbitrary outcome spaces, enabling the delegation and bundling applications.

The cooperative approach to common agency, initiated by \cite{bw1986}
and extended by \cite{LausselLeBreton2001}, characterizes equilibrium
\emph{payoffs} through connections to the core and truthful
equilibria. \cite{MartimortStole2024Menu} recently extend this
menu-auction framework to asymmetric information, using non-smooth
optimal control to aggregate principals' objectives into a surrogate
program via the Aggregate Concurrence Principle
(\cite{MartimortStole2012}).\footnote{Their setting differs from mine
in that principals offer contribution schedules over a common action
in a quasi-linear environment, while my principals control separate
outcomes under general utility; correspondingly, their reduction is
aggregative while mine is separative.} Rather than characterizing the payoff set, I construct the equilibrium \emph{menus} and \emph{strategies} that implement specific outcomes. The menu theorem (\cite{Peters2001CommonAgency}, \cite{ms2002}) and the revelation principle for competing mechanisms (\cite{EpsteinPeters1999}) provide the conceptual foundation for working with menus, but do not address which menus arise in equilibrium. \cite{Peters2003Negotiation} studies existence when principals are restricted to take-it-or-leave-it offers under a no-externalities assumption; Corollary \ref{c2} generalizes this by allowing one principal to screen while maintaining equilibrium existence. \cite{PavanCalzolari2009Sequential, PavanCalzolari2010} introduce extended direct mechanisms and analyze sequential common agency; I discuss the relationship between their approach and mine in Section \ref{sec:extendeddm}, showing that extended mechanisms offer advantages primarily for replicating, rather than generating, equilibria using the decomposition (Proposition \ref{prop:extendeddm2}). More recently, \cite{attar2025keeping} show that principals can improve their payoffs by privately and asymmetrically informing agents about their decision rules, opening a direction I discuss as a natural extension of the present framework.

The delegation application builds on \cite{holmstrom1984} and the Lagrangian methods of \cite{AmadorBagwell2013OptimalDelegation} and \cite{KartikKleinerVanWeelden2021Delegation}, which I adapt to the multi-principal setting. The bundling application extends the framework of \cite{ghili2023characterization}, \cite{yang2023nested}, and \cite{feng2025} from monopoly to duopoly.

~\\
The remainder of the paper is organized as follows. Section \ref{setup} introduces the model and defines equilibrium. Section \ref{sec:mainresults} presents the decomposition results, including the construction and characterization of PBE, Pareto optimality, and the extension to outside options. Sections \ref{sdelegation} and \ref{sec:multiproductduopoly} present the delegation and bundling applications. Section \ref{sec:discussions} discusses extensions to general payoff interdependence and extended direct mechanisms. Section \ref{sec:conclusion} concludes. Proofs are collected in the Appendix, and the envelope formula for indirect utilities is developed in Online Appendix \ref{s21}.

\section{Setup}
\label{setup}

Consider a setting with a single agent $A$ and a set of principals $\mathcal{N} = \{1,...,n\}$. Each principal $i \in \mathcal{N}$ contracts with the agent via a mechanism $(M_i, g_i)$, where $M_i$ denotes the message space and $g_i: M_i \rightarrow \Delta(O_i)$ is the outcome function. The set of feasible outcomes for principal $i$, $O_i$, is assumed to be compact.

The agent's type $t$ is private information, distributed over a compact set $T$. The agent's utility function, $V: \Delta(\prod_{i=1}^n O_i) \times T \rightarrow \mathbb{R}$, is continuously differentiable with respect to $t$ and satisfies the expected utility hypothesis: for any $t \in T$ and $\tilde{o} \in \Delta(\prod_{i=1}^n O_i)$, $V(\tilde{o},t)=\mathbb{E}_{o \sim \tilde{o}}[V(o,t)]$. 

For any $i$, let principal $i$'s utility function be $u_i: \Delta(O_i) \times T \rightarrow \mathbb{R}$. That is, each principal cares about her own outcome and the agent's type. The expected utility hypothesis is also satisfied: for any $t \in T$ and $\tilde{o} \in \Delta(O_i), u_i(\tilde{o},t)=\mathbb{E}_{o \sim \tilde{o}}[u_i(o,t)]$. I extend the setting to general payoff interdependence, where each principal cares about her own outcome, the agent's type, and other principals' outcomes: $u_i: \Delta(\prod_{j=1}^n O_j) \times T \rightarrow \mathbb{R}$, in Section \ref{sec:discussions}. 

I initially analyze the setting in which the agent has no outside options. Outside options are introduced in Section \ref{section:independent}, which focuses on the case where principals have independent payoff structures, where $u_i$ depends solely on outcomes in $O_i$.

The timing of the game is as follows:
\begin{itemize}
    \item At $t=0$, $A$ learns her type.
    \item At $t=1$, principals simultaneously and independently offer mechanisms to $A$.
    \item At $t=2$, after observing the set of offered mechanisms, $A$ simultaneously reports a message $m_i \in M_i$ to each principal $i$.
\end{itemize}

The solution concept that I study for this mechanism game is perfect Bayesian Equilibrium (PBE). A well-known result in the literature, which is often referred to as the menu theorem or the delegation principle (\cite{Peters2001CommonAgency}, \cite{ms2002}), says the following: any PBE outcome in the mechanism game can be replicated by a PBE outcome in the game where for any $i \in \mathcal{N}$, principal $i$ offers menus $\mathscr{M}_i \in 2^{\Delta(O_i)}$ from which the agent chooses. 

Accordingly, it suffices to analyze the \textit{menu game}.\footnote{For expositional clarity, I restrict attention to menus of deterministic outcomes rather than stochastic ones. This simplification is standard in the literature; see, e.g., \cite{PavanCalzolari2010}. This restriction can be relaxed: see Remark \ref{rmk2} following Proposition \ref{p0}.} In the menu game, for any $i \in \mathcal{N}$, a strategy for principal $i$ is a distribution over menus, $\sigma_{i} \in \Delta(2^{O_i})$. The agent's strategy maps the profile of offered menus and her type to a distribution over outcomes: $\sigma_A: \prod_{i=1}^n 2^{O_i} \times T \rightarrow \Delta( \prod_{i=1}^n O_i)$. The formal definition of PBE in this context is as follows:

\begin{definition}
\label{d0}
A strategy profile $(\sigma_{1},...,\sigma_{n}, \sigma_A)$ is a PBE if:
\begin{enumerate}
    \item The agent optimizes given the menu profile:
    \[
    \text{ for any } (\mathscr{M}_1,...,\mathscr{M}_n) \in \prod_{i=1}^n 2^{O_i}, t \in T, \quad
    (o_1,...,o_n) \in \supp(\sigma_A(\mathscr{M}_1,...,\mathscr{M}_n,t)) 
    \]
    \[
    \implies (o_1,...,o_n) \in \argmax_{\tilde{o}_i \in \mathscr{M}_i, \forall i} V(\tilde{o}_1,...,\tilde{o}_n,t).
    \] \label{cd1}
    \item Each principal optimizes given the strategies of others:
    \[
    \text{ for any } i, \quad \mathscr{M}_i \in \supp(\sigma_{i}) \implies \mathscr{M}_i \in \argmax_{\tilde{\mathscr{M}}_i \in 2^{O_i}} \mathbb{E}_t\left[\int_{\mathscr{M}_{-i}} u_i(\sigma_A(\mathscr{M}_1',...,\tilde{\mathscr{M}}_i,...,\mathscr{M}_n',t),t) \prod_{j \neq i} \sigma_{j}(\mathbf{d}\mathscr{M}_j')\right].
    \] \label{cd2}
\end{enumerate}
\end{definition}

Condition \ref{cd1} requires that the agent chooses an outcome optimal for her type given any menu profile. Condition \ref{cd2} requires that each principal chooses a menu strategy that maximizes expected utility, taking into account the strategies of other principals and the agent's best response.

The principals' strategy sets in Definition \ref{d0} include mixed strategies over menus. Throughout, I mainly focus on equilibria in which each principal plays a pure menu. When checking deviations, I follow Definition \ref{d0} and consider all mixed strategies over menus; this is essentially equivalent to checking only pure-menu deviations, since the
linearity of each principal's payoff in her own mixed strategy means no mixed deviation can strictly improve on every pure menu in its support.

\section{Decomposition to Single-Principal Problems}
\label{sec:mainresults}

The primary objective of this paper is to decompose the ``$n$-principals-one-agent'' problem into $n$ distinct single-principal problems. The decomposition proceeds in two logically distinct steps, each addressing a separate analytical challenge:

\begin{enumerate}
    \item Step 1 reduces the problem of choosing an optimal menu to the problem of choosing an optimal direct mechanism. This step eliminates the need to search over all possible menus and ensures that optimality among direct mechanisms implies optimality among \emph{all} mechanisms (Lemma \ref{prop:separate}).
    \item Step 2 assembles individually optimal direct mechanisms into a single equilibrium. This step requires a compatibility condition to ensure that the agent strategies supporting each principal's optimality can be unified into a common strategy (Proposition \ref{p1}).
\end{enumerate}

The practical value of this decomposition is that it reduces a game-theoretic problem---finding a fixed point of mutual best responses over the space of menus---to a collection of optimization problems that an analyst can solve using standard mechanism design tools. The fixed-point structure is not eliminated, but it is \emph{relocated}. 

While the agent's utility in the original game is standard (given by $V$), the construction of his utility in the decomposed setting warrants further elaboration. The following definition formally introduces the \textbf{indirect utility function} used in the single-principal analysis:

\begin{definition}
\label{d1}
    Given menu profile $(\mathscr{M}_1,...,\mathscr{M}_n)$, the agent's \textbf{indirect utility function} in principal $i$ is $v_i(\cdot,\cdot|\mathscr{M}_{-i}): \Delta(O_i) \times T \rightarrow \mathbb{R}$ such that for any $t$
    \begin{equation}
    \label{1}
        v_i(\cdot,t|\mathscr{M}_{-i})=\max \limits_{o_{-i} \in \mathscr{M}_{-i}} V(\cdot,o_{-i},t).\footnote{Throughout the paper, whenever I use ``$\max$'', I implicitly assume that it is well-defined. In fact, as will be shown in Online Appendix \ref{s21} when discussing the envelope formula for the indirect utility function and in Section \ref{sec:multiproductduopoly} when discussing competitive bundling, I mainly focus on finite menus, which makes ``$\max$'' well-defined.}
    \end{equation}
    Moreover, for a stochastic menu profile $\tilde{\mathscr{M}}_{-i} \in \Delta(\mathscr{M}_{-i}), v_i(\cdot,t|\tilde{\mathscr{M}}_{-i})= \displaystyle{\int}_{\mathscr{M}_{-i}} \max \limits_{o_{-i} \in \mathscr{M}_{-i}} V(\cdot,o_{-i},t) \tilde{\mathscr{M}}_{-i}(\mathbf{d}\mathscr{M}_{-i})$. 
\end{definition}

The indirect utility function captures how the competitive environment shapes the agent's preferences over any single principal's offerings. For a fixed outcome $o_i$ from principal $i$, the agent optimally selects among other principals' menus, and the resulting payoff $v_i(o_i, t|\mathscr{M}_{-i})$ is the relevant ``utility'' that governs the agent's choice from principal $i$'s menu. An equivalent interpretation: for any $o_i$, the function $v_i(o_i, \cdot|\mathscr{M}_{-i})$ is the \textbf{upper envelope} of $\{V(o_i, o_{-i}, \cdot)\}_{o_{-i} \in \mathscr{M}_{-i}}$. This upper-envelope structure is what connects the indirect utility to the technical results in Online Appendix \ref{s21} and enables the application of envelope methods in concrete settings.

\subsection{Constructing PBE}

This subsection focuses on constructing PBE using the decomposition method. 

\medskip\noindent\textbf{Step 1: From menus to direct mechanisms.} The first result establishes that, for each principal individually, it suffices to optimize over direct mechanisms rather than over all possible menus. This is a single-principal result---it holds for any fixed rival menu, regardless of whether that menu arises in equilibrium.

\begin{lemma}
    \label{prop:separate}
    For any $i$ and $\mathscr{M}_{-i}$, if $\phi_i^*$ solves
        \begin{equation}
\label{P3separate}
\tag{P1-$\mathscr{M}_{-i}$}
\begin{split}
    \max_{\phi_i: T \rightarrow O_i}\;& \mathbb{E}_t[u_i(\phi_i(t),t)], \\
    \text{s.t. }\; 
    v_i(\phi_i(t),t \mid \mathscr{M}_{-i}) 
    &\ge 
    v_i(\phi_i(t'),t \mid \mathscr{M}_{-i}),\quad \forall t,t',
\end{split}
\end{equation}
then there exists $\sigma_A$ that satisfies condition \ref{cd1} in Definition \ref{cd1} such that 
\[ \{\phi_i^*(s)\}_{s \in T} \in \argmax \limits_{\tilde{\mathscr{M}}_i \in 2^{O_i}} \mathbb{E}_t[u_i(\sigma_A(\tilde{\mathscr{M}}_i, \mathscr{M}_{-i},t),t)].\]
\end{lemma}

Lemma \ref{prop:separate} is the engine of the decomposition. Its content is a
single-principal statement: fixing the rival menu $\mathscr{M}_{-i}$, principal
$i$'s problem of choosing an optimal \emph{menu}---an arbitrary subset of
$O_i$---reduces to choosing an optimal \emph{direct mechanism}---a function from
types to outcomes. That is, the search over $2^{O_i}$, which is essentially a search over all indirect mechanisms, collapses to a search over $O_i^T$, with the optimal menu recovered
simply as the range of the optimal direct mechanism. This reduction is what makes
the ``decomposition'' operational: each principal can solve a standard screening
problem, and the solution automatically pins down her best-response menu.

The logic is as follows. Suppose $\phi_i^*$ solves \eqref{P3separate}, and consider an arbitrary deviation menu $\tilde{\mathscr{M}}_i$. Given $\tilde{\mathscr{M}}_i$ and $\mathscr{M}_{-i}$, any agent strategy satisfying condition \ref{cd1} in Definition \ref{d0} induces some type-contingent outcome for principal $i$, which is itself a feasible direct mechanism under \eqref{P3separate}. But $\phi_i^*$ is already optimal among all feasible direct mechanisms, so the deviation cannot improve principal $i$'s payoff. In short, optimality over direct mechanisms implies optimality over all mechanisms.

\medskip\noindent\textbf{Step 2: From individual optimality to equilibrium.} The second result assembles individually optimal mechanisms into a PBE. This is where the genuinely multi-principal challenge arises.

\begin{proposition}
\label{p1}
For all $(\phi_1^*,...,\phi_n^*)$ such that for any $i, \phi_i^*$ solves
    \begin{equation}
\label{P3}
\tag{P1-$\phi_{-i}^*$}
\begin{split}
    \max_{\phi_i: T \rightarrow O_i}\;& \mathbb{E}_t[u_i(\phi_i(t),t)], \\
    \text{s.t. }\; 
    v_i(\phi_i(t),t \mid \prod_{j \neq i} \{\phi_j^*(s)\}_{s \in T}) 
    &\ge 
    v_i(\phi_i(t'),t \mid \prod_{j \neq i} \{\phi_j^*(s)\}_{s \in T}),\quad \forall t,t',
\end{split}
\end{equation}
    and  
\begin{equation} \label{2+}
\begin{aligned}
\text{for any } t,\ \text{there exists } 
(o_1^t, \dots, o_n^t) 
&\in \argmax_{\substack{
o_i \in \{\phi_i^{*}(s)\}_{s \in T},\forall i
}} V((o_1,\dots,o_n), t) \\
\text{such that for any } i,\quad
u_i(o_i^t, t) 
&= u_i(\phi_i^{*}(t), t),
\end{aligned}
\tag{UPR}
\end{equation}
there exists $\sigma_A$ such that $(\mathscr{M}_1,...,\mathscr{M}_n,\sigma_A)$ is a PBE, in which for any $i, \mathscr{M}_i=\{\phi_i^{*}(s)\}_{s \in T}$.
\end{proposition}

``UPR'' stands for utility-preserving recombination, which captures the essence of the condition: among the agent's utility-maximizing outcome profiles, there exists one that preserves every principal's utility relative to the intended allocation. Note that Lemma \ref{prop:separate} implies that if each $\phi_i^*$ solves \eqref{P3}, then for each principal $i$ there exists some $\sigma_A^i$ satisfying condition \ref{cd1} in Definition \ref{d0} under which $\{\phi_i^*(s)\}_{s \in T}$ is a best response to $\prod_{j \neq i} \{\phi_j^*(s)\}_{s \in T}$. However, different principals may require different agent strategies to sustain their respective best responses. The role of \eqref{2+} is to guarantee that a single agent strategy works for all principals simultaneously---that is, the selections $\sigma_A^1, \ldots, \sigma_A^n$ can be unified into a common $\sigma_A$, thereby ensuring mutual best response.

\medskip\noindent\textbf{Why is \eqref{2+} needed?} In a single-principal setting, if the agent is indifferent between two outcomes, the principal can simply specify that ties are broken in her favor. With multiple principals, however, breaking ties in favor of one principal may come at the expense of another. The following example demonstrates that this tension is genuine: there exist profiles of individually optimal mechanisms for which \emph{no} agent strategy supports a PBE, because the principals' interests at ties are irreconcilable.

\begin{example}
\label{e1+}
Consider an example with two principals $\{1,2\}$ and one agent. Let $O_1=\{a,a'\}, O_2=\{b,b'\}, T=\{t_1,t_2\}$. The principals and the agent's payoffs are given in the tables below, where the first (second) entry is principal 1 (2)'s payoff, and the third entry is the agent's payoff. It is straightforward that fix $t$, $u_i$ only depends on $o_i,$ for any $i \in \{1,2\}$. 

\begin{minipage}{0.45\textwidth}
\centering
\[
\begin{array}{c|cc}
 & b & b' \\ \hline
a & (5,5,5) & (5,0,10) \\
a' & (0,5,10) & (0,0,0) \\
 & t_1  &
\end{array}
\]
\end{minipage}
\quad
\begin{minipage}{0.45\textwidth}
\centering
\[
\begin{array}{c|cc}
 & b & b' \\ \hline
a & (0,0,0) & (0,10,0) \\
a' & (10,0,0) & (10,10,10) \\
 & t_2 &
\end{array}
\]
\end{minipage}

Claim that $\phi_1^*(t_1)=a, \phi_1^*(t_2)=a';\phi_2^*(t_1)=b, \phi_2^*(t_2)=b'$ solves \eqref{P3} for any distribution on $T$. I first check the feasibility. IC for $t_1$ is guaranteed by
\begin{equation*}
    \begin{split}
        v_1(a,t_1|\{b,b'\})=10 \geq 10=v_1(a',t_1|\{b,b'\}), \\
        v_2(b,t_1|\{a,a'\})=10 \geq 10=v_2(b',t_1|\{a,a'\}),
    \end{split}
\end{equation*}
and IC for $t_2$ is guaranteed by
\begin{equation*}
    \begin{split}
        v_1(a',t_2|\{b,b'\})=10 \geq 0=v_1(a,t_2|\{b,b'\}), \\
        v_2(b',t_2|\{a,a'\})=10 \geq 0=v_2(b,t_2|\{a,a'\}).
    \end{split}
\end{equation*}

The optimality of $(\phi_1^*,\phi_2^*)$ can be argued through the following: under $t_1$, $\{a\}$ and $\{b\}$ strictly dominates $\{a'\}$ and $\{b'\}$ respectively, and under $t_2$, $\{a'\}$ and $\{b'\}$ strictly dominates $\{a\}$ and $\{b\}$ respectively. 

However, $(\phi_1^*, \phi_2^*)$ does not satisfy \eqref{2+}. In particular, $\argmax \limits_{o_i \in \{\phi_i^*(s)\}_{s \in T}, \forall i}V((o_1,o_2),t_1)=\{(a,b'), (a',b)\}$, but $u_2((a,b'),t_1)<u_2((a,b),t_1); u_1((a',b),t_1)<u_1((a,b),t_1)$. 

The economic content of this example is transparent: at type $t_1$, the agent prefers complementary ``cross-pairings''---$a$ with $b'$, or $a'$ with $b$---over the ``matching'' pair $(a,b)$ that both principals prefer. Each cross-pairing benefits one principal at the expense of the other. Since neither principal is willing to sacrifice her payoff, no tie-breaking rule can simultaneously satisfy both, and no PBE exists when $t_1$ is sufficiently likely (specifically, when $Pr(t=t_1) > \frac{4}{5}$).

I proceed to show this formally. Given principal 2 plays $\{b,b'\}$, for any $\sigma_A$ that satisfies condition \ref{cd1} in Definition \ref{d0},
\begin{equation*}
\begin{aligned}
u_1(\sigma_A(\{a\}, \{b,b'\}, t_1)) &= 5,  &\qquad
u_1(\sigma_A(\{a\}, \{b,b'\}, t_2)) &= 0, \\
u_1(\sigma_A(\{a'\}, \{b,b'\}, t_1)) &= 0, & 
u_1(\sigma_A(\{a'\}, \{b,b'\}, t_2)) &= 10, \\
u_1(\sigma_A(\{a,a'\}, \{b,b'\}, t_1)) &= 5p_A, &
u_1(\sigma_A(\{a,a'\}, \{b,b'\}, t_2)) &= 10,
\end{aligned}
\end{equation*}
where $p_A$ is given by $\sigma_A(\{a,a'\}, \{b,b'\},t_1):=p_A(a,b')+(1-p_A)(a',b)$. Similarly, given principal 1 plays $\{a,a'\}$, 
\begin{equation*}
\begin{aligned}
u_2(\sigma_A(\{a,a'\}, \{b\}, t_1)) &= 5,  &\qquad
u_2(\sigma_A(\{a,a'\}, \{b\}, t_2)) &= 0, \\
u_2(\sigma_A(\{a,a'\}, \{b'\}, t_1)) &= 0, &
u_2(\sigma_A(\{a,a'\}, \{b'\}, t_2)) &= 10, \\
u_2(\sigma_A(\{a,a'\}, \{b,b'\}, t_1)) &= 5 - 5p_A, &
u_2(\sigma_A(\{a,a'\}, \{b,b'\}, t_2)) &= 10.
\end{aligned}
\end{equation*}

Let $Pr(t=t_1)=p$, then if there exists $\sigma_A$ such that $(\{a,a'\}, \{b,b'\}, \sigma_A)$ is a PBE, then there exists $p_A$ such that
\begin{equation*}
    \begin{split}
        p\cdot 5p_A+10(1-p) \geq 5p, \\
        p \cdot 5(1-p_A)+10(1-p) \geq 5p,
    \end{split}
\end{equation*}
which implies that $\frac{3p-2}{p} \leq p_A \leq \frac{2-2p}{p}$. In that case, for $p$ such that $2-2p<3p-2$, i.e., $p>\frac{4}{5}$, there does not exist any $p_A$ that makes the inequality hold, thus there does not exist $\sigma_A$ such that $(\{a,a'\},\{b,b'\}, \sigma_A)$ is a PBE. 

\end{example}

\begin{remark}
\label{rmk:upnr}
The following claim is not correct: 
\label{prop:uprsimplified}
    there exists $(\phi_1^*,...,\phi_n^*)$ such that for any $i, \phi_i^*$ solves \eqref{P3} and \eqref{2+} is satisfied only if there exists $(\phi_1^{*},...,\phi_n^{*})$ such that for any $i, \phi_i^*$ solves \eqref{P3}, and for any $t$
    \begin{equation*} (\phi_1^{*}(t),...,\phi_n^{*}(t)) \in \argmax \limits_{o_i \in \{\phi_i^{*}(s)\}_{s \in T}, \forall i} V((o_1,...,o_n),t). \end{equation*}
    In particular, $(\phi_1^{**},...,\phi_n^{**})$, where for any $i,t,$
    \[ \phi_i^{**}(t)=o_i^t,\]
    where $o_i^t$ is specified in \eqref{2+}, does not necessarily satisfy for any $i, \phi_i^{**}$ solves \eqref{P3}.\footnote{Rigorously, \eqref{P3} written with respect to $\phi_{-i}^{**}$. } Examples can be found in Online Appendix \ref{sec:examplesrmk1}. 
\end{remark}

\begin{remark}
    The following claim is not correct: given $(\phi_1^*,...,\phi_n^*)$ that solves \eqref{P3} and satisfies \eqref{2+}, for any $(\mathscr{M}_1,...,\mathscr{M}_n)$ such that for any $i, \{\phi_i^*(s)\}_{s \in T} \subseteq \mathscr{M}_i$, there exists $\sigma_A$ such that $(\mathscr{M}_1,...,\mathscr{M}_n, \sigma_A)$ is a PBE. This is because if only $\{\phi_i^*(s)\}_{s \in T} \subseteq \mathscr{M}_i$ is satisfied, then the agent may choose outcomes that are included in $\mathscr{M}_i$ but are excluded from $\{\phi_i^*(s)\}_{s \in T}$ when facing $\mathscr{M}_i$. Example \ref{e1+} can also be used to invalidate this claim: $(\{a,a'\}, \{b,b'\})$ clearly consists of a $(\mathscr{M}_1,\mathscr{M}_2)$ that satisfies for any $i, \{\phi_i^*(s)\}_{s \in T} \subseteq \mathscr{M}_i$ for any $(\phi_1^*, \phi_2^*)$ that solves \eqref{P3} and satisfies \eqref{2+}, but when $Pr(t=t_1)>\frac{4}{5}$, there does not exist $\sigma_A$ such that $(\{a,a'\},\{b,b'\}, \sigma_A)$ is a PBE. 
\end{remark}

\medskip\noindent\textbf{When is \eqref{2+} satisfied?} The following corollaries provide sufficient conditions under which \eqref{2+} holds automatically, so that any profile of individually optimal mechanisms constitutes a PBE without further verification. Each corollary addresses a different source of the guarantee: the first two restrict the agent's utility function $V$ (eliminating ties or ensuring that the agent's ranking of each principal's outcomes does not depend on rivals' outcomes), while the third restricts the equilibrium structure (at most one principal screens). 

\begin{corollary}
    \label{c3}
    For any $(\phi_1^*,...,\phi_n^*)$ such that for any $i, \phi_i^*$ solves \eqref{P3}, if for any $t$ and for any $o,o' \in \prod \limits_{i=1}^n \{\phi_i^*(s)\}_{s \in T}, $ \[o \neq o' \Rightarrow V(o,t) \neq V(o',t),\] then there exists $\sigma_A$ such that $(\mathscr{M}_1,...,\mathscr{M}_n,\sigma_A)$ is a PBE, in which for any $i, \mathscr{M}_i=\{\phi_i^{*}(s)\}_{s \in T}$.
\end{corollary}

The condition in Corollary \ref{c3} is $(\phi_1^*,...,\phi_n^*)$-specific. A more general and stronger version of the condition is the following: for any $t \in T, o,o' \in \prod \limits_{i=1}^n O_i,$
\begin{equation} \label{eq:noindifference} \tag{non-indifference} o \neq o' \Rightarrow V(o,t) \neq V(o',t). \end{equation}

\eqref{eq:noindifference} says that the agent is never indifferent between distinct outcome profiles. The condition is arguably not demanding in settings where both the type space and the outcome space are finite. In continuous-type settings, however, non-indifference may fail ---for example, in the delegation application (Section \ref{sdelegation}), the agent's utility depends only on the sum of the principals' outcomes, so profiles with the same sum are tied. In such cases, \eqref{2+} is verified directly or alternatively through the following corollaries. 

When \eqref{eq:noindifference} holds, \eqref{2+} is naturally satisfied (there are no ties to break), and moreover any PBE must involve a pure agent strategy. To see the latter, note that if $o \neq o'$ both belong to $\supp \sigma_A(\mathscr{M}_1,...,\mathscr{M}_n,t)$, then agent optimality requires $V(o,t)=V(o',t)$, contradicting \eqref{eq:noindifference}.

\begin{corollary}
\label{c0}
    If $V$ is additive separable, i.e., for any $i \in \mathcal{N},$ there exists $ v^i: O_i \times T \rightarrow \mathbb{R}$ such that
    \begin{equation*}
        \text{ for any } (o_1,...,o_n) \in \prod \limits_{j=1}^n O_j, t \in T, V((o_1,...,o_n),t)=\sum \limits_{i=1}^n v^i(o_i,t),
    \end{equation*}
    then for any $(\phi_1^*,...,\phi_n^*)$ such that for any $i, \phi_i^*$ solves \eqref{P3}, there exists $\sigma_A$ such that $(\mathscr{M}_1,...,\mathscr{M}_n,\sigma_A)$ is a PBE, in which for any $i, \mathscr{M}_i=\{\phi_i^{*}(s)\}_{s \in T}$.

    In addition, the result can be strengthen to the case where $V$ is weakly separable \`{a} la \cite{AttarMajumdarPiaserPorteriro2008}. Formally, $V$ is weakly separable if for any $i,t, o_i,o_i',o_{-i},o_{-i}',$
\[
V(o_i,o_{-i},t)
> V(o_i',o_{-i},t)
\Rightarrow
V(o_i,o_{-i}',t)
> V(o_i',o_{-i}',t).\footnote{Weak separability as defined implies the weak-inequality
version: $V(o_i,o_{-i},t) \geq V(o_i',o_{-i},t)$ implies
$V(o_i,o_{-i}',t) \geq V(o_i',o_{-i}',t)$. Suppose not: there exist
$i, t, o_i, o_i', o_{-i}, o_{-i}'$ such that
$V(o_i,o_{-i},t) \geq V(o_i',o_{-i},t)$ but
$V(o_i,o_{-i}',t) < V(o_i',o_{-i}',t)$. Applying the weak-separability
condition with $o_i'$ and $o_i$ swapped yields
$V(o_i',o_{-i}',t) > V(o_i,o_{-i}',t) \Rightarrow V(o_i',o_{-i},t) >
V(o_i,o_{-i},t)$, contradicting the first inequality.}
\]
\end{corollary}

\begin{corollary}
\label{c2}
    For any $(\phi_1^*,...,\phi_n^*)$ such that for any $i, \phi_i^*$ solves \eqref{P3}, if there exists $\mathcal{N}' \subseteq \mathcal{N}$ such that 
    \begin{enumerate}
    \item $|\mathcal{N}'|=n-1$, and 
    \item for any $i \in \mathcal{N}', \card(\{\phi_i^*(s)\}_{s \in T})=1$, 
    \end{enumerate}
    then there exists $\sigma_A$ such that $(\mathscr{M}_1,...,\mathscr{M}_n,\sigma_A)$ is a PBE, in which for any $i, \mathscr{M}_i=\{\phi_i^{*}(s)\}_{s \in T}$.
    
\end{corollary}

Corollary \ref{c2} says that \eqref{2+} is automatically satisfied whenever at most one principal offers a non-trivial menu. The intuition is simple: if $n-1$ principals each offer a single outcome, then the only ``choice'' the agent makes is from the remaining principal's menu. With only one menu to choose from, the single-principal logic applies and ties can be broken in that principal's favor without affecting the others.

This result shares the spirit of \cite{Peters2003Negotiation}, which investigates the existence of equilibria where every principal is restricted to a take-it-or-leave-it offer (a singleton menu). My result is less restrictive in two respects: it requires only $n-1$ (rather than all $n$) principals to offer singletons, and it establishes equilibrium existence through the decomposition program \eqref{P3} rather than through primitive restrictions on utility functions (the \emph{no-externalities assumption} in \cite{Peters2003Negotiation}).

\subsection{Characterization of Constructed PBE}
\label{subsec:pareto}

Which equilibria can be recovered through the decomposition? The following definition formalizes the class of equilibria that the decomposition method produces.

\begin{definition}
    \label{def:p3induced}
     $(\mathscr{M}_1,...,\mathscr{M}_n)$ is a \eqref{P3}-induced menu profile if there exists $(\phi_1^*,...,\phi_n^*)$ such that for any $i, \phi_i^*$ solves \eqref{P3}, and \eqref{2+} is satisfied, and for any $i, \mathscr{M}_i=\{\phi_i^*(s)\}_{s \in T}$. Correspondingly, a PBE $(\mathscr{M}_1,...,\mathscr{M}_n, \sigma_A)$ is a \eqref{P3}-induced PBE if $(\mathscr{M}_1,...,\mathscr{M}_n)$ is a \eqref{P3}-induced menu profile. 
\end{definition}

A natural question is whether \eqref{P3}-induced PBE are ``special'' or whether the decomposition captures a broad class of equilibria. The answer depends on the richness of the agent's preferences. Under \eqref{eq:noindifference}, the characterization is sharp:

\begin{proposition}
\label{prop:p3induced}
    If \eqref{eq:noindifference} holds, then a PBE $(\mathscr{M}_1,...,\mathscr{M}_n, \sigma_A)$ is a \eqref{P3}-induced PBE if and only if for any $i \in \mathcal{N}, o_i \in \mathscr{M}_i$, there exists $t^{o_i} \in T$ such that  
    \[ o_i =\sigma_A(\mathscr{M}_1,...,\mathscr{M}_n,t^{o_i})\big|_{O_i}.\]
\end{proposition}

The condition in Proposition \ref{prop:p3induced} requires that every item in every principal's menu is actually chosen by some type of agent in equilibrium---there are no ``unused'' menu items. Under \eqref{eq:noindifference}, this is the \emph{only} distinction between \eqref{P3}-induced PBE and general PBE. In other words, the decomposition captures all equilibria except those in which some principal's menu contains items that no type selects on the equilibrium path.

At first glance, equilibria with unused menu items may seem pathological. However, they can play a substantive economic role: an unused item in principal $i$'s menu can affect the agent's \emph{indirect utility} for principal $-i$'s outcomes, thereby tightening the incentive constraints that principal $-i$ faces. In this sense, unused items function as \textbf{strategic shields}---they have no direct effect on the equilibrium allocation, but their presence in the menu deters the rival from profitable deviations by redirecting the agent's choices when the rival deviates. The \eqref{P3} decomposition, by construction, works with minimal menus (the ranges of direct mechanisms) and therefore cannot capture equilibria sustained by such shields.

I now turn to the welfare properties of \eqref{P3}-induced PBE. A natural concern is whether the decomposition method, by working with minimal menus, systematically misses welfare-superior equilibria.

\begin{definition}
\label{def:paretooptimal}
    Given a type distribution. A PBE $(\mathscr{M}_1,...,\mathscr{M}_n,\sigma_A)$ is Pareto optimal if there does not exist PBE $(\mathscr{M}_1',...,\mathscr{M}_n',\sigma_A')$ such that for any $i \in \mathcal{N}$,
    \[ \mathbb{E}_t[u_i(\sigma_A'(\mathscr{M}_1',...,\mathscr{M}_n',t),t)] \geq \mathbb{E}_t[u_i(\sigma_A(\mathscr{M}_1,...,\mathscr{M}_n,t),t)] ,\]
    and there exists $j \in \mathcal{N}$ such that the inequality holds strictly. Correspondingly, if such $(\mathscr{M}_1',...,\mathscr{M}_n',\sigma_A')$ exists, then $(\mathscr{M}_1',...,\mathscr{M}_n',\sigma_A')$ Pareto dominates $(\mathscr{M}_1,...,\mathscr{M}_n,\sigma_A)$. 
\end{definition}


The characterization in Proposition \ref{prop:p3induced} immediately constraints where Pareto improvements can come from:  under \eqref{eq:noindifference}, the \emph{only} potential source of Pareto improvement over the decomposition is through equilibria sustained by strategic shields---unused menu items whose presence deters rival deviations. An example that shows this potential source of Pareto improvement can indeed make a difference is demonstrated in Online Appendix \ref{sec:paretoimprovements}. 

Beyond the non-indifference condition, the following results establish Pareto optimality of \eqref{P3}-induced PBE in two additional settings. 

\begin{proposition}
\label{prop:PO}
The followings hold:
\begin{enumerate}
    \item If $\card(T)=1$, then any PBE that is Pareto optimal among all \eqref{P3}-induced PBE is Pareto optimal. 
\item If $n=2$, and there exists $i \in \{1,2\}$ such that for any $t,t'$, \[\argmax \limits_{o_i} u_i(o_i,t) =\argmax \limits_{o_i} u_i(o_i,t') := \{o_i^*\},\]  then for all type distributions, any PBE that is Pareto optimal among all \eqref{P3}-induced PBE is Pareto optimal. If \eqref{eq:noindifference} is further satisfied, then any \eqref{P3}-induced PBE is Pareto optimal.\footnote{This statement can be extended \`{a} la Corollary \ref{c2}, i.e., the result still holds if there exists $\mathcal{N}' \subseteq \mathcal{N}$ such that $|\mathcal{N}'|=n-1$, and for any $i \in \mathcal{N}', t,t', $ \[\argmax \limits_{o_i} u_i(o_i,t) =\argmax \limits_{o_i} u_i(o_i,t') := \{o_i^*\}.\]} 
\end{enumerate}
\end{proposition}

\begin{remark}
    The second statement of Proposition \ref{prop:PO} cannot be proved through the following seemingly natural argument: if a \eqref{P3}-induced PBE $(\{o_i^*\}, \mathscr{M}_{-i}, \sigma_A)$ is Pareto dominated by another PBE $(\{o_i^*\}, \mathscr{M}_{-i}', \sigma_A')$, then $\mathscr{M}_{-i}'$ is a profitable deviation for $-i$, contradicting the original profile being a PBE. The flaw is that the agent's strategy $\sigma_A$ may respond differently to $\mathscr{M}_{-i}'$ than $\sigma_A'$ does:
    \[ \mathbb{E}_t[u_{-i}(\sigma_A(\{o_i^*\}, \mathscr{M}_{-i}',t),t)] \leq \mathbb{E}_t[u_{-i}(\sigma_A(\{o_i^*\}, \mathscr{M}_{-i},t),t)] < \mathbb{E}_t[u_{-i}(\sigma_A'(\{o_i^*\}, \mathscr{M}_{-i}',t),t)].\]
    Under $\sigma_A$, the deviation to $\mathscr{M}_{-i}'$ is not profitable; under $\sigma_A'$, it would be. The Pareto improvement is possible precisely because the dominating PBE uses a \emph{different} agent strategy that responds more favorably to the new menu. This subtlety---that Pareto comparisons across PBE involve different agent strategies, while deviation checks within a PBE hold the agent's strategy fixed---is a distinctive feature of the multi-principal setting.
\end{remark}


\subsection{Outside Options}
\label{section:independent}

I now specialize to the setting where $u_i: \Delta(O_i) \rightarrow \mathbb{R}$ for each principal $i$. Without outside options, every solution to the decomposition program automatically constitutes a PBE. Formally, for all $(\phi_1^*,\ldots,\phi_n^*)$ such that for any $i$, $\phi_i^*$ solves
\begin{equation}
\label{P4}
\tag{P2}
\begin{aligned}
\max_{\phi_i: T \rightarrow O_i}\;
&\mathbb{E}_t\!\left[u_i(\phi_i(t))\right] \\
\text{s.t. }\;
v_i\!\left(\phi_i(t),\, t \mid \prod_{j \neq i} \{\phi_j^*(s)\}_{s \in T}\right)
&\ge
v_i\!\left(\phi_i(t'),\, t \mid \prod_{j \neq i} \{\phi_j^*(s)\}_{s \in T}\right),
\quad \forall\, t,t',
\end{aligned}
\end{equation}
there exists $\sigma_A$ such that $(\mathscr{M}_1,\ldots,\mathscr{M}_n,\sigma_A)$ is a PBE with $\mathscr{M}_i = \{\phi_i^*(s)\}_{s \in T}$. The argument is immediate: since $u_i$ is type-independent, the constant mechanism $\phi_i(t) = \argmax_{\tilde{o}_i} u_i(\tilde{o}_i)$ is always feasible, so any solution $\phi_i^*$ satisfies $u_i(\phi_i^*(t)) = \max_{\tilde{o}_i} u_i(\tilde{o}_i)$ for all $t$, making \eqref{2+} hold trivially.

Outside options restore a meaningful screening problem. I distinguish two forms. Under \emph{intrinsic} common agency, the agent must contract with all principals or none:
\begin{equation} \label{i} v_i(\phi_i(t),t|\mathscr{M}_{-i}) \geq 0. \tag{I} \end{equation}
Under \emph{delegated} common agency, the agent can selectively opt out of individual principals:
\begin{equation} \label{d} v_i(\phi_i(t),t|\textstyle\prod_{j \neq i} [\mathscr{M}_j \cup \{\underline{o}_j\}]) \geq v_i(\underline{o}_i, t |\textstyle\prod_{j \neq i} [\mathscr{M}_j \cup \{\underline{o}_j\}]), \tag{D} \end{equation}
where $\underline{o}_i \in O_i$ is the outside option at principal $i$. The agent's optimality condition adjusts correspondingly: under intrinsic common agency,
\begin{equation}
\label{1i}
\sigma_A(\mathscr{M}_1,\ldots,\mathscr{M}_n,t) = \begin{cases} \argmax_{o \in \times_i \mathscr{M}_i} V(o,t) & \text{if } \max_{o \in \times_i \mathscr{M}_i} V(o,t) \geq 0, \\ quit & \text{otherwise;} \end{cases}
\tag{1I}
\end{equation}
under delegated common agency,
\begin{equation}
\label{1d}
(o_1,\ldots,o_n) \in \supp \sigma_A(\mathscr{M}_1,\ldots,\mathscr{M}_n,t) \;\Rightarrow\; (o_1,\ldots,o_n) \in \argmax_{\tilde{o}_i \in \mathscr{M}_i \cup \{\underline{o}_i\},\, \forall i} V(\tilde{o}_1,\ldots,\tilde{o}_n,t).
\tag{1D}
\end{equation}

The decomposition extends directly, with adapted UPR conditions:

\begin{proposition}
\label{prop:intrinsicca}
Under intrinsic common agency, for all $(\phi_1^*,\ldots,\phi_n^*)$ such that for any $i$, $\phi_i^*$ solves
\begin{equation}
\label{P4'}
\tag{P2-I}
\begin{split}
\max_{\phi_i: T \rightarrow O_i \cup \{quit\}} \;&\mathbb{E}_t[u_i(\phi_i(t))], \\
\text{s.t. } v_i(\phi_i(t),t|\prod_{j \neq i} \{\phi_j^*(s)\}_{s \in T} \setminus \{quit\}) &\geq v_i(\phi_i(t'),t|\prod_{j \neq i} \{\phi_j^*(s)\}_{s \in T} \setminus \{quit\}), \;\forall t,t', \\
v_i(\phi_i(t),t|\prod_{j \neq i} \{\phi_j^*(s)\}_{s \in T} \setminus \{quit\}) &\geq 0, \;\forall t,
\end{split}
\end{equation}
and \eqref{eq:4i} holds, there exists $\sigma_A$ such that $(\mathscr{M}_1,\ldots,\mathscr{M}_n,\sigma_A)$ is a PBE with $\mathscr{M}_i = \{\phi_i^*(s)\}_{s \in T} \setminus \{quit\}$.

Under delegated common agency, for all $(\phi_1^*,\ldots,\phi_n^*)$ such that for any $i$, $\phi_i^*$ solves
\begin{equation}
\label{P4''}
\tag{P2-D}
\begin{split}
\max_{\phi_i: T \rightarrow O_i} \;&\mathbb{E}_t[u_i(\phi_i(t))], \\
\text{s.t. } v_i(\phi_i(t),t|\textstyle\prod_{j \neq i} [\{\phi_j^*(s)\}_{s \in T} \cup \{\underline{o}_j\}]) &\geq v_i(\phi_i(t'),t|\textstyle\prod_{j \neq i} [\{\phi_j^*(s)\}_{s \in T} \cup \{\underline{o}_j\}]), \;\forall t,t', \\
v_i(\phi_i(t),t|\textstyle\prod_{j \neq i} [\{\phi_j^*(s)\}_{s \in T} \cup \{\underline{o}_j\}]) &\geq v_i(\underline{o}_i,t|\textstyle\prod_{j \neq i} [\{\phi_j^*(s)\}_{s \in T} \cup \{\underline{o}_j\}]), \;\forall t,
\end{split}
\end{equation}
and \eqref{4D} holds, there exists $\sigma_A$ such that $(\mathscr{M}_1,\ldots,\mathscr{M}_n,\sigma_A)$ is a PBE with $\mathscr{M}_i = \{\phi_i^*(s)\}_{s \in T}$.
\end{proposition}

The UPR conditions are:
\begin{equation}
\label{eq:4i}
\tag{UPR-I}
\begin{split}
&\max_{o_i \in \{\phi_i^{*}(s)\} \setminus \{quit\},\, \forall i} V((o_1,\ldots,o_n),t) < 0 \;\Rightarrow\; \forall i,\; \phi_i^{*}(t) = quit; \\
&\max_{o_i \in \{\phi_i^{*}(s)\} \setminus \{quit\},\, \forall i} V((o_1,\ldots,o_n),t) \ge 0 \;\Rightarrow\; \exists (o_1^t,\ldots,o_n^t) \in \argmax V \text{ s.t. } \forall i,\; u_i(o_i^t) = u_i(\phi_i^{*}(t)),
\end{split}
\end{equation}
\begin{equation}
\label{4D}
\tag{UPR-D}
\forall t,\; \exists (o_1^t,\ldots,o_n^t) \in \argmax_{o_i \in \{\phi_i^{*}(s)\}_{s \in T} \cup \{\underline{o}_i\},\, \forall i} V((o_1,\ldots,o_n),t) \text{ s.t. } \forall i,\; u_i(o_i^t) = u_i(\phi_i^{*}(t)).
\end{equation}

The logic is the same as Proposition \ref{p1}: the UPR conditions ensure aligned tie-breaking, now accounting for the quit option (intrinsic) or outside options (delegated). An example demonstrating that violations of \eqref{eq:4i} or \eqref{4D} can preclude equilibrium is in Online Appendix \ref{appendixegoo}. 

It is worth mentioning that Corollary \ref{c2}, where at least all but one principal offer singleton menus, still applies to both intrinsic common agency and delegated common agency. Formally, for intrinsic common agency: 

\begin{corollary}
\label{cor:n-1singletonoo}
For any $(\phi_1^*,\ldots,\phi_n^*)$ that solves \eqref{P4'}, if there exists $\mathcal{N}' \subseteq \mathcal{N}$ with $|\mathcal{N}'|=n-1$ such that for any $i \in \mathcal{N}'$, $\card(\{\phi_i^*(s)\}_{s \in T} \setminus \{quit\})=1$, and all principals assign the same types to the quit option (i.e., $\{t: \phi_{n_1}(t)=quit\} = \{t: \phi_{n_2}(t)=quit\}$ for any $n_1, n_2 \in \mathcal{N}$), then there exists $\sigma_A$ such that $(\mathscr{M}_1,\ldots,\mathscr{M}_n,\sigma_A)$ is a PBE with $\mathscr{M}_i = \{\phi_i^*(s)\}_{s \in T} \setminus \{quit\}$.

\end{corollary}

\section{Delegation}
\label{sdelegation}

The two applications that follow---delegation (this section) and competitive bundling (Section \ref{sec:multiproductduopoly})---both feature solving the decomposition program \eqref{P3} using the envelope theorem. This requires the envelope formula to hold for the indirect utility function. Since the indirect utility function is defined as the upper envelope of a family of functions (Definition \ref{d1}), it is generically not differentiable everywhere, and the standard envelope theorem does not apply directly. In Online Appendix \ref{s21}, I establish that the envelope formula remains valid for upper envelopes of $C^1$ functions, provided that all kinks are ``upward''---that is, the right derivative weakly exceeds the left derivative at every point (Lemma \ref{l2})---and show that this condition is automatically satisfied for upper envelopes of $C^1$ functions (Lemma \ref{l3+}). I also show that such upper envelopes are absolutely continuous and admit integrable derivative bounds (Lemma \ref{l1}), which ensures that the ``upper envelope of upper envelopes'' structure arising from the indirect utility is well-behaved. 

For notational simplicity, I focus on two principals, indexed by $i \in \{1,2\}$, with compact outcome spaces $O_1, O_2 \subseteq \mathbb{R}$. The agent's private type $t$ is drawn from $T:=[\underline{t},\bar{t}]$ according to a distribution $F$ with continuous, strictly positive density $f$. Utilities are:
\[ 
V((o_1,o_2),t) = -(t - o_1 - o_2)^2, \qquad u_i(o_i,t) = -r_i(t)(o_i - o_i^*(t))^2,
\]
where $r_i: T \rightarrow \mathbb{R}_{+}$ and $o_i^*:T \rightarrow \mathbb{R}$ represent principal $i$'s weight and ideal point. The agent has no outside options, so the setting aligns with \eqref{P3}.

The quadratic structure has a transparent economic interpretation: the agent wants total policy $o_1 + o_2$ to match her type $t$, while each principal wants her own outcome $o_i$ close to her ideal $o_i^*(t)$. The tension between the agent's preference for adaptation and each principal's preference for control is the source of the screening problem.

\subsection{A Type-Independent-Peaked Principal}
\label{sec:singlepeaked}

I first consider the case where one principal has a type-independent ideal point: there exists $i$ such that
\begin{equation} \label{eq:single-peaked} o_i^*(t)=o_i^* \text{ for all } t \in T. \tag{type-independent-peaked} \end{equation}
Since $\argmax_{o_i} u_i(o_i,t)=\{o_i^*\}$ for all $t$, principal $i$ always wants the same outcome regardless of the agent's type. Consequently, $\phi_i^*(t)=o_i^*$ for all $t$ solves \eqref{P3} for principal $i$ regardless of $\phi_{-i}^*$, and Corollary \ref{c2} ensures that \eqref{2+} is automatically satisfied. In addition, by Proposition \ref{prop:PO}, any PBE that is Pareto optimal among \eqref{P3}-induced PBE is Pareto optimal in this setting.

The problem reduces entirely to \eqref{P3} for principal $-i$:
\begin{equation*}
\begin{split}
    \max_{o_{-i}: T \rightarrow O_{-i}}\; \mathbb{E}_t[-&r_{-i}(t)(o_{-i}(t)-o^*_{-i}(t))^2], \\
    \text{s.t. }\; 
    -(t-o_i^*-o_{-i}(t))^2 
    &\ge 
    -(t-o_i^*-o_{-i}(t'))^2 ,\quad \forall t,t'.
\end{split}
\end{equation*}
This is a standard single-principal delegation problem with the agent's bliss point shifted by $o_i^*$. Applying the envelope formula converts the IC constraint into an integral condition.\footnote{Here the envelope formula can be applied directly without the justifications in Online Appendix \ref{s21}, as the indirect utility $v_{-i}(o_{-i}, t | \{o_i^*\}) = -(t - o_i^* - o_{-i})^2$ is $C^1$ in $t$.}  To solve this, I follow the Lagrangian approach (\cite{AmadorBagwell2013OptimalDelegation}, \cite{KartikKleinerVanWeelden2021Delegation}): I construct a relaxed problem that penalizes violations of the envelope condition in the objective with weight $\kappa := \frac{\min_t r_{-i}(t)}{2}$, assuming it is well-defined. The choice of $\kappa$ ensures that the relaxed objective is a concave functional of $o_{-i}(\cdot)$, so that any solution satisfying the original envelope condition also solves the original problem. The details of the Lagrangian construction and the proof are in Appendix \ref{appendix:finitemenus}.

\begin{proposition}
\label{prop:singlepeaked}
Suppose \eqref{eq:single-peaked} is satisfied for principal $i$. The following are \eqref{P3}-induced menu profiles:
\begin{enumerate}
    \item \textbf{Full delegation} for principal $-i$: if for any $t, t-o_i^* \in O_{-i}$, and 
    \begin{itemize}
    \item $\kappa F(t)+2r_{-i}(t)(t-o_i^*-o^*_{-i}(t))f(t)$ is right-continuous and increasing on $[\underline{t}, \bar{t})$, and  
    \item $\underline{t}-o_i^*=o^*_{-i}(\underline{t})$, and $\bar{t}-o_i^*-o^*_{-i}(\bar{t}) \leq 0$,
    \end{itemize}
    then $\{\phi_i^*(s)\}_{s \in T}=\{o_i^*\}, \{\phi_{-i}^*(s)\}_{s \in T}=\{s-o_i^*\}_{s \in T}$.
    \item \textbf{No compromise} for principal $-i$: if for any $t$,
    \[(t-o_i^*)o^*_{-i}(t)-\frac{o^*_{-i}(t)^2}{2}+\frac{o^*_{-i}(\underline{t})^2}{2}-(\underline{t}-o_i^*)o^*_{-i}(\underline{t})
    =
    \int_{\underline{t}}^t o^*_{-i}(s)\mathbf{d}s,\]
    then $\{\phi_i^*(s)\}_{s \in T}=\{o_i^*\}, \{\phi_{-i}^*(s)\}_{s \in T}=\{o^*_{-i}(s)\}_{s \in T}$.\footnote{A special case to check: for all $t,t', o^*_{-i}(t)=o^*_{-i}(t'):=o^*_{-i}$, i.e., the case where \eqref{eq:single-peaked} is also satisfied for principal $-i$. In this case, $\underline{t}-o_i^*-o^*_{-i}=0$ and $\bar{t}-o_i^*-o^*_{-i} \leq 0$ cannot be satisfied at the same time, so that Proposition \ref{prop:singlepeaked} does not suggest Full Delegation for principal $-i$ is part of a \eqref{P3}-induced menu profile. Meanwhile,
\[(t-o_i^*)o^*_{-i}-\frac{(o^*_{-i})^2}{2}+[\frac{(o^*_{-i})^2}{2}-(\underline{t}-o_i^*)o^*_{-i}]
    =
    \int_{\underline{t}}^t o^*_{-i}\mathbf{d}s,\]
so that Proposition \ref{prop:singlepeaked} suggests that $(\{o_i^*\}, \{o^*_{-i}\})$ is a \eqref{P3}-induced menu profile, which is straightforward. } 
\end{enumerate}
\end{proposition}

The full delegation result says that principal $-i$ grants the agent complete discretion---the agent chooses $o_{-i} = t - o_i^*$, exactly matching her bliss point. This is optimal when principal $-i$'s ideal point is sufficiently aligned with the agent's, formalized by the monotonicity condition on the Lagrangian. The no-compromise result says that principal $-i$ implements her own ideal point $o^*_{-i}(t)$ for every type, with no concession to the agent's preferences. This is optimal when the envelope condition is satisfied at the principal's ideal---intuitively, when the principal's ideal point moves with $t$ in a way that is already incentive-compatible.

\subsection{The General Case}
I now consider the general case where both principals may have
type-dependent ideal points. This is the substantive delegation
problem: neither principal can implement her ideal by default, and
the screening structure must be derived endogenously. I present two
results. The first (Proposition \ref{prop:delegationeasy}) identifies
a special alignment under which both principals achieve no
compromise; the second (Proposition \ref{prop:delegationhard})
characterizes equilibria in which one principal offers a coarse
discrete menu while the other exercises fine discretion within each
segment, which is the central delegation structure produced by the
decomposition. 

Given principal $i$ offering menu $\{\phi_i^*(s)\}_{s \in T}$, define the agent's optimal selection from this menu as:
\[
\tilde{o}_i(t,t'\mid o_{-i})\in\argmax_{o_i\in \{\phi_i^*(s)\}_{s \in T}}-(t-o_i-o_{-i}(t'))^2.
\]
When $S_i := \{\phi_i^*(s)\}_{s \in T}$ is compact, this equals the nearest-point projection $\Pi_{S_i}(t - o_{-i}(t'))$. Since $\Pi_{S_i}$ is monotone non-decreasing on $\mathbb{R}$ (by a direct optimality argument), it is differentiable almost everywhere by Lebesgue's theorem, with $\Pi_{S_i}'(\cdot) \in \{0,1\}$ almost everywhere.

Applying the envelope formula to the indirect utility transforms principal $-i$'s IC constraint into an integral condition. I then construct a relaxed problem similar to the one in Section \ref{sec:singlepeaked} The Lagrangian and Gateaux differential for this relaxed problem are developed in Appendix \ref{appendix:finitemenus}.

I begin with a result that does not require the Lagrangian machinery.

\begin{proposition}
\label{prop:delegationeasy}
 If $o_{i}^*(t)+o^*_{-i}(t)=t$ for all $t$,
        then No Compromise for both principals is a \eqref{P3}-induced menu profile: $\{\phi_i^*(s)\}_{s \in T}=\{o_i^*(s)\}_{s \in T}, \{\phi_{-i}^*(s)\}_{s \in T}=\{o^*_{-i}(s)\}_{s \in T}$.
\end{proposition}

When the principals' ideal points sum to the agent's type, each principal can implement her ideal without violating IC, because the agent's bliss point is automatically hit. The proof is immediate: $\phi_i^*(t) = o_i^*(t)$ solves the \emph{unconstrained} version of \eqref{P3}, and feasibility and UPR both follow from the fact that the intended allocation $(o_i^*(t), o^*_{-i}(t))$ is the agent's unique optimum.

\begin{remark}
If the condition is weakened to: for any $t$, there exist $t',t''$ such that $t-o_{i}^*(t)=o^*_{-i}(t')$ and $t-o^*_{-i}(t)=o_i^*(t'')$, then $\phi_i^*(t)=o_i^*(t)$ still mutually solves \eqref{P3}, but \eqref{2+} may fail. The weaker condition ensures that the ``right'' outcome for each type exists in the rival's menu, but does not prevent the agent from selecting a ``wrong'' cross-pairing that hurts one of the principals.
\end{remark}

The next result uses the Lagrangian approach to characterize equilibria
where principal 1 offers a finite menu of ``regimes'' and principal 2
receives piecewise full delegation within each regime.

\begin{proposition}\label{prop:delegationhard}
    Suppose there exist $K \in \mathbb{N}$, $t^1,...,t^K$ such that $\underline{t}=t^0<t^1<\dots<t^K<t^{K+1}=\bar{t}$, and $o_1^1,\dots,o_1^{K+1} \in O_1$ such that for all $x \in \{1,\dots,K+1\}$ and $t \in [t^{x-1}, t^x)$, $o_1^*(t)=o_1^x$ and $t-o_1^x \in O_2$. Furthermore, assume:
    \begin{enumerate}
        \item[(i)] $\kappa F(t)+2r_2(t)(t-o_1^x-o_2^*(t))f(t)$ is right-continuous and increasing on $[t^{x-1}, t^x)$ for all $x \in \{1,\dots,K+1\}$; and 
        \item[(ii)] for all $x \in \{1,\dots,K\}$,
        \[ \kappa F(t^x)+2r_2(t^x)(t^x-o_1^{x+1}-o_2^*(t^x))f(t^x) \geq \lim_{t \to (t^x)^{-}} \big[ \kappa F(t)+2r_2(t)(t-o_1^{x}-o_2^*(t))f(t) \big]; \text{ and } \]
        \item[(iii)] $\underline{t}-o_1^1-o_2^*(\underline{t})=0$ and $\bar{t}-o_1^{K+1}-o_2^*(\bar{t}) \leq 0$,
    \end{enumerate}
    then $\big(\{o_1^1,\dots,o_1^{K+1}\}, \bigcup_{x=1}^{K+1} \{s-o_1^x\}_{s \in [t^{x-1}, t^x)}\big)$ is a \eqref{P3}-induced PBE menu profile.\footnote{The interval $[t^{K}, t^{K+1})=[t^K, \bar{t})$ should be regarded as the closed $[t^K, \bar{t}]$.}
\end{proposition}

Proposition \ref{prop:delegationhard} produces equilibria where
principal 1 partitions the type space into $K+1$ segments and offers a
distinct outcome for each, while principal 2 has full discretion within
each segment---choosing $o_2 = t - o_1^x$ to match the agent's bliss
point given principal 1's regime $o_1^x$. Since the agent achieves her
bliss point under the equilibrium allocation, there are no ties and
\eqref{2+} is automatically satisfied.

The conditions are piecewise generalizations of Proposition
\ref{prop:singlepeaked}'s full delegation conditions. Within each
segment $[t^{x-1}, t^x)$, the first condition is identical to the
full delegation condition in Proposition \ref{prop:singlepeaked},
applied to principal 2 facing a type-independent-peaked principal 1
with ideal $o_1^x$. The second condition governs the transitions
between segments: the Lagrange multiplier must jump upward at each
boundary $t^x$, ensuring that the piecewise-constructed multiplier
remains increasing globally. The third condition is the same boundary
condition as in Proposition \ref{prop:singlepeaked}.

When $K = 0$, there are no sub-segments, the conditions reduce exactly
to Proposition \ref{prop:singlepeaked} (full delegation for principal
2 under a type-independent-peaked principal 1), and principal 1 offers
a singleton menu. The result thus strictly generalizes the
type-independent-peaked case to settings where principal 1's ideal
point is piecewise constant: she wants different outcomes for
different segments of the type space, but her preference is rigid
within each segment.

The structure captures environments in which the agent combines a
discrete outcome chosen from one principal's menu with a continuous
outcome chosen from another's---one principal supplying a coarse
component of the final allocation, the other supplying a fine
adjustment. The combination is familiar in hierarchical economic
settings. In coalition policymaking, one actor may offer one of a
few discrete policy positions while another offers a continuous
fiscal adjustment; the implemented policy combines the two selected
components. Under layered regulation, a federal authority chooses
among a few compliance regimes while a state authority sets a
continuous enforcement or pricing adjustment within the chosen
regime. In franchise networks, the franchisor selects among a
discrete set of store formats while the franchisee adapts pricing
and product mix locally within the format. Proposition
\ref{prop:delegationhard} shows that this coarse-plus-fine division
of discretion is not merely a descriptive pattern but an equilibrium
outcome of the delegation game under the stated conditions.

 \section{Bundling}
    \label{sec:multiproductduopoly}

I now apply the decomposition to a duopoly extension of the multi-product bundling model with non-additive values and one-dimensional types (\cite{ghili2023characterization}, \cite{yang2023nested}, \cite{feng2025}). There are $r$ goods indexed by $\{1,\ldots,r\}$, and two firms (principals) each sell bundles $b \in \mathcal{B} := 2^{\{1,\ldots,r\}}$. A consumer (the agent) with private type $t \sim F$ on $T = [\underline{t}, \bar{t}]$ chooses from both firms' menus simultaneously.

The outcome space for each firm is $O_i = \mathcal{B} \times [0, \bar{p}]$, where $\bar{p}$ is a maximum transfer set to ensure compactness, and the consumer's utility from buying bundle $b_i$ at price $p_i$ from each firm is
\[
V((b_1,p_1),(b_2,p_2),t) = U((b_1,b_2),t) - p_1 - p_2,
\]
where $U: \mathcal{B}^2 \times T \to \mathbb{R}$ is the gross valuation, continuously differentiable in $t$ and linear in probabilities over bundles. Firms maximize revenue (costs normalized to zero). I assume \emph{intrinsic common agency}: the consumer either buys from both firms or neither, corresponding to \eqref{P4'} with the outside option of buying nothing.

Two features of this setting are worth noting. First, in any \eqref{P4'}-induced menu profile, each firm's menu is finite: incentive compatibility precludes offering the same bundle at two different prices, so each bundle appears at most once. This ensures that the indirect utilities are well-defined and the envelope formula (Online Appendix \ref{s21}) applies. Second, because firms' payoffs are type-independent (equal to the transfer), the analysis falls under the framework of Section \ref{section:independent}.

\subsection{Market Splitting}

How do competing firms divide a multi-product market? The
decomposition yields a class of equilibria exhibiting a 
market-splitting structure: firms partition the product space into
complementary segments, with each firm's bundles pairing with the
rival's to form combinations that jointly maximize consumer utility
and virtual surplus. The following examples illustrate the structure
before I state the general result.

\begin{example}
\label{e1}
Let $U((b,b'),t) = t \cdot g(b \cup b')$, where $g: \mathcal{B} \to \mathbb{R}_+$ is strictly increasing in the set inclusion order. The consumer values the \emph{union} of goods purchased from both firms---there is no penalty for overlap. If $t - \frac{1-F(t)}{f(t)}$ is strictly increasing and positive, and $\frac{\underline{t}}{2} < \frac{1}{f(\underline{t})}$, then $(\{(b_1,p_1),\ldots,(b_l,p_l)\}, \{(b_1',p_1'),\ldots,(b_m',p_m')\})$ is a \eqref{P4'}-induced menu profile if and only if every bundle from firm 1 can be paired with some bundle from firm 2 to form the grand bundle $b^*$ (and vice versa), and all prices equal $\frac{t_*}{2} g(b^*)$, where $t_*$ satisfies $\frac{t_*}{2} = \frac{1-F(t_*)}{f(t_*)}$.

In this setting, any pair of bundles whose union covers all goods is optimal. Firms split the market by offering complementary bundles that jointly cover $b^*$, but the specific partition is indeterminate---many equilibrium menus coexist.
\end{example}

\begin{example}
\label{e2}
Let $U((b,b'),t) = t \cdot [g(b \cup b') - g(b \cap b')]$, where $g$ is as above. Now overlap is costly: buying the same good from both firms wastes surplus. Under the same distributional conditions, a menu profile is \eqref{P4'}-induced if and only if every bundle from firm 1 is the \emph{complement} of some bundle from firm 2 (i.e., $b_i = (b_j')^\complement$), and all prices equal $\frac{t_*}{2} g(b^*)$.

This is a sharper form of market splitting: firms must partition the goods with no overlap. The overlap penalty forces complete specialization. Note that the complementarity condition implies $l = m$: the two menus have the same number of items.
\end{example}

\begin{remark}
In Example \ref{e2}, if production costs differ across bundles, the equilibrium sharpens further: each firm offers a single bundle, and the two bundles are mutual complements. Heterogeneous costs break the indifference across different partitions of $b^*$, selecting the cost-minimizing split.
\end{remark}

The following proposition generalizes both examples. Define the set of \emph{jointly optimal} bundle pairs:
\begin{equation*}
\mathcal{B}_*^2 := \bigcap_{t \in T} \argmax_{(\tilde{b},\tilde{b}')} U((\tilde{b},\tilde{b}'),t) \;\;\bigcap\;\; \bigcap_{t \in T} \argmax_{(\tilde{b},\tilde{b}')} \left[ U((\tilde{b},\tilde{b}'),t) - \frac{1-F(t)}{f(t)} U_t((\tilde{b},\tilde{b}'),t) \right].
\end{equation*}
These are the bundle pairs that simultaneously maximize gross utility \emph{and} virtual surplus for every type---the analogue of ``efficient bundles'' in the monopoly problem, extended to pairs.

\begin{proposition}
\label{t1}
If for any $(b,b') \in \mathcal{B}^2$, both $U((b,b'),t)$ and $U((b,b'),t) - \frac{1-F(t)}{f(t)} U_t((b,b'),t)$ are strictly increasing in $t$, and $\mathcal{B}_*^2 \neq \emptyset$, then a menu profile $(\{(b_1,p_1),\ldots,(b_l,p_l)\}, \{(b_1',p_1'),\ldots,(b_m',p_m')\})$ is \eqref{P4'}-induced if
\begin{equation}
\label{eq:marketsplitting}
\begin{aligned}
&\forall i,\; \exists j \text{ s.t. } (b_i,b_j') \in \mathcal{B}_*^2, \qquad \forall j',\; \exists i' \text{ s.t. } (b_{i'},b_{j'}') \in \mathcal{B}_*^2, \\[0.5ex]
&p_i = p_j' = \tfrac{1}{2}\max_{(b,b')} U((b,b'),t_*), \quad \forall i,j,
\end{aligned}
\end{equation}
where $t_*$ satisfies $\frac{1}{2}\max_{(b,b')} U((b,b'),t_*) = \max_{(b,b')} U((b,b'),t_*) - \frac{1-F(t_*)}{f(t_*)} U_t((b,b'),t_*)$. If additionally, for any $b \in \mathcal{B}$ there exists $b'$ with $(b,b') \in \mathcal{B}_*^2$, the converse also holds.
\end{proposition}

Proposition \ref{t1} says that when the conditions hold, firms'
equilibrium menus exhibit a cross-matching structure into
$\mathcal{B}_*^2$ and symmetric pricing at $t_*$. The economic
prediction is structural rather than pointwise: the decomposition
pins down the form of the split (complementary pairing, equal
markups) but not the specific partition across firms.

\begin{remark}
In Example \ref{e1}, $\mathcal{B}_*^2 = \{(b,b') \in \mathcal{B}^2 : b \cup b' = b^*\}$; in Example \ref{e2}, $\mathcal{B}_*^2 = \{(b,b') \in \mathcal{B}^2 : b = (b')^\complement\}$. The overlap penalty in Example \ref{e2} shrinks $\mathcal{B}_*^2$ from all pairs covering $b^*$ to only disjoint partitions, producing the sharper specialization noted there.
\end{remark}

\subsection{Base Product $+$ Complementary Upgrades}

In many markets, competition takes an asymmetric form: one firm acts as a gatekeeper providing a fundamental ``base'' good, while another provides a menu of differentiated ``upgrades''. A canonical example is the mobile telecommunications market: a network carrier offers a standard service plan (the base product), while a device manufacturer offers a tiered lineup of smartphones (the upgrades), ranging from budget models to premium flagships. This structure fits perfectly into the \emph{intrinsic} common agency setting: the consumer generally views the SIM card and the handset as a single functional unit—buying the service plan without the phone is useless, and vice versa. 

The following results characterize the equilibrium where the ``base'' principal offers a single, take-it-or-leave-it contract that extracts the baseline surplus, while the ``upgrade'' principal performs the screening task, offering a sophisticated, possibly multi-upgrade-line menu to discriminate among consumer types.
    
\begin{proposition}
\label{t2}
If $U$ is strictly increasing in $t$ and satisfies:
\begin{enumerate}
    \item \label{p10cd1} for any $b_1 \in \mathcal{B}$: the virtual surplus $\varphi^{b_1}(b,t) := U((b_1,b),t) - \frac{1-F(t)}{f(t)} U_t((b_1,b),t)$ satisfies MD$^\star$ (\cite{kartik2023});\footnote{For any $a, a' \in \Delta(\mathcal{B})$,
    the difference $\varphi^{b_1}(a,t) - \varphi^{b_1}(a',t)$ is
    monotonic in $t$. An additional assumption is that for
    any $a \in \Delta(\mathcal{B})$, $\varphi^{b_1}(a,t)$ is
    monotonic in $t$. This is implied by MD$^\star$ if there exists
    $a$ with $\varphi^{b_1}(a,t) = 0$ for all $t$.} the complement $b_1^\complement$ uniquely maximizes the lowest type's virtual surplus; and for any $b \not\supset b_1^\complement$, $U((b_1,b),\bar{t}) < U((b_1,b_1^\complement),\bar{t})$,
    \item \label{p10cd2} for any $b_2 \in \mathcal{B}$: the complement $b_2^\complement$ maximizes both gross utility and virtual surplus for all types, and $U((b_2^\complement, b_2),t) - \frac{1-F(t)}{f(t)} U_t((b_2^\complement, b_2),t)$ is strictly increasing in $t$,
    \item \label{p10cd3} partition symmetry: $U((b_1,b_2),\cdot) = U((b_1',b_2'),\cdot)$ whenever $b_1 \cup b_2 = b_1' \cup b_2' = b^*$ and $b_1 \cap b_2 = b_1' \cap b_2' = \emptyset$,
\end{enumerate}
then for any $b \in \mathcal{B}$, there exists a \eqref{P4'}-induced menu profile $(\{(b,p)\}, \{(b_1^*,p_1^*),\ldots,(b_m^*,p_m^*)\})$, where $\{b_1^*,\ldots,b_m^*\}$ is a nested or tree menu rooted at $b^\complement$ in the sense of \cite{feng2025}, and prices satisfy $p = \min_x p_x^* = \frac{1}{2} U((b, b^\complement), t_*)$, where $t_*$ satisfies $\frac{1}{2} U((b,b^\complement),t_*) = \frac{1-F(t_*)}{f(t_*)} U_t((b,b^\complement),t_*)$.
\end{proposition}

Condition 1 ensures that firm 2's screening problem---conditional on firm 1's base product $b$---satisfies the regularity conditions of \cite{feng2025}, so the optimal upgrade menu takes a nested or tree structure rooted at the complement $b^\complement$. Condition 2 ensures that firm 1's singleton contract is a best response: for any base product $b$, the complement $b^\complement$ is always optimal for the rival, so firm 1 cannot profitably deviate by changing its bundle. Condition 3 (partition symmetry) ensures that the consumer's valuation depends only on which goods are covered and which are duplicated, not on which firm sells which good.

Under the utility specification of Example \ref{e2}---$U((b,b'),t) =
t \cdot [g(b \cup b') - g(b \cap b')]$---all three
conditions hold but the equilibrium degenerates: the upgrade menu
reduces to the singleton $\{(b^\complement, p)\}$, so both firms
post take-it-or-leave-it contracts. Conditions \ref{p10cd2} and
\ref{p10cd3} together preclude any bundle strictly dominating
$b^\complement$ at the top type. The next corollary identifies
conditions under which the upgrade menu is genuinely
non-degenerate.

\begin{corollary}
\label{c1}
If the conditions of Proposition \ref{t2} hold with condition \ref{p10cd1} strengthened to: there exists $b_0 \in \mathcal{B}$ with $b_0' \supsetneq b_0^\complement$ such that $U((b_0, b_0'), \bar{t}) > U((b_0, b_0^\complement), \bar{t})$, and condition \ref{p10cd3} replaced by: for any $b' \subsetneq b_0$, $\mathbb{E}_t[\varphi^{b_0}(b_0^\complement, t)] \geq \mathbb{E}_t[\varphi^{b'}((b')^\complement, t)]$, then $(\{(b_0,p_0)\}, \{(b_1^*,p_1^*),\ldots,(b_m^*,p_m^*)\})$ is a \eqref{P4'}-induced menu profile with $m \geq 2$---the upgrade menu is non-degenerate.
\end{corollary}

The additional condition says that some upgrade bundle strictly dominates the base complement for the highest type, so firm 2 has a genuine incentive to screen. The expected virtual surplus condition ensures that $b_0$ is the most profitable base product for firm 1. Together, these produce the asymmetric equilibrium: firm 1 sells a standardized base product $b_0$ at a markup that extracts half the baseline surplus, while firm 2 offers a tree of upgrades rooted at $b_0^\complement$---from a basic complement for low types to premium bundles for high types.

\medskip\noindent\textbf{Discussion.} Propositions \ref{t1} and
\ref{t2} identify two classes of equilibria in the duopoly bundling
game. Market splitting is symmetric: both firms offer comparable
menus that cross-match into $\mathcal{B}_*^2$. The base-plus-upgrades
structure is asymmetric: one firm's singleton menu reduces the
rival's problem to monopoly screening, which is then solved by the
tools of \cite{feng2025}. I do not claim these exhaust the
equilibrium set; they are the classes the decomposition constructs
under the stated conditions.

A natural next step is to characterize \emph{all} \eqref{P4'}-induced menu profiles of the form \[(\{(b_1,p_1)\}, \{(b_2,p_2),\ldots,(b_m,p_m)\}).\] Under MD$^\star$, the upgrade firm's optimal menu is determined by the upper envelope of $\varphi^{b_1}$. The remaining step---verifying that the base contract $(b_1,p_1)$ is a best response to that menu---is more subtle, as it depends on the interaction between the consumer's selection from the upgrade menu and the base firm's virtual surplus.

\section{Extensions}

The preceding analysis relies on principals' payoffs being the forms of $u_i(o_i,t)$, which makes each principal's problem self-contained once the rivals' menus are fixed. When principals' payoffs depend on the full allocation profile, the agent's
equilibrium selection rule $\sigma_A$ enters each principal's
objective directly. The decomposition does not adapt cleanly to this
setting: the indirect utility can no longer absorb all dependence on
rivals, and each principal's problem depends on the specific
selection rule the agent uses off-path.

I identify regularity conditions on $\sigma_A$ under which a modified
decomposition delivers a best-response characterization. These
conditions are substantive restrictions on the
agent's off-path behavior, not implications of rational choice in
the underlying game. Results in this section should be read
accordingly: conditional on an off-path choice rule of the
specified form, a decomposition-like construction remains possible.

\label{sec:discussions}
\subsection{Characterizing PBE under General Utilities}
I begin with two regularity conditions on the agent's strategy. The first govern pure strategies: 
\begin{equation}
\label{piia}
    \text{ for any } t \in T, \mathcal{M}' \subseteq \mathcal{M} \in \prod_{i=1}^n 2^{O_i}, \text{ if } \supp \sigma_A(\mathcal{M},t) \subseteq \mathcal{M}', \text{ then } \sigma_A(\mathcal{M}',t)=\sigma_A(\mathcal{M},t).
    \tag{IIA-1}
\end{equation}

In words: if the agent's chosen outcome under the larger menu profile $\mathcal{M}$ happens to lie in the smaller profile $\mathcal{M}'$, then the agent makes the same choice when facing $\mathcal{M}'$. This is a consistency requirement: removing options the agent would not have chosen anyway does not change her behavior.

The second condition governs mixed strategies:
\begin{equation}
\label{contraction}
\tag{IIA-2}
\parbox{0.9\textwidth}{
\centering
$\text{ for any } t \in T,\;
\mathcal{M}' \subseteq \mathcal{M} \in \prod_{i=1}^n 2^{O_i},\;
\text{for any measurable } E \subseteq \supp \sigma_A(\mathcal{M},t) \cap \mathcal{M}' \neq \emptyset,$\\[4pt]
$\sigma_A(\mathcal{M}',t)(E)
= \dfrac{ \sigma_A(\mathcal{M},t)(E) }
{ \displaystyle \int_{O} 1_{\{ o \in \supp \sigma_A(\mathcal{M},t) \cap \mathcal{M}' \}} \, \sigma_A(\mathcal{M},t)(\mathbf{d} o) }.$
}
\end{equation}

This says: when the menu contracts, the agent's randomization over surviving outcomes is obtained by conditioning-rescaling the original probabilities proportionally. \eqref{contraction} coincides with the choice axiom of \cite{luce1959}, long
studied in the decision-theoretic literature. While this gives the
condition a familiar form, it does not make it a weak assumption:
Luce-style stochastic choice is a specific functional form for
menu-sensitive choice probabilities, not a generic consequence of
rational optimization.

\begin{remark}
\eqref{piia} is implied by \eqref{contraction}, and the two are equivalent when $\sigma_A$ is a pure strategy. Under \eqref{eq:noindifference}, any PBE involves a pure agent strategy (as shown in Section \ref{sec:mainresults}), and \eqref{piia} is automatically satisfied: if the agent's uniquely optimal outcome under $\mathcal{M}$ lies in $\mathcal{M}' \subseteq \mathcal{M}$, it remains uniquely optimal under $\mathcal{M}'$. Thus, in the non-indifference case, no additional regularity conditions on $\sigma_A$ are needed beyond agent optimality.
\end{remark}

The decomposition program under general utilities takes the following form. Because principal $i$'s payoff depends on outcomes at other principals, the agent's selection rule $\sigma_A$ enters the objective directly:

\begin{proposition}
\label{p0}
    If $\sigma_A$ is a pure strategy that satisfies condition \ref{cd1} in Definition \ref{d0} and 
\eqref{piia}, and for any $i, \phi_i^*$ solves
   \begin{equation}
\label{P1}
\begin{aligned}
\max_{\phi_i: T \rightarrow O_i}\;& 
\mathbb{E}_t\!\left[u_i\!\left(\sigma_A\big(\{\phi_1^*(s)\}_{s\in T},\ldots,\{\phi_i(s)\}_{s\in T},\ldots,\{\phi_n^*(s)\}_{s\in T},t\big),t\right)\right] \\[4pt]
\text{s.t.}\;& 
v_i\!\left(\phi_i(t),t\,\big|\,\{\phi_{-i}^*(s)\}_{s\in T}\right)
\;\ge\;
v_i\!\left(\phi_i(t'),t\,\big|\,\{\phi_{-i}^*(s)\}_{s\in T}\right),
\quad \forall t,t',
\end{aligned}
\tag{P3}
\end{equation}

    then $(\mathscr{M}_1,...,\mathscr{M}_n,\sigma_A)$ is a PBE, in which for any $i, \mathscr{M}_i=\{\phi_i^{*}(s)\}_{s \in T}$.

    Alternatively, if $\sigma_A$ is a mixed strategy that satisfies condition \ref{cd1} in Definition \ref{d0} and \eqref{contraction}, and for any $i, \phi_i^*: T \rightarrow O_i$ solves 
    \begin{equation}
\label{P1'}
\tag{P3'}
\begin{aligned}
    \max_{\phi_i: T \rightarrow \Delta(O_i)}\;
    &\mathbb{E}_t\!\left[
        \int_{O_i}
        u_i\!\left(
            \sigma_A\!\left(
                \{\phi_1^*(s)\}_{s \in T},
                \ldots,
                \tilde{o}_i,
                \ldots,
                \{\phi_n^*(s)\}_{s \in T},
                t
            \right),
            t
        \right)
        \phi_i(t)(\mathbf{d}\tilde{o}_i)
    \right] \\[4pt]
    \text{s.t. }\;
    &v_i(\phi_i(t), t \mid \prod_{j \neq i} \{\phi_j^*(s)\}_{s \in T})
    \;\ge\;
    v_i(\phi_i(t'), t \mid \prod_{j \neq i} \{\phi_j^*(s)\}_{s \in T}),
    \quad \forall t,t'.
\end{aligned}
\end{equation}

    and for any $i$,
    \begin{equation}
    \label{2*}
        \mathbb{E}_t[u_i(\sigma_A(\{\phi_1^*(s)\}_{s \in T},...,\{\phi_i^*(s)\}_{s \in T},...,\{\phi_n^*(s)\}_{s \in T},t),t)] \geq \mathbb{E}_t[u_i(\sigma_A(\{\phi_1^*(s)\}_{s \in T},...,\{\phi_i^*(t)\},...,\{\phi_n^*(s)\}_{s \in T},t),t)], 
        \tag{MEN}
    \end{equation}
    then $(\mathscr{M}_1,...,\mathscr{M}_n,\sigma_A)$ is a PBE, in which for any $i, \mathscr{M}_i=\{\phi_i^{*}(s)\}_{s \in T}$.
\end{proposition}

The role of \eqref{piia} in the pure-strategy case can be understood as follows. Suppose principal $i$ deviates to some menu. Given the deviating menu and other principals' menus, the agent's optimal choice trace out a type-contingent selection from the deviating menu, which defines a direct mechanism, with its range consists of a trimmed menu. \eqref{piia} ensures that the agent's behavior is identical under the deviating and the trimmed menu: since the agent's chosen outcome at each type already belongs to the trimmed menu, removing the unchosen items does not change her choices. 

Consequently, principal $i$'s payoff from offering the deviating menu equals her payoff from offering the trimmed menu. But the trimmed menu is the range of a direct mechanism, and the decomposition program \eqref{P1} already optimizes over all such mechanisms. So the deviation cannot be profitable. 

The mixed-strategy case requires a separate compatibility condition: \eqref{2*}, which I call \textbf{menu-expansion nonharmfulness}. Note that \eqref{2*} plays the same role as \eqref{2+} in Section \ref{sec:mainresults}---both ensure that the agent's behavior at the full menu profile is compatible with all principals' interests---but \eqref{2*} is stated in terms of $\sigma_A$ directly, while \eqref{2+} is stated in terms of the utility functions.

\begin{remark} \eqref{piia} is related to but different from \cite{luce1959}'s axiom of regularity, whereas \eqref{contraction} is \cite{luce1959}'s choice axiom. Replacing \eqref{piia} with regularity---which requires $\supp \sigma_A(\mathcal{M}',t) \subseteq \supp \sigma_A(\mathcal{M},t) \cap \mathcal{M}'$ whenever the intersection is nonempty---also makes Proposition \ref{p0} hold.

\end{remark}

\begin{remark}
\label{rmk2}
    If principals are not restricted to menus of deterministic outcomes but can offer menus of stochastic outcomes, then in \eqref{P1}, each $\phi_i^*$ solves the same problem except over $\max \limits_{\phi_i: T \rightarrow \Delta(O_i)}$, and both \eqref{piia} and \eqref{contraction} are imposed on $\mathcal{M}' \subseteq \mathcal{M} \in \prod_{i=1}^n 2^{\Delta(O_i)}$. 
\end{remark}

\subsection{Partial Converse}

A natural question is the converse: given a PBE, can it be recovered through the decomposition program \eqref{P1}? The answer is generally no---in particular, when $\card(T)< \card(O_i)$, some menus cannot be written as the range of a direct mechanism. The following result establishes a partial converse: within any PBE, one can find a \emph{sub-menu} for each principal that solves the decomposition program.

\begin{proposition}
\label{p2}
For any PBE $(\mathscr{M}_1,...,\mathscr{M}_n, \sigma_A)$, if $\sigma_A$ is a pure strategy that satisfies \eqref{piia}, then there exists $\{\phi_i^*(t)_{t \in T}\}_{i=1}^n$ such that for any $i, \{\phi_i^*(s)\}_{s \in T} \subseteq \mathscr{M}_i$, and $\phi_i^*$ solves:
\begin{equation}
\label{P2}
\tag{P4}
\begin{aligned}
    &\max_{\phi_i: T \rightarrow O_i}\;
    \mathbb{E}_t\!\left[
        u_i\!\left(
            \sigma_A(
                \mathscr{M}_1,
                \ldots,
                \{\phi_i(s)\}_{s \in T},
                \ldots,
                \mathscr{M}_n,
                t
            ),
            t
        \right)
    \right] \\[6pt]
    &\text{s.t. }\;
    v_i(\phi_i(t), t \mid \mathscr{M}_{-i})
    \;\ge\;
    v_i(\phi_i(t'), t \mid \mathscr{M}_{-i}),
    \quad \forall t, t'.
\end{aligned}
\end{equation}

Alternatively, if $\sigma_A$ is a mixed strategy that satisfies \eqref{contraction}, and for any i, there exists $\{o_i^t\}_{t \in T}$ such that for any $t$,
\begin{equation}
\label{3*}
\begin{aligned}
    &o_i^t \in 
    \argmax_{o_i \in \supp \sigma_A(\mathscr{M}_1,\ldots,\mathscr{M}_n,t)\big|_{O_i}}
    \int_{O_{-i}}
        u_i(o_i, o_{-i}, t)\,
        \sigma_A(\mathscr{M}_1,\ldots,\mathscr{M}_n,t)(\mathbf{d}o_{-i}\mid o_i),
    \\[6pt]
    \text{and for any } t \neq t',
    \\[-2pt]
    &\hspace{-4.5em}
    o_i^t \notin 
        \supp \sigma_A(\mathscr{M}_1,\ldots,\mathscr{M}_n,t')\big|_{O_i}
        \setminus
        \argmax_{o_i \in \supp \sigma_A(\mathscr{M}_1,\ldots,\mathscr{M}_n,t')\big|_{O_i}}
        \int_{O_{-i}}
            u_i(o_i, o_{-i}, t')\,
            \sigma_A(\mathscr{M}_1,\ldots,\mathscr{M}_n,t')(\mathbf{d}o_{-i}\mid o_i),
\end{aligned}
\end{equation}

then there exists $\{\phi_i^{**}(s)_{s \in T}\}_{i=1}^n$ such that for any $i, \{\phi_i^{**}(s)\}_{s \in T} \subseteq \mathscr{M}_i$, and $\phi_i^{**}$ solves \eqref{P2}.

\end{proposition}

The pure-strategy case is clean: under \eqref{piia}, any PBE contains a sub-menu profile that solves the decomposition. The mixed-strategy case requires the additional condition \eqref{3*}, which is an alignment condition between the agent and the principals: if an outcome of principal $i$ is chosen by type $t$ and maximizes principal $i$'s conditional payoff, then this outcome is either not chosen by other types, or when chosen by type $t'$, it also maximizes principal $i$'s conditional payoff at $t'$. In other words, the agent's randomization does not ``misassign'' outcomes across types from any individual principal's perspective.

Note that Proposition \ref{p2} recovers a sub-menu $\{\phi_i^*(s)\}_{s \in T} \subseteq \mathscr{M}_i$, not the full menu $\mathscr{M}_i$. This gap---between the sub-menu that solves the decomposition and the full equilibrium menu---reflects the possible presence of ``unused'' menu items that affect incentives but are not part of the decomposition's range. I show below that extended direct mechanisms can close this gap under certain conditions.

 \subsection{Extended Direct Mechanisms}
 \label{sec:extendeddm}

\cite{PavanCalzolari2009Sequential, PavanCalzolari2010} propose \textit{extended direct mechanisms}, in which the agent reports not only her private type but also the outcomes with other principals: $\phi_i: T \times O_{-i} \rightarrow O_i$. The richer message space allows each principal's mechanism to condition on the agent's entire experience, not just her type.

For the forward direction (from decomposition to PBE), extended mechanisms offer no additional power:
    \begin{proposition}
    \label{prop:extendeddm1}
If $\sigma_A$ is a pure strategy that satisfies condition \ref{cd1} in Definition \ref{d0} and \eqref{piia}, and for any $i, \phi_i^*$ solves
\begin{equation*}
\begin{array}{c}
\displaystyle
\max_{\phi_i: T \times O_{-i} \to O_i}
\;
\mathbb{E}_t\!\left[
u_i\!\left(
\sigma_A(
\{\phi_1^*(s,o_{-1})\}_{s\in T,\, o_{-1}\in O_{-1}},\ldots,
\{\phi_i(s,o_{-i})\}_{s\in T,\, o_{-i}\in O_{-i}},\ldots,
\{\phi_n^*(s,o_{-n})\}_{s\in T,\, o_{-n}\in O_{-n}},
t
),
t
\right)
\right]
\\[1em]
\text{s.t.} \quad \displaystyle
v_i\!\left(\phi_i(t,o_{-i}),\, t \mid \phi_{-i}^*\right)
\;\ge\;
v_i\!\left(\phi_i(t',o_{-i}'),\, t \mid \phi_{-i}^*\right),
\\[0.5em]
\forall\, t,t'\in T,\; o_{-i},o_{-i}'\in O_{-i},
\end{array}
\end{equation*}

 then $(\mathscr{M}_1,...,\mathscr{M}_n,\sigma_A)$ is a PBE, in which for any $i, \mathscr{M}_i=\{\phi_i^{*}(s,o_{-i})\}_{s \in T, o_{-i} \in O_{-i}}$.

 \end{proposition}
     
The limitation is that, in the proof, the alternative mechanism constructed to derive a contradiction is effectively a standard direct mechanism. The forward direction thus does not exploit the richer message space of extended mechanisms.

The value of extended mechanisms lies instead in the \emph{converse}: they can replicate full equilibrium menus, not just sub-menus.

     \begin{proposition}
\label{prop:extendeddm2}
     For any PBE $(\mathscr{M}_1,...,\mathscr{M}_n, \sigma_A)$, where $\sigma_A$ is a pure-strategy, and for any $i, \card(\mathscr{M}_i) \leq \card(O_{-i})$, there exists $(\phi_1^*,...,\phi_n^*)$ such that for any $i, \phi_i^*$ solves 
          \begin{equation*}
\begin{aligned}
\max_{\phi_i: T \times O_{-i} \rightarrow O_i} \quad
& \mathbb{E}_t\!\left[
u_i\!\left(
\sigma_A(\mathscr{M}_1,\ldots,\{\phi_i(s,o_{-i})\}_{s\in T,\, o_{-i}\in O_{-i}},\ldots,\mathscr{M}_n,t),
t
\right)
\right] \\
\text{s.t.} \quad
& v_i\!\left(\phi_i(t,o_{-i}),\, t \mid \mathscr{M}_{-i}\right)
\;\ge\;
v_i\!\left(\phi_i(t',o_{-i}'),\, t \mid \mathscr{M}_{-i}\right), \\
& \forall\, t,t',o_{-i}',\;
\forall\, o_{-i} \in
\supp \sigma_A(\mathscr{M}_1,\ldots,\mathscr{M}_n,t)\big|_{O_{-i}}, 
\end{aligned}
\end{equation*}

     and for any $i, \mathscr{M}_i=\{\phi_i^{*}(s,o_{-i})\}_{s \in T, o_{-i} \in O_{-i}}$.
\end{proposition}

The condition $\card(\mathscr{M}_i) \leq \card(O_{-i})$ ensures that the extended mechanism has a rich enough message space to ``encode'' each menu item through a distinct report of rivals' outcomes. When this condition holds, Proposition \ref{prop:extendeddm2} recovers the \emph{full} equilibrium menu as the range of an extended direct mechanism---closing the gap left by Proposition \ref{p2}. The economic interpretation is that extended mechanisms allow the principal to use the agent's report of rivals' outcomes as a coordination device: different reports lead to different outcomes from principal $i$, effectively replicating any menu the principal might want to offer.

\section{Conclusion}
\label{sec:conclusion}

This paper develops a decomposition methodology for common agency
games in which each principal's payoff depends on her own outcome
and the agent's type. Lemma \ref{prop:separate} reduces each
principal's best-response problem to a standard screening problem
over the agent's indirect utility; Proposition \ref{p1} gives the
compatibility condition (\eqref{2+}) under which individually
best-responding mechanisms assemble into a pure-menu PBE. The
fixed-point structure of the original game is not eliminated, but
the applications establish existence constructively: in a
quadratic-loss delegation model, equilibria feature one principal
partitioning the type space into discrete regimes while the other
adapts freely within each segment; in a bundling duopoly, the
decomposition yields market splitting and asymmetric
base-plus-upgrades structures.

Three directions remain open. Extending the decomposition to
sequential common agency (\cite{PratRustichini1998SequentialCommonAgency},
\cite{PavanCalzolari2009Sequential}) is natural, as the leader's
contract transforms the follower's indirect utility. The strategic
shields phenomenon---menu items that no type selects but that
discipline rivals' screening problems---deserves further
investigation. Finally, \cite{attar2025keeping} show that
principals can benefit from privately disclosing information
about their decision rules; extending the decomposition to such
settings would require the indirect utility to depend on private
signals.

\appendix
\addtocontents{toc}{\protect\setcounter{tocdepth}{1}}

\section{Proofs for Constructing PBE}

\subsection{Proof of Lemma \ref{prop:separate}}

\begin{proof}
    $\phi_i^*$ satisfies the IC constraints in \eqref{P3separate} implies that for any $t$, there exists $o_{-i} \in \mathscr{M}_{-i}$ such that 
    \[ (\phi_i^*(t), o_{-i}) \in \argmax \limits_{o \in (\{\phi_i^*(s)\}_{s \in T}, \mathscr{M}_{-i})} V(o,t). \]
    For any $t$, let $\sigma_A(\{\phi_i^*(s)\}_{s \in T}, \mathscr{M}_{-i},t)=(\phi_i^*(t), o_{-i})$, which does not violate condition \ref{cd1} in Definition \ref{d0} on $(\{\phi_i^*(s)\}_{s \in T}, \mathscr{M}_{-i})$. In that case,
    \[ \mathbb{E}_t[u_i(\sigma_A(\{\phi_i^*(s)\}_{s \in T}, \mathscr{M}_{-i},t),t)]=\mathbb{E}_t[u_i(\phi_i^*(t),t)]. \]
    Suppose there exists $\mathscr{M}_i'$ such that 
    \[ \mathbb{E}_t[u_i(\sigma_A(\mathscr{M}_i', \mathscr{M}_{-i},t),t)]>\mathbb{E}_t[u_i(\sigma_A(\{\phi_i^*(s)\}_{s \in T}, \mathscr{M}_{-i},t),t)],\]
    then consider $\phi_i': T \rightarrow O_i$ such that for any $t, \phi_i'(t) \in \argmax \limits_{o_i \in \supp \sigma_A(\mathscr{M}_i',\mathscr{M}_{-i},t)|_{O_i}}u_i(o_i,t)$. For any $i$,
    \begin{equation}
\label{166}
\begin{split}
    \mathbb{E}_t[u_i(\phi_i'(t),t)]
    &\geq \mathbb{E}_t[u_i(\sigma_A(\mathscr{M}_i',\mathscr{M}_{-i},t),t)] \\
    &> \mathbb{E}_t[u_i(\sigma_A(\mathscr{M}_i,\mathscr{M}_{-i},t),t)] \\
    &= \mathbb{E}_t[u_i(\phi_i^*(t),t)].
\end{split}
\end{equation}
In addition, for any $t$, there exists $o'_{-i} \in \mathscr{M}_{-i}$ such that $(\phi_i'(t), o_{-i}') \in \supp \sigma_A(\mathscr{M}_i',\mathscr{M}_{-i},t)$, which, by condition \ref{cd1} in Definition \ref{d0}, implies that $(\phi_i'(t), o_{-i}') \in \argmax \limits_{o \in \mathscr{M}_{i}' \times \mathscr{M}_{-i}} V(o,t)$.

This directly leads to
\begin{equation}
\label{177}
    \max \limits_{o \in \mathscr{M}_i' \times \mathscr{M}_{-i}} V(o,t)=v_i(\phi_i'(t),t|\mathscr{M}_{-i}) \geq v_i(\phi_i'(t'),t|\mathscr{M}_{-i}), \forall t,t'.
\end{equation}

(\ref{166}) and (\ref{177}) together imply that $\phi_i^*$ does not solve \eqref{P3separate}, a contradiction. Hence, 
\[ \{\phi_i^*(s)\}_{s \in T} \in \argmax \limits_{\tilde{\mathscr{M}}_i \in 2^{O_i}} \mathbb{E}_t[u_i(\sigma_A(\tilde{\mathscr{M}}_i, \mathscr{M}_{-i},t),t].\]
\end{proof}

\subsection{Proof of Proposition \ref{p1}}

\begin{proof}

By Lemma \ref{prop:separate}, if for any $i, \phi_i^*$ solves \eqref{P3}, then for any $i$, there exists $\sigma_A^i$ that satisfies condition \ref{cd1} in Definition \ref{d0} such that 

\begin{equation} \label{eq:ibestresponse} \{\phi_i^*(s)\}_{s \in T} \in \argmax \limits_{\tilde{\mathscr{M}}_i \in 2^{O_i}} \mathbb{E}_t[u_i(\sigma_A^i(\tilde{\mathscr{M}}_i, \prod_{j \neq i} \{\phi_j^*(s)\}_{s \in T},t),t].\end{equation}

In addition, the proof of Lemma \ref{prop:separate} shows that \eqref{eq:ibestresponse} holds if 
\[ \mathbb{E}_t[u_i(\sigma_A^i(\{\phi_i^*(s)\}_{s \in T}, \mathscr{M}_{-i},t),t)]=\mathbb{E}_t[u_i(\phi_i^*(t),t)]. \]

 For any $t$, let \begin{equation*} \sigma_A(\mathscr{M}_1,...,\mathscr{M}_n,t)=(o_1^t,...,o_n^t), \end{equation*}
where $(o_1^t,...,o_n^t)$ is the outcome profile specified by (\ref{2+}). In that case, $\sigma_A$ does not violate condition \ref{cd1} in Definition \ref{d0} on $(\mathscr{M}_1,...,\mathscr{M}_n)$, and satisfies for any $i$, $\mathbb{E}_t[u_i(\sigma_A(\mathscr{M}_1,...,\mathscr{M}_n,t),t)]= \mathbb{E}_t[u_i(\phi_i^*(t)),t)]$. Let $\sigma_A$ also satisfy condition \ref{cd1} in Definition \ref{d0} on menu profiles other than $(\mathscr{M}_1,...,\mathscr{M}_n)$ for any $t$. Combining with \eqref{eq:ibestresponse} holds for any $i$, $(\mathscr{M}_1,...,\mathscr{M}_n,\sigma_A)$ is a PBE.  

\end{proof}

\subsection{Proof of Corollary \ref{c3}}

\begin{proof}
    Suppose that there exists $t$ such that $(\phi_1^*(t),...,\phi_n^*(t)) \notin \argmax \limits_{o_i \in \{\phi_i^*(s)\}_{s \in T}, \forall i} V(o_1,...,o_n,t)$. Pick \[(\bar{o}_1,...,\bar{o}_n) \in \argmax \limits_{o_i \in \{\phi_i^*(s)\}_{s \in T}, \forall i} V(o_1,...,o_n,t).\] There exists $i$ such that $\bar{o}_i \neq \phi_i^*(t)$. In that case,
    \[ \max \limits_{o_i \in \{\phi_i^*(s)\}_{s \in T}, \forall i} V(o_1,...,o_n,t)=v_i(\bar{o}_i|\prod_{j \neq i} \{\phi_j^*(s)\}_{s \in T}) \geq v_i(\phi_i^*(t)|\prod_{j \neq i} \{\phi_j^*(s)\}_{s \in T}):=V(\phi_i^*(t), \tilde{o}_{-i},t).\]
    By the condition in Corollary \ref{c3}, $\bar{o}_i \neq \phi_i^*(t) \Rightarrow (\bar{o}_1,...,\bar{o}_n) \neq (\phi_i^*(t), \tilde{o}_{-i})$, which further implies that
    \[v_i(\bar{o}_i|\prod_{j \neq i} \{\phi_j^*(s)\}_{s \in T}) \neq v_i(\phi_i^*(t)|\prod_{j \neq i} \{\phi_j^*(s)\}_{s \in T}). \]
    Hence, $v_i(\bar{o}_i|\prod_{j \neq i} \{\phi_j^*(s)\}_{s \in T}) > v_i(\phi_i^*(t)|\prod_{j \neq i} \{\phi_j^*(s)\}_{s \in T})$, a contradiction to $\phi_i^*$ solves \eqref{P3}. Therefore, for any $t$, $(\phi_1^*(t),\ldots,\phi_n^*(t))
\in \argmax_{\substack{o_i \in \{\phi_i^*(s)\}_{s \in T}, \forall i}}
V(o_1,\ldots,o_n,t)$, which naturally satisfies \eqref{2+}. 
\end{proof}

\subsection{Proof of Corollary \ref{c0}}

\begin{proof}

For any $(\phi_1^*,...,\phi_n^*)$ such that for any $i$, $\phi_i^*$ solves \eqref{P3}. If $V$ is additive separable, then for any $t$, 
\begin{equation}
\label{eq:additiveseparable}
\begin{aligned}
(o_1^*,\ldots,o_n^*)
&\in \argmax_{\substack{o_i \in \{\phi_i^*(s)\}_{s \in T}, \forall i}}
V(o_1,\ldots,o_n,t) \\
&\Leftrightarrow
(o_1^*,\ldots,o_n^*)
\in \argmax_{\substack{o_i \in \{\phi_i^*(s)\}_{s \in T}, \forall i}}
\sum_{i=1}^n v^i(o_i,t) \\
&\Leftrightarrow
\forall i,\;
o_i^* \in \argmax_{o_i \in \{\phi_i^*(s)\}_{s \in T}} v^i(o_i,t).
\end{aligned}
\end{equation}
Note that for any $i$, $\phi_i^*$ satisfies the IC constraint in \eqref{P3} implies that
\[ \max \limits_{o_{-i} \in \prod_{j \neq i} \{\phi_j^*(s)\}_{s \in T}} V(\phi_i^*(t), o_{-i},t) \geq \max \limits_{o_{-i} \in \prod_{j \neq i} \{\phi_j^*(s)\}_{s \in T}} V(\phi_i^*(t'), o_{-i},t), \forall t,t',\]
which, by additive separability, further implies that
\[v^i(\phi_i^*(t),t) \geq v^i(\phi_i^*(t'),t), \forall t,t',\]
i.e., for any $t$, $\phi_i^*(t) \in \argmax \limits_{o_i \in \{\phi_i^*(s)\}_{s \in T}} v^i(o_i,t)$. Since this holds for any $i$, by \eqref{eq:additiveseparable}, for any $t$
\[ (\phi_1^*(t),\ldots,\phi_n^*(t))
\in \argmax_{\substack{o_i \in \{\phi_i^*(s)\}_{s \in T}, \forall i}}
V(o_1,\ldots,o_n,t), \]
which naturally satisfies \eqref{2+}.

If $V$ is weakly separable, suppose that there exists $t$ such that \begin{equation} \label{eq:notmaximizer} (\phi_1^*(t),...,\phi_n^*(t)) \notin \argmax \limits_{o_i \in \{\phi_i^*(s)\}_{s \in T}, \forall i} V(o_1,...,o_n,t).\end{equation} Pick $(\bar{o}_1,...,\bar{o}_n) \in \argmax \limits_{o_i \in \{\phi_i^*(s)\}_{s \in T}, \forall i} V(o_1,...,o_n,t)$. It is straightforward that
\[V(\bar{o}_1,\bar{o}_2,...,\bar{o}_n,t) \geq V(\phi_1^*(t), \bar{o}_2,...,\bar{o}_n,t).\]
In addition, since $V(\bar{o}_1, \bar{o}_2,...,\bar{o}_n,t) \geq V(\bar{o}_1, \phi_2^*(t),...,\bar{o}_n,t)$, by weakly separability, 
\[V(\phi_1^*(t), \bar{o}_2,...,\bar{o}_n,t) \geq V(\phi_1^*(t), \phi_2^*(t),...,\bar{o}_n,t). \]
Iterate this argument,
\begin{equation}
\label{eq:wssequence}
\begin{aligned}
V(\bar{o}_1,\bar{o}_2,\ldots,\bar{o}_n,t)
&\ge V(\phi_1^*(t),\bar{o}_2,\ldots,\bar{o}_n,t) \\
&\ge V(\phi_1^*(t),\phi_2^*(t),\ldots,\bar{o}_n,t) \\
&\ge \;\cdots\; \ge
V(\phi_1^*(t),\ldots,\phi_n^*(t),t).
\end{aligned}
\end{equation}
By \eqref{eq:notmaximizer}, at least one of the inequalities in \eqref{eq:wssequence} has to be strict. Let \[V(\phi_1^*(t),...,\phi_{i-1}^*(t), \bar{o}_i,\bar{o}_{i+1},...,\bar{o}_n,t) > V(\phi_1^*(t),...,\phi_{i-1}^*(t), \phi_i^*(t),\bar{o_{i+1}},...,\bar{o}_n,t),\]
then by weakly separability,
\[v_i(\phi_i^*(t)|\prod_{j \neq i} \{\phi_j^*(s)\}_{s \in T}) := V(\phi_i^*(t), \tilde{o}_{-i},t) < V(\bar{o}_i, \tilde{o}_{-i},t) \leq v_i(\bar{o}_i|\{\phi_{-i}^*(t)\}),\]
a contradiction to $\phi_i^*$ solves \eqref{P3}. Hence, for any $t$, $(\phi_1^*(t),\ldots,\phi_n^*(t))
\in \argmax_{\substack{o_i \in \{\phi_i^*(s)\}_{s \in T}, \forall i}}
V(o_1,\ldots,o_n,t)$, which naturally satisfies \eqref{2+}. 

\end{proof}

\subsection{Proof of Corollary \ref{c2}}

\begin{proof}
    Without loss of generality, let $\mathcal{N}'=\{1,...,n-1\}$. For any $i \in \mathcal{N}$, let $\phi_i^*(t)=\bar{o}_i, \forall t$. In that case, for any $t$, \[\argmax \limits_{o_i \in \{\phi_i^*(s)\}_{s \in T}, \forall i} V((o_1,...,o_n),t)=\argmax \limits_{o_n \in \{\phi_n^*(s)\}_{s \in T}} V((\bar{o}_1,...,\bar{o}_{n-1}, o_n),t).\] Since $(\phi_1^*,...,\phi_n^*)$ solves (\ref{P3}), 
    \begin{equation*}
        V((\bar{o}_1,...,\bar{o}_{n-1}, \phi_n^*(t)),t) \geq V((\bar{o}_1,...,\bar{o}_{n-1}, o_n),t), \forall o_n \in \{\phi_i^*(s)\}_{s \in T}, \forall t,
    \end{equation*}
    which implies that for any $t, (\bar{o}_1,...,\bar{o}_{n-1}, \phi_i^*(t)) \in \argmax \limits_{o_i \in \{\phi_i^*(s)\}_{s \in T}, \forall i} V((o_1,...,o_n),t)$, i.e., for any $t$, 
    \begin{equation*}
             (\phi_1^*(t),...,\phi_n^*(t)) \in \argmax \limits_{o_i \in \{\phi_i^*(s)\}_{s \in T}, \forall i} V((o_1,...,o_n),t),
    \end{equation*}
    so that (\ref{2+}) naturally holds. 
\end{proof}

\section{Proofs for Characterizing PBE}

\subsection{Proof of Proposition \ref{prop:p3induced}}
\begin{proof}

\textbf{``only if'': } If $(\mathscr{M}_1,...,\mathscr{M}_n, \sigma_A)$ is a \eqref{P3}-induced PBE, let it be induced be $(\phi_1^*,...,\phi_n^*)$. In that case, for any $i \in \mathcal{N}, o_i \in \mathscr{M}_i$, there exists $t^{o_i} \in T$ such that $\phi_i^*(t^{o_i})=o_i$. 

For any $i, \phi_i^*$ satisfies the IC constraints in \eqref{P3} implies that for any $t',$
\[ v_i(\phi_i^*(t^{o_i})|\{\phi^*_{-i}(s)\}_{s \in T}) \geq v_i(\phi_i^*(t')|\prod_{j \neq i} \{\phi_j^*(s)\}_{s \in T}), \]
which further implies that
\[ o_i \in \argmax \limits_{o \in \mathscr{M}_1 \times ... \times \mathscr{M}_n} V(o,t^{o_i})\big|_{O_i}.\]
With \eqref{eq:noindifference}, this directly implies that 
\[ \{o_i\}= \sigma_A(\mathscr{M}_1,...,\mathscr{M}_n,t^{o_i})\big|_{O_i}.\]

    \textbf{``if'': }For any $i$, consider $\phi_i^*: T \rightarrow O_i$ such that for any $t$, $\phi_i^*(t):=\sigma_A(\mathscr{M}_1,...,\mathscr{M}_n,t)\big|_{O_i}$, which is well-defined since $\sigma_A$ is a pure strategy. Clearly, for any $i, \{\phi_i^*(s)\}_{s \in T} \subseteq \mathscr{M}_i$. Meanwhile, for any $o_i \in \mathscr{M}_i$, the existence of $t^{o_i}$ such that $o_i \in \sigma_A(\mathscr{M}_1,...,\mathscr{M}_n,t^{o_i})\big|_{O_i}$ and the fact that $\sigma_A$ is a pure-strategy together imply that $o_i=\phi_i^*(t^{o_i})$, so that $\mathscr{M}_i \subseteq \{\phi_i^*(s)\}_{s \in T}$, which further implies that $\{\phi_i^*(s)\}_{s \in T} =\mathscr{M}_i$. 

    Claim that for $\{\phi_i^*\}_{i \in \mathcal{N}}$ defined above, for any $i, \phi_i^*$ solves \eqref{P3} and \eqref{2+} is satisfied. The latter is not hard to see as for any $t$,
    \[ (\phi_1^*(t),...,\phi_n^*(t))=\sigma_A(\mathscr{M}_1,...,\mathscr{M}_n,t).\]

    Regarding the former, suppose there exists $i$ such that $\phi_i^*$ does not solve \eqref{P3}, then there exists $\phi_i': T \rightarrow O_i$ such that $\mathbb{E}_t[u_i(\phi_i'(t),t)]>\mathbb{E}_t[u_i(\phi_i^*(t),t)]$, and $\phi_i'$ satisfies the IC constraint in \eqref{P3}. Consider $\mathscr{M}_i'=\{\phi_i'(s)\}_{s \in T}$, then 
\begin{equation*}\label{eq:deviation-chain}
\begin{aligned}
\mathbb{E}_t\!\Big[u_i\big(\sigma_A(\mathscr{M}_1,\ldots,\mathscr{M}_i',\ldots,\mathscr{M}_n,t),t\big)\Big]
&=
\mathbb{E}_t\!\Big[u_i\big(\sigma_A(\mathscr{M}_1,\ldots,\{\phi_i'(t)\},\ldots,\mathscr{M}_n,t),t\big)\Big] \\
&=
\mathbb{E}_t\!\Big[u_i\big(\phi_i'(t),t\big)\Big]
>
\mathbb{E}_t\!\Big[u_i\big(\phi_i^*(t),t\big)\Big] \\
&=
\mathbb{E}_t\!\Big[u_i\big(\sigma_A(\mathscr{M}_1,\ldots,\mathscr{M}_n,t),t\big)\Big],
\end{aligned}
\end{equation*}
where the first ``$=$'' is because for any $t$, $\phi_i'(t) \in \argmax \limits_{o \in \mathscr{M}_1 \times ... \times \mathscr{M}_i' \times ... \times \mathscr{M}_n} V(o,t) \big|_{O_i}$ implies that, under \eqref{eq:noindifference}, for any $t$,
\[ \{\phi_i'(t)\} \in \argmax \limits_{o \in \mathscr{M}_1 \times ... \times \mathscr{M}_i' \times ... \times \mathscr{M}_n} V(o,t) \big|_{O_i}, \]
which further implies that for any $t, \sigma_A(\mathscr{M}_1,...,\mathscr{M}_i',...,\mathscr{M}_n,t)|_{O_i}=\{\phi_i'(t)\}$. Meanwhile, the last ``$=$'' is because for any $i$, $\{\phi_i^*(s)\}_{s \in T}=\mathscr{M}_i$. In that case, $\mathscr{M}_i'$ is a profitable deviation for principal $i$, a contradiction to $(\mathscr{M}_1,...,\mathscr{M}_n,\sigma_A)$ is a PBE. Therefore, for any $i, \phi_i^*$ solves \eqref{P3}, so that $(\mathscr{M}_1,...,\mathscr{M}_n,\sigma_A)$ is a \eqref{P3}-induced PBE.

\end{proof}

\subsection{Proof of Proposition \ref{prop:PO}}
I first show the proof of the first statement. 
\begin{proof}
    Let $T:=\{\mathfrak{t}\}$. Suppose there exists a PBE $(\mathscr{M}_1,...,\mathscr{M}_n,\sigma_A)$ such that it is Pareto optimal among all \eqref{P3}-induced PBE but is not Pareto optimal. By Definition \ref{def:paretooptimal}, there exists PBE $(\mathscr{M}_1',...,\mathscr{M}_n',\sigma_A')$ such that for any $i \in \mathcal{N}$,
    \[u_i(\sigma_A'(\mathscr{M}_1',...,\mathscr{M}_n',\mathfrak{t}), \mathfrak{t}) \geq u_i(\sigma_A(\mathscr{M}_1,...,\mathscr{M}_n,\mathfrak{t}), \mathfrak{t}),\]
    and there exists $j \in \mathcal{N}$ such that the inequality holds strictly. 
    
    For any $i$, let $\phi_i^*(\mathfrak{t})$ satisfy
    \[ \phi_i^*(\mathfrak{t}) \in \argmax \limits_{o_i \in \supp \sigma_A'(\mathscr{M}_1',...,\mathscr{M}_n',\mathfrak{t})\big|_{O_i}} u_i(o_i, \mathfrak{t}).\footnote{It can be further observed that for any $\tilde{o}_i \in \supp \sigma_A'(\mathscr{M}_1',...,\mathscr{M}_n',\mathfrak{t}), \tilde{o}_i \in \argmax \limits_{o_i \in \supp \sigma_A'(\mathscr{M}_1',...,\mathscr{M}_n',\mathfrak{t})\big|_{O_i}} u_i(o_i, \mathfrak{t})$.}\]
    Claim that $(\{\phi_1^*(\mathfrak{t})\},...,\{\phi_n^*(\mathfrak{t})\})$ is a \eqref{P3}-induced menu profile, i.e., for any $i, \phi_i^*$ solves \eqref{P3}, and \eqref{2+} is satisfied. The IC constraint is trivially satisfied as $\card(T)=1$. \eqref{2+} is also trivially satisfied as each principal plays a singleton menu. Suppose there exists $i$ and $\phi_i'$ such that
    \[ u_i(\phi_i'(\mathfrak{t}), \mathfrak{t}) > u_i(\phi_i^*(\mathfrak{t}), \mathfrak{t}),\]
    then $(\mathscr{M}_1',...,\mathscr{M}_n',\sigma_A')$ cannot be a PBE, as principal $i$ has a profitable deviation $\{\phi_i'(\mathfrak{t})\}$:
    \[ u_i(\sigma_A'(\mathscr{M}_1',...,\{\phi_i'(\mathfrak{t})\},...,\mathscr{M}_n',\mathfrak{t}),\mathfrak{t})=u_i(\phi_i'(\mathfrak{t}),\mathfrak{t})>u_i(\phi_i^*(\mathfrak{t}), \mathfrak{t}) \geq u_i(\sigma_A'(\mathscr{M}_1',...,\mathscr{M}_n',\mathfrak{t}), \mathfrak{t}),\]
    which is a contradiction. Hence, $(\{\phi_1^*(\mathfrak{t})\},...,\{\phi_n^*(\mathfrak{t})\})$ is a \eqref{P3}-induced menu profile. However, for any $\tilde{\sigma_A}$ that makes $(\{\phi_1^*(\mathfrak{t})\},...,\{\phi_n^*(\mathfrak{t})\},\tilde{\sigma_A})$ a PBE, for any $i$,
     \[u_i(\tilde{\sigma_A}(\{\phi_1^*(\mathfrak{t})\},...,\{\phi_n^*(\mathfrak{t})\},\mathfrak{t}), \mathfrak{t}) \geq u_i(\sigma_A'(\mathscr{M}_1',...,\mathscr{M}_n',\mathfrak{t}), \mathfrak{t}) \geq u_i(\sigma_A(\mathscr{M}_1,...,\mathscr{M}_n,\mathfrak{t}), \mathfrak{t}),\]
     and there exists $j \in \mathcal{N}$ such that the inequality holds strictly, which is a contradiction to $(\mathscr{M}_1,...,\mathscr{M}_n,\sigma_A)$ is Pareto optimal among all \eqref{P3}-induced PBE. Therefore, the proposition holds. 
\end{proof}

I proceed to show the proof of the second statement. 

\begin{proof}
It is not hard to observe that in any PBE, it has to be the case that for any $t$
\[ \supp \sigma_A(\cdot,\cdot,t) \bigg|_{O_i}=\{o_i^*\},\]
otherwise $\{o_i^*\}$ is a profitable deviation for principal $i$. With this proposition, I focus on PBE where principal $i$ plays $\{o_i^*\}$. 

Suppose there exists a type distribution under which there exists a PBE $(\{o_i^*\}, \mathscr{M}_{-i}, \sigma_A)$ that is Pareto optimal among all \eqref{P3}-induced PBE, but is Pareto dominated by another, in particular, Pareto optimal PBE $(\{o_i^*\}, \mathscr{M}_{-i}',\sigma_A')$, then it has to be the case that 
\[ \mathbb{E}_t[u_{-i}(\sigma_A'(\{o_{i}^*\}, \mathscr{M}_{-i}',t),t)] > \mathbb{E}_t[u_{-i}(\sigma_A(\{o_{i}^*\}, \mathscr{M}_{-i},t),t)].\]

Consider $\phi_{-i}':T \rightarrow O_{-i}$ such that for any $t$,
\[\phi_{-i}'(t) \in \argmax \limits_{o_{-i} \in \supp \sigma_A'(\{o_i^*\}, \mathscr{M}_{-i}',t)\big|_{O_{-i}}} u_{-i}(o_{-i},t).\]

Since for any $t, \phi_{-i}'(t) \in \supp \sigma_A'(\{o_i^*\}, \mathscr{M}_{-i}',t),$
\[ V(o_i^*, \phi_{-i}'(t),t) \geq V(o_i^*, \phi_{-i}'(t'),t), \forall t,t'. \]

Suppose $\phi_{-i}'$ does not solve \eqref{P3}, then let $\phi_{-i}''$ be a solution of \eqref{P3}. It satisfies
\[ V(o_i^*, \phi_{-i}''(t),t) \geq V(o_i^*, \phi_{-i}''(t'),t), \forall t,t', \]
and 
\[ \mathbb{E}_t[u_{-i}(\phi_{-i}''(t),t)] > \mathbb{E}_t[u_{-i}(\phi_{-i}'(t),t)].\]
Consider $\sigma_A''$ such that for any $t$,
\[ \sigma_A''(\{o_i^*\}, \{\phi_{-i}''(t)\}_{t \in T}, t)=(o_i^*, \phi_{-i}''(t)),\]
and $\sigma_A''$ satisfies condition \ref{cd1} in Definition \ref{d0} on other menu profiles. By Proposition \ref{p1}, $(\{o_i^*\}, \{\phi_{-i}''(t)\}_{t \in T}, \sigma_A'')$ is a PBE, with 
\[
\begin{aligned}
\mathbb{E}_t\!\Big[u_{-i}\big(\sigma_A''(\{o_i^*\}, \{\phi_{-i}''(t)\}_{t\in T}, t),\, t\big)\Big]
&= \mathbb{E}_t\!\Big[u_{-i}\big(\phi_{-i}''(t),\, t\big)\Big] \\
&> \mathbb{E}_t\!\Big[u_{-i}\big(\phi_{-i}'(t),\, t\big)\Big] \\
&\ge \mathbb{E}_t\!\Big[u_{-i}\big(\sigma_A'(\{o_i^*\}, \mathscr{M}_{-i}', t),\, t\big)\Big],
\end{aligned}
\]
a contradiction to $(\{o_i^*\}, \mathscr{M}_{-i}',\sigma_A')$ is Pareto optimal. Hence, $\phi_{-i}'$ solves \eqref{P3}. Consider $\sigma_A'''$ such that for any $t$,
\[ \sigma_A'''(\{o_i^*\}, \{\phi_{-i}'(t)\}_{t \in T}, t)=(o_i^*, \phi_{-i}'(t)),\]
and $\sigma_A'''$ satisfies condition \ref{cd1} in Definition \ref{d0} on other menu profiles. In that case, $(\{o_i^*\}, \{\phi_{-i}'(t)\}_{t \in T}, \sigma_A''')$ is a \eqref{P3}-induced PBE,
and 
\[
\begin{aligned}
\mathbb{E}_t\!\Big[u_{-i}\big(\sigma_A'''(\{o_i^*\}, \{\phi_{-i}'(t)\}_{t\in T}, t),\, t\big)\Big]
&= \mathbb{E}_t\!\Big[u_{-i}\big(\phi_{-i}'(t),\, t\big)\Big] \\
&\ge \mathbb{E}_t\!\Big[u_{-i}\big(\sigma_A'(\{o_i^*\}, \mathscr{M}_{-i}', t),\, t\big)\Big] \\
&> \mathbb{E}_t\!\Big[u_{-i}\big(\sigma_A(\{o_i^*\}, \mathscr{M}_{-i}, t),\, t\big)\Big],
\end{aligned}
\]
a contradiction to $(\{o_i^*\}, \mathscr{M}_{-i}, \sigma_A)$ is Pareto optimal among all \eqref{P3}-induced PBE. Therefore, $(\{o_i^*\}, \mathscr{M}_{-i}, \sigma_A)$ is not Pareto dominated.

If \eqref{eq:noindifference} is further satisfied, then for any \eqref{P3}-induced PBE, let it be induced by $(\phi_i \equiv o_i^*, \phi_{-i})$. Suppose it is still Pareto dominated by $(\{o_i^*\}, \mathscr{M}_{-i}', \sigma_A')$. Following the notations above, 
\[
\begin{aligned}
\mathbb{E}_t\!\Big[u_{-i}\big(\phi_{-i}'(t),\, t\big)\Big]
&> \mathbb{E}_t\!\Big[u_{-i}\big(\sigma_A(\{o_i^*\}, \mathscr{M}_{-i}, t),\, t\big)\Big] \\
&= \mathbb{E}_t\!\Big[u_{-i}\big(\sigma_A(\{o_i^*\}, \{\phi_{-i}(t)\}, t),\, t\big)\Big] \\
&= \mathbb{E}_t\!\Big[u_{-i}\big(\phi_{-i}(t),\, t\big)\Big],
\end{aligned}
\]
where the second ``$=$'' is by \eqref{eq:noindifference} and the fact that $\phi_{-i}'(t) \in \argmax \limits_{o_{-i} \in \supp \sigma_A'(\{o_i^*\}, \mathscr{M}_{-i}',t)|_{O_{-i}}}$. This is a contradiction to $\phi_{-i}$ solves \eqref{P3}. Therefore, any \eqref{P3}-induced PBE is not Pareto dominated.  
\end{proof}

\section{Discussions and Proofs for Delegation}

\subsection{A Type-Independent-Peaked Principal}
The Lagrangian for the relaxed problem is
\[
\mathcal{L}(o_{-i}, \Lambda) :=
\int_{\underline{t}}^{\bar{t}}
\Bigl(-r_{-i}(t)(o_{-i}(t)-o^*_{-i}(t))^2 f(t)
- \kappa f(t)\, \Gamma(t)\Bigr)\mathbf{d}t
+ \int_{\underline{t}}^{\bar{t}} \Gamma(t)\,\mathbf{d}\Lambda(t),
\]
where $\Gamma(t) := (t-o_i^*)o_{-i}(t) - \frac{o_{-i}(t)^2}{2}
+ \frac{o_{-i}(\underline{t})^2}{2} - (\underline{t}-o_i^*)o_{-i}(\underline{t})
- \int_{\underline{t}}^t o_{-i}(s)\,\mathbf{d}s$ is the envelope
slack, and $\Lambda$ is a right-continuous, increasing function
satisfying $\Lambda(\bar{t}) - \Lambda(\underline{t})
= \kappa(F(\bar{t}) - F(\underline{t}))$.
Integration by parts yields
\begin{equation*}
\begin{aligned}
\mathcal{L}(o_{-i}, \Lambda)
={}&
\int_{\underline{t}}^{\bar{t}} \Bigl(
-r_{-i}(t)(o_{-i}(t)-o^*_{-i}(t))^2 f(t)
- o_{-i}(t)\,[\kappa F(t)-\Lambda(t)]
- \kappa f(t)\Bigl[(t-o_i^*)o_{-i}(t)
  -\tfrac{o_{-i}(t)^2}{2}\Bigr]
\Bigr)\mathbf{d}t \\
&+ \int_{\underline{t}}^{\bar{t}}
\Bigl((t-o_i^*)o_{-i}(t)-\tfrac{o_{-i}(t)^2}{2}\Bigr)
\mathbf{d}\Lambda(t)
+ \Bigl(\int_{\underline{t}}^{\bar{t}} o_{-i}(t)\,\mathbf{d}t\Bigr)
[\kappa F(\bar{t})-\Lambda(\bar{t})].
\end{aligned}
\end{equation*}
The pointwise second derivative of the integrand with respect to
$o_{-i}(t)$ is $f(t)[\kappa - 2r_{-i}(t)] < 0$, so $\mathcal{L}$
is concave in $o_{-i}$. By the sufficient conditions in
\cite{KartikKleinerVanWeelden2021Delegation}, $o_{-i}$ solves the
relaxed problem if $\partial \mathcal{L}(o_{-i},
\bar{o}_{-i} - o_{-i}, \Lambda) \leq 0$ for all
$\bar{o}_{-i}: T \to O_{-i}$, where the Gateaux differential, conditional on $\kappa F(\bar{t})-\Lambda(\bar{t})=0$, is
\begin{equation*}
\begin{split}
\partial \mathcal{L}(o_{-i}, \bar{o}_{-i}, \Lambda)
={}&
\int_{\underline{t}}^{\bar{t}}
\Bigl[-2r_{-i}(t)(o_{-i}(t)-o^*_{-i}(t))f(t)
- \kappa F(t)+\Lambda(t)\Bigr]
\bar{o}_{-i}(t)\,\mathbf{d}t \\
&+ \int_{\underline{t}}^{\bar{t}}
[t-o_i^*-o_{-i}(t)]\,
\bar{o}_{-i}(t)\,
\mathbf{d}[\Lambda(t)-\kappa F(t)].
\end{split}
\end{equation*}

 I proceed to demonstrate the proof of Proposition \ref{prop:singlepeaked}.

\begin{proof}
Regarding the Full Delegation result, consider \[
\Lambda(t)=
\begin{cases}
\kappa F(t) + 2r_{-i}(t)\bigl(t - o_i^* - o_{-i}^*(t)\bigr) f(t), & t \in [\underline{t}, \bar{t}) \\
\kappa F(t), & t = \bar{t}.
\end{cases}
\] By the conditions specified in Proposition \ref{prop:singlepeaked}, $\Lambda$ is right-continuous and increasing on $[\underline{t}, \bar{t}]$. In addition, $\Lambda(\bar{t})-\kappa F(\bar{t})=\Lambda(\underline{t})-\kappa F(\underline{t})=0$. Consider $o_{-i}(t)=t-o_i^*, \forall t$, then under the $\Lambda(t)$ constructed above, for any $\bar{o}_{-i}, \partial \mathcal{L}(o_{-i}, \bar{o}_{-i}-o_{-i}, \Lambda) \leq 0$. Hence $o_{-i}(t)=t-o_i^*, \forall t$ is the solution to the relaxed problem. As for the envelope condition, 
\[ \dfrac{(t-o_i^*)^2}{2}-\dfrac{(\underline{t}-o_i^*)^2}{2}=\int_{\underline{t}}^t s-o_i^* \mathbf{d}s.\]
Therefore, $o_{-i}(t)=t-o_i^*, \forall t$ is a solution to the original problem. By Corollary \ref{c2}, $(\{o_i^*\}, \{s-o_i^*\}_{s \in T})$ is a \eqref{P3}-induced menu profile.

Regarding the No Compromise result, for any $t \in [\underline{t}, \bar{t}]$, let $\Lambda(t)=\kappa F(t)$, which is naturally right-continuous and increasing. Consider $o_{-i}(t)=o_{-i}^*(t), \forall t$, then under the constructed $\Lambda(t)$, for any $\bar{o}_{-i}, \partial \mathcal{L}(o_{-i}, \bar{o}_{-i}-o_{-i}, \Lambda) \leq 0$. Hence $o_{-i}(t)=o_{-i}^*(t), \forall t$ is the solution to the relaxed problem. The envelope condition is satisfied by the condition specified in Proposition \ref{prop:singlepeaked}. Therefore, $o_{-i}(t)=o^*_{-i}(t), \forall t$ is a solution to the original problem. Again by Corollary \ref{c2}, $(\{o_i^*\}, \{o^*_{-i}(s)\}_{s \in T})$ is a \eqref{P3}-induced menu profile.
\end{proof}

\subsection{The General Case}
\label{appendix:finitemenus}

I first prove Proposition \ref{prop:delegationeasy} without the use of the Lagrangian approach.

\begin{proof}
It is straightforward that for any $i, \phi_i^*$ solves \eqref{P3}. In particular, $\phi_i^*$ solves the unconstrained \eqref{P3}, so that it suffices to check that $\phi_i^*$ is feasible. Since for any $t$,

  \[ (\phi_i^*(t), \phi_{-i}^*(t)) \in \argmax \limits_{o \in O_1 \times O_2} V(o,t), \]

  both feasibility and \eqref{2+} directly follow. Hence, $(\{o_i^*(s)\}_{s \in T}, \{o^*_{-i}(s)\}_{s \in T})$ is a \eqref{P3}-induced menu profile. 

\end{proof}

Given $\{\phi_i^*(s)\}_{s \in T}$, let $S_i := \{\phi_i^*(s)\}_{s \in T}$ be a compact
subset of $\mathbb{R}$. For the quadratic utility $V = -(t - o_i - o_{-i})^2$,
the agent's optimal selection from $S_i$ satisfies
\[
\tilde{o}_i(t, t' \mid o_{-i})
= \Pi_{S_i}\bigl(t - o_{-i}(t')\bigr),
\]
where $\Pi_{S_i}: \mathbb{R} \to S_i$ denotes the nearest-point projection onto $S_i$,
defined by
\[
\Pi_{S_i}(z) := \argmin_{s \in S_i} |s - z|.
\]
Since $S_i$ is a compact subset of $\mathbb{R}$, the minimum is attained for every
$z \in \mathbb{R}$. The minimizer is unique except when $z$ is equidistant from
two distinct points of $S_i$; in $\mathbb{R}$, this occurs only at midpoints of
gaps in $S_i$, which form an at most countable set. At such points, I adopt the
convention that $\Pi_{S_i}$ selects the smaller element.\footnote{The choice of
tie-breaking convention is immaterial for the analysis, as the set of
non-uniqueness points has measure zero and does not affect any integral.}

The projection $\Pi_{S_i}$ is monotone non-decreasing. By Lebesgue's
theorem, $\Pi_{S_i}$ is therefore differentiable almost everywhere, with
$\Pi_{S_i}'(z) \in \{0, 1\}$ for a.e.\ $z$:
\begin{itemize}
    \item if $z \in \mathrm{int}(S_i)$, then $\Pi_{S_i}(z) = z$ in a neighborhood of
    $z$, so $\Pi_{S_i}'(z) = 1$;
    \item if $z \notin S_i$, then $\Pi_{S_i}(z)$ is constant in a neighborhood of $z$
    (equal to the nearest point of $S_i$), so $\Pi_{S_i}'(z) = 0$.\footnote{In particular, when $S_i$ is a finite set, $\mathrm{int}(S_i) = \emptyset$, so $\Pi_{S_i}' = 0$ a.e. When $S_i$ is an interval $[\underline{s}, \bar{s}]$,
this recovers the projection
$\Pi_{S_i}(z) = \min\{\bar{s}, \max\{\underline{s}, z\}\}$, and the three regions
$\{z < \underline{s}\}$, $\{z \in (\underline{s}, \bar{s})\}$,
$\{z > \bar{s}\}$ correspond to $\Pi_{S_i}' = 0, 1, 0$ respectively.}
\end{itemize}

For each fixed $t$, the map $y \mapsto \Pi_{S_i}(t - y)$ is monotone
non-increasing, so by Lebesgue's theorem it is differentiable at a.e.\ $y$.
In particular, for a.e.\ $t$, the directional derivative
\[
\delta \tilde{o}_i(t,t)
:= \lim_{\varepsilon \to 0}
\frac{\tilde{o}_i(t, t \mid o_{-i} + \varepsilon h)
      - \tilde{o}_i(t, t \mid o_{-i})}{\varepsilon}
= -\Pi_{S_i}'(t - o_{-i}(t)) \cdot h(t)
\]
exists.

Let
\[
\Xi(t) := \Psi_t - \Psi_{\underline{t}}
- \int_{\underline{t}}^t
\bigl(\tilde{o}_i(s, s \mid o_{-i}) + o_{-i}(s)\bigr)\, \mathbf{d}s,
\qquad
\Psi_t := t\bigl(\tilde{o}_i(t,t \mid o_{-i}) + o_{-i}(t)\bigr)
- \frac{\bigl(\tilde{o}_i(t,t \mid o_{-i}) + o_{-i}(t)\bigr)^2}{2}.
\]
The relaxed problem penalizes violations of $\Xi(t) \geq 0$ in the objective:
\[
\max_{o_{-i}: T \to O_{-i}} \;
\int_{\underline{t}}^{\bar{t}}
\Bigl(
-r_{-i}(t)\bigl(o_{-i}(t) - o_{-i}^*(t)\bigr)^2
- \kappa\, \Xi(t)
\Bigr)\, \mathbf{d}F(t),
\qquad \text{s.t.} \quad \Xi(t) \geq 0, \; \forall t,
\]
where $\kappa := \frac{\min_t r_{-i}(t)}{2}$. The choice of $\kappa$ ensures
that the objective is a concave functional of $o_{-i}$: the second-order
contribution from perturbing $o_{-i}(t)$ is proportional to
$2\kappa(1 - \Pi_{S_i}'(t - o_{-i}(t))) - 2r_{-i}(t) < 0$.

The Lagrangian is
$\mathcal{L}(o_{-i}, \Lambda)
:= \int_{\underline{t}}^{\bar{t}}
\bigl(-r_{-i}(t)(o_{-i}(t) - o_{-i}^*(t))^2 f(t)
- \kappa f(t)\, \Xi(t)\bigr)\, \mathbf{d}t
+ \int_{\underline{t}}^{\bar{t}} \Xi(t)\, \mathbf{d}\Lambda(t)$,
where $\Lambda$ is a right-continuous, increasing function with
$\Lambda(\bar{t}) - \Lambda(\underline{t})
= \kappa(F(\bar{t}) - F(\underline{t}))$.
Integration by parts and the Gateaux differential in the direction
$\bar{o}_{-i}$ yield
\begin{equation*}
\begin{aligned}
\partial \mathcal{L}(o_{-i}, \bar{o}_{-i}, \Lambda)
= {} &
\int_{\underline{t}}^{\bar{t}}
\Bigl[
-2 r_{-i}(t)\bigl(o_{-i}(t) - o_{-i}^*(t)\bigr) f(t)
- (1 - \Pi_{S_i}')\bigl(\kappa F(t) - \Lambda(t)\bigr)
\Bigr]\,
\bar{o}_{-i}(t)\, \mathbf{d}t
\\[0.5em]
&\quad
+ \int_{\underline{t}}^{\bar{t}}
(1 - \Pi_{S_i}')
\Bigl(
t - o_{-i}(t) - \tilde{o}_i(t,t \mid o_{-i})
\Bigr)\,
\bar{o}_{-i}(t)\,
\mathbf{d}\bigl[\Lambda(t) - \kappa F(t)\bigr]
\\[0.5em]
&\quad
- \bigl(\Lambda(\bar{t}) - \kappa F(\bar{t})\bigr)
\int_{\underline{t}}^{\bar{t}}
(1 - \Pi_{S_i}')\, \bar{o}_{-i}(t)\, \mathbf{d}t.
\end{aligned}
\end{equation*}
where $\Pi_{S_i}' := \Pi_{S_i}'(t - o_{-i}(t))$.
By the sufficient conditions in \cite{KartikKleinerVanWeelden2021Delegation},
if $\partial \mathcal{L}(o_{-i}, \bar{o}_{-i} - o_{-i}, \Lambda) \leq 0$
for all $\bar{o}_{-i}: T \to O_{-i}$, then $o_{-i}$ solves the relaxed problem.

I proceed to show the proof of Proposition \ref{prop:delegationhard}. 

\begin{proof}

Regarding principal $2$'s problem, let
\[
o_2(t)=t-o_1^x,\qquad \forall t\in [t^{x-1},t^x),\; x\in\{1,\dots,K+1\}.
\]
Consider
\[
\Lambda(t)=
\begin{cases}
\kappa F(t)+2r_2(t)\bigl(t-o_1^x-o_2^*(t)\bigr)f(t),
& t\in [t^{x-1},t^x),\ x\in\{1,\dots,K+1\},\\[0.4em]
\kappa F(t), & t=\bar t.
\end{cases}
\]
By the first condition in the proposition, $\Lambda$ is right-continuous and increasing on each
interval $[t^{x-1},t^x)$. By the second condition, $\Lambda$ has no downward jump at any cutoff
$t^x$, $x\in\{1,\dots,K\}$. Hence $\Lambda$ is right-continuous and increasing on $[\underline t,\bar t]$.
Moreover, by the third condition and the construction, $\Lambda(\underline t)-\kappa F(\underline t)=\Lambda(\bar t)-\kappa F(\bar t)=0.$

Now let $S_1 = \{o_1^1, \dots, o_1^{K+1}\}$. Since $S_1$ is a finite set, its derivative is $\Pi_{S_1}' = 0$ almost everywhere. For any type $t \in [t^{x-1}, t^x)$, we can evaluate:
\[ 
\tilde o_1(t,t\mid o_2) = \Pi_{S_1}\bigl(t-o_2(t)\bigr) = \Pi_{S_1}(o_1^x) = o_1^x. 
\]
Therefore, the combined allocation simplifies to $t - o_2(t) - \tilde o_1(t,t\mid o_2) = 0$ for all $t \in [\underline t, \bar t)$.

When substituting $o_2$ into the Gateaux differential, the second integral vanishes because $t - o_2(t) - \tilde o_1(t,t\mid o_2) = 0$ almost everywhere, and the boundary term vanishes because $\Lambda(\bar t) - \kappa F(\bar t) = 0$. For the first integral, recalling that $\Pi_{S_1}' = 0$ a.e., its coefficient becomes:
\[ 
-2r_2(t)\bigl(o_2(t) - o_2^*(t)\bigr)f(t) - \bigl(\kappa F(t) - \Lambda(t)\bigr). 
\]
However, on the interval $[t^{x-1}, t^x)$, we know that $o_2(t) - o_2^*(t) = t - o_1^x - o_2^*(t)$. By the definition of our multiplier $\Lambda$, this implies the coefficient is exactly zero. Hence, for any alternative strategy $\bar o_2: T \to O_2$, we have $\partial \mathcal{L}(o_2, \bar o_2 - o_2, \Lambda) \le 0$. By the sufficient conditions in \cite{KartikKleinerVanWeelden2021Delegation}, this confirms $o_2$ solves the relaxed problem.

Turning to the envelope condition, for every $t \in [\underline t, \bar t]$, we have $\tilde o_1(t,t\mid o_2) + o_2(t) = t$ whenever $t \in [t^{x-1}, t^x)$. We can then calculate the utility parameter:
\[ 
\Psi_t = t\bigl(\tilde o_1(t,t\mid o_2) + o_2(t)\bigr) - \frac{\bigl(\tilde o_1(t,t\mid o_2) + o_2(t)\bigr)^2}{2} = \frac{t^2}{2}. 
\]
Consequently, the utility difference from the lowest type unfolds as:
\[ 
\Psi_t - \Psi_{\underline t} = \frac{t^2 - \underline t^{\,2}}{2} = \int_{\underline t}^t s \,ds = \int_{\underline t}^t \bigl(\tilde o_1(s,s\mid o_2) + o_2(s)\bigr) \,ds. 
\]
It immediately follows that the constraint function $\Xi(t) = \Psi_t - \Psi_{\underline t} - \int_{\underline t}^t \bigl(\tilde o_1(s,s\mid o_2) + o_2(s)\bigr) \,ds = 0$ for all $t$. Hence $o_2$ also solves the original problem.

As for principal 1, by construction she offers her pointwise ideal outcome $o_1(t) = o_1^x = o_1^*(t)$ for all $t \in [t^{x-1}, t^x)$. She therefore attains her maximal payoff at every type, successfully solving her decomposition program.

Finally, for every $t \in [t^{x-1}, t^x)$, the total allocation $o_1(t) + o_2(t) = o_1^x + (t - o_1^x) = t$ gives the agent her exact bliss point. Therefore, under the corresponding optimal agent strategy, the menu profile
\[ 
\left( \{o_1^1, \dots, o_1^{K+1}\}, \bigcup_{x=1}^{K+1} \{s - o_1^x\}_{s \in [t^{x-1}, t^x)} \right) 
\]
constitutes a \eqref{P3}-induced PBE menu profile.

\end{proof}

\section{Proof for Bundling: Proposition \ref{t1}}

\begin{proof}

I first show the ``if'' direction, i.e., any $(\{(b_1,p_1),...,(b_l,p_l)\}, \{(b_1',p_1'),...,(b_m',p_m')\})$ that satisfies \eqref{eq:marketsplitting} is a \eqref{P4'}-induced menu profile. Let $\mathscr{M}_2:=\{(b_1',p_1'),...,(b_m',p_m')\}$, then the \eqref{P4'} for principal 1 to solve for is
\begin{equation*}
\begin{aligned}
\max_{\phi_1: T \to O_1 \cup \{quit\}} \quad
& \mathbb{E}_t\!\left[
v_1(\phi_1(t),t \mid \mathscr{M}_2)
-\frac{1-F(t)}{f(t)}(v_1)_t(\phi_1(t),t \mid \mathscr{M}_2)
\right] \\
\text{s.t.} \quad
& v_1(\phi_1(t),t \mid \mathscr{M}_2)
\ge v_1(\phi_1(t'),t \mid \mathscr{M}_2),
\qquad \forall\, t,t', \\
& v_1(\phi_1(t),t \mid \mathscr{M}_2) \ge 0,
\qquad \forall\, t.
\end{aligned}
\end{equation*}

For any $t$, let $\phi_1(t):=(b^1(t),p^1(t)) \in \mathcal{B} \times [0, \bar{p}]$ if $\phi_1(t) \neq quit$.\footnote{If $\phi_1(t)=quit$, then $b^1(t)=p^1(t)=quit$, and \[\max \limits_{(b',p') \in \mathscr{M}_2}U((quit,b'),t)
-\frac{1-F(t)}{f(t)}\frac{\partial \max \limits_{(b',p') \in \mathscr{M}_2}U((quit,b'),t)}{\partial t}-\dfrac{1}{2}\max_{(b,b')} U\!\bigl((b,b'),t_*\bigr)=\max \limits_{(b',p') \in \mathscr{M}_2}U((quit,b'),t)-quit-\dfrac{1}{2}\max_{(b,b')} U\!\bigl((b,b'),t_*\bigr):=0.\]} The problem can be further written out as
\begin{equation*}
\begin{aligned}
\max_{\phi_1: T \to O_1 \cup \{quit\}} \quad
& \mathbb{E}_t\!\left[
\max \limits_{(b',p') \in \mathscr{M}_2}U((b^1(t),b'),t)
-\frac{1-F(t)}{f(t)}\frac{\partial \max \limits_{(b',p') \in \mathscr{M}_2}U((b^1(t),b'),t)}{\partial t}-\dfrac{1}{2}\max_{(b,b')} U\!\bigl((b,b'),t_*\bigr)
\right] \\
\text{s.t.} \quad
& \max \limits_{(b',p') \in \mathscr{M}_2}U((b^1(t),b'),t)-p^1(t)-\dfrac{1}{2}\max_{(b,b')} U\!\bigl((b,b'),t_*\bigr)
\\ &\ge \max \limits_{(b',p') \in \mathscr{M}_2}U((b^1(t'),b'),t)-p^1(t')-\dfrac{1}{2}\max_{(b,b')} U\!\bigl((b,b'),t_*\bigr),
\qquad \forall\, t,t', \\
& \max \limits_{(b',p') \in \mathscr{M}_2}U((b^1(t),b'),t)-p^1(t)-\dfrac{1}{2}\max_{(b,b')} U\!\bigl((b,b'),t_*\bigr) \ge 0,
\qquad \forall\, t.
\end{aligned}
\end{equation*}

I proceed to construct a $\phi_1^*$ that solves the above problem. Clearly, there exists a surjection $s: [t_*, \bar{t}] \rightarrow \{1,...,l\}$. Consider the following $\phi_1^*$: 
\begin{equation}
\label{eq:phi1star}
\phi_1^*(t)=
\begin{cases}
quit, 
& \text{if } t \in [\underline{t},\, t_*), \\[6pt]
\bigl(b_{s(t)},\, \dfrac{1}{2}\max_{(b,b')} U\!\bigl((b,b'),\, t_*\bigr)\bigr),
& \text{if } t \in [t_*,\, \bar{t}].
\end{cases}
\end{equation}

I first show $\phi_1^*$ given by \eqref{eq:phi1star} is feasible. For the IR constraints, if $t \in [\underline{t}, t_*)$, then
\[\max \limits_{(b',p') \in \mathscr{M}_2}U((b^1(t),b'),t)-p^1(t)-\dfrac{1}{2}\max_{(b,b')} U\!\bigl((b,b'),t_*\bigr)=0,\]
if $t \in [t_*, \bar{t}]$, then by $U((b,b'),\cdot)$ is strictly increasing in $t$ and $\mathcal{B}_*^2$ is non-empty, 
\begin{align*}
\max_{(b',p') \in \mathscr{M}_2}
&\; U\bigl((b^1(t),b'),t\bigr) - p^1(t)
- \dfrac{1}{2}\max_{(b,b')} U\!\bigl((b,b'),t_*\bigr) \\
\ge\;&
\max_{(b,b')} U\!\bigl((b,b'),t_*\bigr)
- \dfrac{1}{2}\max_{(b,b')} U\!\bigl((b,b'),t_*\bigr)
- \dfrac{1}{2}\max_{(b,b')} U\!\bigl((b,b'),t_*\bigr) \\
=\;& 0 .
\end{align*}

As for the IC constraints, 
\begin{itemize}
    \item if $t,t' \in [\underline{t}, t_*)$, then $0 \geq 0$,
    \item if $t \in [\underline{t}, t_*), t' \in [t_*, \bar{t}]$, then $0 \geq \max_{(b,b')} U\!\bigl((b,b'),t\bigr)-\max_{(b,b')} U\!\bigl((b,b'),t_*\bigr),$
    \item if $t \in [t_*, \bar{t}], t' \in [\underline{t}, t_*)$, then $\max_{(b,b')} U\!\bigl((b,b'),t\bigr)-\max_{(b,b')} U\!\bigl((b,b'),t_*\bigr) \geq 0,$
    \item if $t,t' \in [t_*, \bar{t}]$, then $\max_{(b,b')} U\!\bigl((b,b'),t\bigr)-\max_{(b,b')} U\!\bigl((b,b'),t_*\bigr) \geq \max_{(b,b')} U\!\bigl((b,b'),t'\bigr)-\max_{(b,b')} U\!\bigl((b,b'),t_*\bigr)$. 
\end{itemize}

I next show $\phi_1^*$ given by \eqref{eq:phi1star} maximizes the objective. Note that for any $t$,

\begin{equation}
\label{eq:optimality}
\begin{split}
\max_{(b',p') \in \mathscr{M}_2}
&\; U\bigl((b^1(t),b'),t\bigr)
-\frac{1-F(t)}{f(t)}
\frac{\partial}{\partial t}
\max_{(b',p') \in \mathscr{M}_2}
U\bigl((b^1(t),b'),t\bigr)
-\dfrac{1}{2}\max_{(b,b')} U\!\bigl((b,b'),t_*\bigr) \\
\le\;&
\max_{(b,b')}
\Bigl[
U\bigl((b,b'),t\bigr)
-\frac{1-F(t)}{f(t)}\, U_t\bigl((b,b'),t\bigr)
\Bigr]
-\dfrac{1}{2}\max_{(b,b')} U\!\bigl((b,b'),t_*\bigr).
\end{split}
\end{equation}

By the definition of $t_*$ and $U((b,b'),t)-\dfrac{1-F(t)}{f(t)}U_t((b,b'),t)$ is strictly increasing in $t$, for any $t \in [\underline{t}, t_*)$,
\[ \max_{(b',p') \in \mathscr{M}_2}U\bigl((b^1(t),b'),t\bigr)
-\frac{1-F(t)}{f(t)}
\frac{\partial}{\partial t}
\max_{(b',p') \in \mathscr{M}_2}
U\bigl((b^1(t),b'),t\bigr)
-\dfrac{1}{2}\max_{(b,b')} U\!\bigl((b,b'),t_*\bigr) \le 0,\]
where the equality can be attained through $\phi_1^*(t)$ given by \eqref{eq:phi1star}, in particular, $\phi_1^*(t)=quit, \forall t \in [\underline{t}, t_*)$. Meanwhile, for any $t \in [t_*, \bar{t}]$, by \eqref{eq:marketsplitting}, the equality of \eqref{eq:optimality} can also be attained through $\phi_1^*$ given by \eqref{eq:phi1star}. 

Symmetrically, consider surjection $s': [t_*, \bar{t}] \rightarrow \{1,...,m\}$ and the following $\phi_2^*:$

\begin{equation*}
\label{eq:phi2star}
\phi_2^*(t)=
\begin{cases}
quit, 
& \text{if } t \in [\underline{t},\, t_*), \\[6pt]
\bigl(b'_{s'(t)},\, \dfrac{1}{2}\max_{(b,b')} U\!\bigl((b,b'),\, t_*\bigr)\bigr),
& \text{if } t \in [t_*,\, \bar{t}].
\end{cases}
\end{equation*}

It solves for principal 2's \eqref{P4'}. In addition, $\phi_1^*$ and $\phi_2^*$ constructed above satisfy \eqref{eq:4i}: for any $t \in [\underline{t}, t_*),$
\[  \max_{o_i \in \{\phi_i^{*}(s)\}_{s \in T} \setminus \{quit\},\ \forall i} 
    V((o_1,o_2),t)=  \max_{(b,b')} U\!\bigl((b,b'),\, t\bigr)\bigr)-\max_{(b,b')} U\!\bigl((b,b'),\, t_*\bigr)\bigr)<0,\]
    and for any $t \in [\underline{t},t_*), \phi_1^*(t)=\phi_2^*(t)=quit$. For any $t \in [t_*, \bar{t}],$
    \[  \max_{o_i \in \{\phi_i^{*}(s)\}_{s \in T} \setminus \{quit\},\ \forall i} 
    V((o_1,o_2),t)=  \max_{(b,b')} U\!\bigl((b,b'),\, t\bigr)\bigr)-\max_{(b,b')} U\!\bigl((b,b'),\, t_*\bigr)\bigr) \geq 0,\]
    and for any $t \in [t_*, \bar{t}], (\phi_1^{*}(t), \phi_2^{*}(t)) 
        \in 
        \argmax_{o_i \in \{\phi_i^{*}(s)\}_{s \in T} \setminus \{quit\},\ \forall i}
        V((o_1,o_2),t)$. 

        Hence, $(\{\phi_1^*(s)\}_{s \in T} \backslash \{quit\}, \{\phi_2^*(t)\}_{t \in T} \backslash \{quit\})=(\{(b_1,p_1),...,(b_l,p_l)\}, \{(b_1',p_1'),...,(b_m',p_m')\})$ is a \eqref{P4'}-induced menu profile.

Regarding the ``only if'' direction, let $(\{(b_1,p_1),...,(b_l,p_l)\}, \{(b_1',p_1'),...,(b_m',p_m')\})$ be a generic \eqref{P4'}-induced menu profile. Similar to the arguments above, let $\mathscr{M}_1:=\{(b_1,p_1),...,(b_l,p_l)\}$, then principal 2's \eqref{P4'} can be written out as:
\begin{equation*}
\begin{aligned}
\max_{\phi_2: T \to O_2 \cup \{quit\}} \quad
& \mathbb{E}_t\!\left[
\max \limits_{(b,p) \in \mathscr{M}_1}U((b,b^2(t)),t)-p
-\frac{1-F(t)}{f(t)}\frac{\partial \max \limits_{(b,p) \in \mathscr{M}_1}U((b,b^2(t)),t)}{\partial t}
\right] \\
\text{s.t.} \quad
& \max \limits_{(b,p) \in \mathscr{M}_1}U((b,b^2(t)),t)-p^2(t)-p
\\ &\ge \max \limits_{(b,p) \in \mathscr{M}_1}U((b,b^2(t)),t)-p^2(t')-p,
\qquad \forall\, t,t', \\
& \max \limits_{(b,p) \in \mathscr{M}_1}U((b,b^2(t)),t)-p^2(t)-p \ge 0,
\qquad \forall\, t.
\end{aligned}
\end{equation*}

It is not hard to see that $\phi_2(t)=quit, \forall t \leq \underaccent{\dot}{t}; \phi_2(t)=(b',p'), \forall p > \underaccent{\dot}{t}$, where 
\begin{itemize}
\item there exists $b \in \argmin \limits_{(b,p) \in \mathscr{M}_1} p \big|_{\mathcal{B}}$ such that $(b,b') \in \mathcal{B}_*^2$,
\item $\max \limits_{(\tilde{b},\tilde{b}')} U((\tilde{b}, \tilde{b}'),\underaccent{\dot}{t})-\dfrac{1-F(\underaccent{\dot}{t})}{f(\underaccent{\dot}{t})}U_t((\tilde{b}, \tilde{b}'),\underaccent{\dot}{t})-\min \limits_{(b,p) \in \mathscr{M}_1}p=0,$
\item $\max \limits_{(\tilde{b},\tilde{b}')} U((\tilde{b}, \tilde{b}'),\underaccent{\dot}{t})-\min \limits_{(b,p) \in \mathscr{M}_1}p-p'=0$,
\end{itemize}
always solves principal 2's \eqref{P4'}. In particular, it also solves the relaxed \eqref{P4'} without constraints. Hence, if $(\{(b_1,p_1),...,(b_l,p_l)\}, \{(b_1',p_1'),...,(b_m',p_m')\})$ is a \eqref{P4'}-induced menu profile, then it has to be the case that there exists $b_i' \in \{b_1',...,b_m'\}$ such that there exists $b_j \in \{b_1,...,b_l\}$ such that $(b_j, b_i') \in \mathcal{B}_*^2$, and $b_j \in \argmin \limits_{(b,p) \in \mathscr{M}_1}p \big|_{\mathcal{B}}$, and $p_1'=...=p_m'=p'$.\footnote{The requirement of the existence of $b_j$ follows from the condition after ``If additionally'' in Proposition \ref{t1}. The same-price structure follows from the proposed $\phi_2'$ also solves the relaxed \eqref{P4'}. In addition, for cleanliness of the result, I ignore the case where there exists a measure zero set $\tilde{T} \subseteq T$ such that $\phi_2(T)|_{[0,\bar{p}]} \neq \{p'\}$. } Suppose there exist $k,k' \in \{1,...,l\}$ such that $p_k \neq p_{k'}$, then $\argmin \limits_{(b,p) \in \mathscr{M}_1}p|_{\mathcal{B}} \neq \{b_1,...,b_l\}$. Take $b_{j'} \in \{b_1,...,b_l\} \backslash \argmin \limits_{(b,p) \in \mathscr{M}_1} p \big|_{\mathcal{B}}$, and let $\phi_1(t)=(b_j,p_j), \phi_1(t')=(b_{j'}, p_{j'})$. In that case,
\[
\begin{aligned}
v_1\!\left(\phi_1(t'),\, t' \mid \{(b_1',p_1'),\ldots,(b_m',p_m')\}\right)
&\le \max_{(\tilde b,\tilde b')}\, U\!\left((\tilde b,\tilde b'),\, t'\right) - p_{j'} - p' \\
&< \max_{(\tilde b,\tilde b')}\, U\!\left((\tilde b,\tilde b'),\, t'\right) - p_j - p' \\
&= v_1\!\left(\phi_1(t),\, t' \mid \{(b_1',p_1'),\ldots,(b_m',p_m')\}\right),
\end{aligned}
\]
a violation of IC, and a contradiction to $(\{(b_1,p_1),...,(b_l,p_l)\}, \{(b_1',p_1'),...,(b_m',p_m')\})$ is a \eqref{P4'}-induced menu profile. Therefore, $p_1=...=p_l=p$. Given the same-price structure, if there exists $b_{j''} \in \{b_1,..,b_l\}$ such that there does not exist $b_i' \in \{b_1',...,b_m'\}$ such that $(b_{j''}, b_i') \in \mathcal{B}^2_{*}$, then let $\phi_1(t'')=(b_{j''}, p)$. Similar to the argument above,
\[
\begin{aligned}
v_1\!\left(\phi_1(t''),\, t'' \mid \{(b_1',p_1'),\ldots,(b_m',p_m')\}\right)
&< \max_{(\tilde b,\tilde b')}\, U\!\left((\tilde b,\tilde b'),\, t''\right) - p - p' \\
&= v_1\!\left(\phi_1(t),\, t'' \mid \{(b_1',p_1'),\ldots,(b_m',p_m')\}\right),
\end{aligned}
\]
again a violation of IC. Thus, the first two lines of \eqref{eq:marketsplitting} holds. As for the relationship between $p$ and $p'$, by \eqref{eq:4i}, there exists $t_*$ such that 
\[
\begin{aligned}
\max_{(\tilde b,\tilde b')}\; 
&U((\tilde b,\tilde b'),t_*)
-\frac{1-F(t_*)}{f(t_*)}\,U_t((\tilde b,\tilde b'),t_*)
- p  = 0,\\[4pt]
\max_{(\tilde b,\tilde b')}\; 
&U((\tilde b,\tilde b'),t_*)
-\frac{1-F(t_*)}{f(t_*)}\,U_t((\tilde b,\tilde b'),t_*)
- p' = 0,\\[4pt]
\max_{(\tilde b,\tilde b')}\; 
&U((\tilde b,\tilde b'),t_*) = p + p'.
\end{aligned}
\]

Therefore, $p=p'=\dfrac{\max_{(\tilde b,\tilde b')}
U((\tilde b,\tilde b'),t_*)}{2}$, where $t_*$ satisfies $\dfrac{\max \limits_{(b,b')} U((b,b'), t_*)}{2}= \max \limits_{(b,b')} U((b, b'),t_*)-\dfrac{1-F(t_*)}{f(t_*)}U_t((b, b'),t_*)$. 

\end{proof}

\bibliographystyle{apalike}
\bibliography{decomposingca}

\newpage
\normalsize

\pagenumbering{arabic}
\setcounter{page}{1}
\begin{center}
    {\LARGE \textbf{Online Appendix}}\\[1em]

    \end{center}
\addcontentsline{toc}{section}{Online Appendix}

\section{Envelope Formula under Indirect Utility}
\label{s21}

This section provides a technical justification for applying the envelope formula to the indirect utility function, and more generally, to functions expressible as the upper envelope of a family of continuously differentiable functions on a closed interval of $\mathbb{R}$. This analysis establishes the foundation for applying the decomposition method to standard principal-agent problems solvable via the envelope theorem.

\begin{lemma}
\label{l2}
    Let $T$ be a closed interval of $\mathbb{R}$. Consider $\mathfrak{f}: X \times T \rightarrow \mathbb{R}$. If $\mathfrak{f}_2^{+}(x, t)$ and $\mathfrak{f}_2^{-}(x, t)$ exist and $\mathfrak{f}_2^{+}(x, t) \geq \mathfrak{f}_2^{-}(x, t)$ for any $x \in X,t \in T$, and $\mathfrak{V}(t):=\sup \limits_{x \in X} \mathfrak{f}(x,t):=\mathfrak{f}(x^*(t),t)$ is absolute continuous, then $\mathfrak{V}'$ exists only if $\mathfrak{f}_2(x^*(\cdot), \cdot)$ exists, and 
    \begin{equation}
    \label{eq4}
        \mathfrak{V}(t)=\mathfrak{V}(\underline{t})+\displaystyle{\int}_{\underline{t}}^t \mathfrak{f}_2(x^*(s),s)\mathbf{d}s.
    \end{equation}
    The contrapositive can be stated as: for $t$ such that $\mathfrak{f}_2^{+}(x^*(t),t)>\mathfrak{f}_2^{-}(x^*(t),t)$, $\mathfrak{V}'(t)$ does not exist.\footnote{For any $x \in X, t\in T, \mathfrak{f}_2^{+}(x,t):=\lim \limits_{\varepsilon \rightarrow 0^{+}} \frac{\mathfrak{f}(x,t+\varepsilon)-\mathfrak{f}(x,t)}{\varepsilon}, \mathfrak{f}_2^{-}(x,t):=\lim \limits_{\varepsilon \rightarrow 0^{+}} \frac{\mathfrak{f}(x,t)-\mathfrak{f}(x,t-\varepsilon)}{\varepsilon}$.} 
\end{lemma}

\begin{proof}
    It is not hard to observe that \[
\begin{aligned}
\lim_{\varepsilon \to 0^{+}} 
\frac{\mathfrak{V}(t+\varepsilon)-\mathfrak{V}(t)}{\varepsilon}
&= \lim_{\varepsilon \to 0^{+}} 
\frac{\mathfrak{f}\!\left(x^*(t+\varepsilon),\,t+\varepsilon\right)
      -\mathfrak{f}\!\left(x^*(t),\,t\right)}{\varepsilon} \\
&\ge \lim_{\varepsilon \to 0^{+}} 
\frac{\mathfrak{f}\!\left(x^*(t),\,t+\varepsilon\right)
      -\mathfrak{f}\!\left(x^*(t),\,t\right)}{\varepsilon} \\
&= \mathfrak{f}_2^{+}\!\left(x^*(t),\,t\right).
\end{aligned}
\]
 Similarly, $\lim \limits_{\varepsilon \rightarrow 0^{+}} \dfrac{\mathfrak{V}(t)-\mathfrak{V}(t-\varepsilon)}{\varepsilon} \leq \mathfrak{f}_2^{-}(x^*(t),t)$. If $\mathfrak{V}'$ exists, then 
    \begin{equation}
    \label{eq5}
        \mathfrak{f}_2^{+}(x^*(t),t) \leq \lim \limits_{\varepsilon \rightarrow 0^{+}} \dfrac{\mathfrak{V}(t+\varepsilon)-\mathfrak{V}(t)}{\varepsilon} = \lim \limits_{\varepsilon \rightarrow 0^{+}} \dfrac{\mathfrak{V}(t)-\mathfrak{V}(t-\varepsilon)}{\varepsilon} \leq \mathfrak{f}_2^{-}(x^*(t),t).
    \end{equation}
    By the assumption that $\mathfrak{f}_2^{+} \geq \mathfrak{f}_2^{-}$, (\ref{eq5}) implies that $\mathfrak{f}_2^{+}=\mathfrak{f}_2^{-}$, which further implies that $\mathfrak{f}_2(x^*(\cdot),\cdot)$ exists and $\mathfrak{V}'(\cdot)=\mathfrak{f}_2(x^*(\cdot),\cdot)$. (\ref{eq4}) then follows from the absolute continuity of $\mathfrak{V}$. 
\end{proof}

The following remark demonstrates an incorrect result that tries to get rid of the $\mathfrak{f}_2^{+} \geq \mathfrak{f}_2^{-}$ assumption when applying the envelop formula to not everywhere differentiable $\mathfrak{f}$, which further emphasizes the importance of $\mathfrak{f}_2^+ \geq \mathfrak{f}_2^{-}$:

\begin{remark}
     Let $T$ be a closed interval on $\mathbb{R}$. Consider $\mathfrak{f}: X \times T \rightarrow \mathbb{R}$. $\mathfrak{f}(x,\cdot)$ is absolute continuous for all x, $\mathfrak{f}_2^{+}(x, t)$ and $\mathfrak{f}_2^{-}(x, t)$ exist for all $x$ and $t$, and $\mathfrak{V}(t):=\sup \limits_{x \in X} \mathfrak{f}(x,t):=\mathfrak{f}(x^*(t),t)$ is absolute continuous.
     
An argument can be the following: let \[
\begin{aligned}
T_1 &:= \{\, t \in T : \mathfrak{V}'(t)\ \text{does not exist} \,\}, \\
T_2 &:= \{\, t \in T : \mathfrak{V}'(t)\ \text{exists but } \mathfrak{V}'(t)\neq \mathfrak{f}_2(x^*(t),t) \,\}.
\end{aligned}
\] 
By absolute continuity of $\mathfrak{V}$, $T_1$ is of zero measure. In addition, for any $t \in T_2, \mathfrak{f}_2(x^*(t),t)$ does not exist.\footnote{This follows from (\ref{eq5}) and $\mathfrak{V}' \neq \mathfrak{f}_2 \Rightarrow \mathfrak{f}_2^+ \neq \mathfrak{f}_2^-$.} However, this proposition is not sufficient to say that $T_2$ is of zero measure: even though for any $x$, $\mathfrak{f}_2(x,\cdot)$ is differentiable almost everywhere, $x^*(\cdot)$ may not be constant and $T_2$ can be an uncountable union of measure zero sets which may not be of zero measure. Hence, it is incorrect to say that $\mathfrak{V}'(\cdot)=\mathfrak{f}_2(x^*(\cdot), \cdot)$ almost everywhere and Lemma \ref{l2} does not hold. In fact, as shown in \cite{ce2013}, ``downward kinks'', i.e., types where $\mathfrak{f}_2^{+}<\mathfrak{f}_2^{-}$, do matter. 
\end{remark}

The following lemma shows that when a function is the upper envelope of a set of continuous differentiable functions, it is naturally satisfied that the right derivative is no less than the left derivative pointwise.

\begin{lemma}
\label{l3+}
For any finite set of real-valued functions $\{\mathfrak{f}^j\}_{j \in \mathcal{J}}$ defined on $[a,b] \subseteq \mathbb{R}$, if for any $j \in \mathcal{J}, \mathfrak{f}^j \in C^1([a,b])$, then $\mathfrak{g}(\cdot):=\max \limits_{j \in \mathcal{J}} \mathfrak{f}^j(\cdot)$ satisfies:

\begin{equation*}
\label{eq36}
    \mathfrak{g}'_{+}(x) \geq \mathfrak{g}'_{-}(x), \forall x \in (a,b).
\end{equation*}
The result still holds when $\mathcal{J}$ is an infinite set and $\mathfrak{g}(\cdot):=\max \limits_{j \in \mathcal{J}}\mathfrak{f}^j(\cdot)$ is still well-defined. 
    
\end{lemma}

\begin{proof}
    For any $x \in (a,b)$, consider $J(x):=\{j \in \mathcal{J}: \mathfrak{f}^j(x)=\mathfrak{g}(x)\}$. Since for any $j, \mathfrak{f}^j \in C^1([a,b]), \mathfrak{g}_{+}' \text{ and } \mathfrak{g}_{-}'$ are well-defined almost everywhere. Note that for any $x \in (a,b)$, 
\[
\begin{aligned}
\mathfrak{g}'_+(x)
&= \lim_{\varepsilon \to 0^+}\frac{\mathfrak{g}(x+\varepsilon)-\mathfrak{g}(x)}{\varepsilon} = \lim_{\varepsilon \to 0^+}\frac{\max_{j\in \mathcal{J}} \mathfrak{f}^j(x+\varepsilon)-\mathfrak{g}(x)}{\varepsilon} \\
&\ge \lim_{\varepsilon \to 0^+}\frac{\max_{j\in J(x)} \mathfrak{f}^j(x+\varepsilon)-\mathfrak{g}(x)}{\varepsilon} = \max_{j\in J(x)} (\mathfrak{f}^j)'(x).
\end{aligned}
\]
Similarly,  
\[
\begin{aligned}
\mathfrak{g}'_-(x)
&= \lim_{\varepsilon \to 0^+}\frac{\mathfrak{g}(x)-\mathfrak{g}(x-\varepsilon)}{\varepsilon} = \lim_{\varepsilon \to 0^+}\frac{\mathfrak{g}(x)- \max_{j\in \mathcal{J}}\mathfrak{f}^j(x-\varepsilon)}{\varepsilon} \\
&\le \lim_{\varepsilon \to 0^+}\frac{\mathfrak{g}(x)-\max_{j\in J(x)} \mathfrak{f}^j(x-\varepsilon)}{\varepsilon} = \min_{j\in J(x)} (\mathfrak{f}^j)'(x).
\end{aligned}
\]
Hence, for any $x \in (a,b)$, $\mathfrak{g}'_{+}(x) \geq \max_{j\in J(x)} (\mathfrak{f}^j)'(x) \geq \min_{j\in J(x)} (\mathfrak{f}^j)'(x) \geq \mathfrak{g}'_{-}(x)$. 

Clearly, when $\mathcal{J}$ is an infinite set and $\mathfrak{g}(\cdot):=\max \limits_{j \in \mathcal{J}}\mathfrak{f}^j(\cdot)$ is still well-defined, nothing of the arguments above needs to be changed, and the result follows. 
\end{proof}

The last lemma helps say more about the upper envelope of a set of continuous differentiable functions, in particular, help showing the ``upper envelope of upper envelope'' is absolute continuous. 

\begin{lemma}
\label{l1}
    For any finite set of real-valued functions $\{\mathfrak{f}^j\}_{j \in \mathcal{J}}$ defined on $[a,b] \subseteq \mathbb{R}$, if for any $j \in \mathcal{J}, \mathfrak{f}^j \in C^1([a,b])$, then $\mathfrak{g}(\cdot):=\max \limits_{j \in \mathcal{J}} \mathfrak{f}^j(\cdot)$ satisfies:
\begin{enumerate}
    \item $\mathfrak{g}(\cdot)$ is absolute continuous on $[a,b]$,
    \item there exists an integrable $\mathfrak{h}: [a,b] \rightarrow \mathbb{R}_{+}$ such that $|\mathfrak{g}'(\cdot)| \leq \mathfrak{h}(\cdot)$ a.e.,
    \item if for any $j,  \mathfrak{f}^j$ is increasing on $[a,b]$, then $\mathfrak{g}$ is also increasing on $[a,b]$. 

    \end{enumerate}
    The result still holds when $\mathcal{J}$ is an infinite set, $\{\mathfrak{f}^j\}_{j \in \mathcal{J}}$ is a set of equi-Lipschitz functions, and $\mathfrak{g}(\cdot):=\sup \limits_{j \in \mathcal{J}}\mathfrak{f}^j(\cdot)<\infty$. 
\end{lemma}

\begin{proof}
    For any $j, \mathfrak{f}^j \in C^1([a,b])$ implies that $\mathfrak{f}^j$ is Lipschitz on $[a,b]$. Let the Lipschitz constant for $\mathfrak{f}^j$ be $L^j \in \mathbb{R}_{+}$. For any $x,y \in [a,b],$
    \[
\begin{aligned}
|\mathfrak{g}(x)-\mathfrak{g}(y)|
&= |\max_j \mathfrak{f}^j(x) - \max_j \mathfrak{f}^j(y)| \\
&\le \max_j \lvert \mathfrak{f}^j(x)-\mathfrak{f}^j(y)\rvert \\
&\le \max_j L^j\,\lvert x-y\rvert .
\end{aligned}
\]

 Hence, $\mathfrak{g}$ is $\max \limits_j L^j$-Lipschitz on $[a,b]$, which implies that $\mathfrak{g}$ is absolute continuous on $[a,b]$. 

    As for the second property, notice that $\mathfrak{g}$ is absolute continuous implies that $\mathfrak{g}$ is differentiable almost everywhere. Since for any $j, \mathfrak{f}^j \in C^1([a,b]), (\mathfrak{f}^j)'$ is bounded on $[a,b]$, which further implies that there exists $M \in \mathbb{R}_{+}$ such that \[\max \limits_j \max \limits_{x \in [a,b]} |(\mathfrak{f}^j)'(x)| \leq M.\] Hence, $|\mathfrak{g}'(x)| \leq M$ whenever it exists, and the integrable $\mathfrak{h}$ function can simply be $\mathfrak{h}(x)=M, \forall x \in [a,b]$.

    As for the third property, suppose there exist $a \leq x < x' \leq b$ such that $\mathfrak{g}(x')<\mathfrak{g}(x)$. Let $\mathfrak{g}(x')=\mathfrak{f}^{j'}(x')$ and $\mathfrak{g}(x)=\mathfrak{f}^j(x)$, then $\mathfrak{f}^j(x') \leq \mathfrak{f}^{j'}(x') < \mathfrak{f}^j(x)$, a contradiction to $\mathfrak{f}^j$ is increasing on $[a,b]$. 

    As for the case where $\mathcal{J}$ is an infinite set, and $\{\mathfrak{f}^j\}_{j \in \mathcal{J}}$ is a set of equi-Lipschitz functions, for any $j \in \mathcal{J}$, let $\mathfrak{f}^j$ be $L$-Lipschitz. For any $x,y \in [a,b],$
    \[
\begin{aligned}
|\mathfrak{g}(x)-\mathfrak{g}(y)|
&= |\sup_j \mathfrak{f}^j(x) - \sup_j \mathfrak{f}^j(y)| \\
&\le \sup_j \lvert \mathfrak{f}^j(x)-\mathfrak{f}^j(y)\rvert \\
&\le L\,\lvert x-y\rvert .
\end{aligned}
\]
 Hence, $\mathfrak{g}$ is $L$-Lipschitz on $[a,b]$, which implies that $\mathfrak{g}$ is absolute continuous on $[a,b]$. 

 As for the second property, for any $x \in [a,b],$
 \[ |\mathfrak{g}'(x)|=|\lim \limits_{\varepsilon \rightarrow 0} \frac{\mathfrak{f}(x+\varepsilon)-\mathfrak{f}(x)}{\varepsilon}| \leq L,\]
 so that the integrable $\mathfrak{h}$ function can simply be $\mathfrak{h}(x)=L, \forall x \in [a,b]$.

 As for the third property, suppose there exist $a \leq x < x' \leq b$ such that $\mathfrak{g}(x')<\mathfrak{g}(x)$, then there exists $j \in \mathcal{J}$ and $\eta>0$ such that
 \[\mathfrak{f}^j(x') \leq \mathfrak{g}(x')<\mathfrak{g}(x)-\eta<\mathfrak{f}^j(x),\]
 a contradiction to $\mathfrak{f}^j$ is increasing on $[a,b]$. 

\end{proof}

To ensure that $\mathfrak{V}(\cdot)=\sup \limits_{x \in X} \mathfrak{f}(x, \cdot)$ is absolute continuous, by Lemma \ref{l1}, if $\mathfrak{f}(x,\cdot)$ is the upper envelope of a (certain) set of $C^1$ functions for any $x$, then $\mathfrak{f}(x, \cdot)$ is absolute continuous and there exists $M_x \in \mathbb{R}_{+}$ such that $|\mathfrak{f}(x, \cdot)| \leq M_x$ for almost all $t$. 
By \cite{mise2002}, this ensures that $\mathfrak{V}(\cdot)$ is absolute continuous.

\section{Omitted Proofs}
\subsection{Proof for Constructing PBE under Outside Options: Proposition \ref{prop:intrinsicca}}

\begin{proof}
    I first prove the intrinsic common agency case. For all $(\phi_1^*,...,\phi_n^*)$ such that for any $i, \phi_i^*$ solves \eqref{P4'} and satisfies \eqref{eq:4i}, let $\mathscr{M}_i:=\{\phi_i^*(s)\}_{s \in T} \backslash \{quit\}, \forall i$. Consider $\sigma_A$ such that for any $t$,
    \[ \sigma_A(\mathscr{M}_1,...,\mathscr{M}_n,t)=\begin{cases} quit, & \text{ if } \max \limits_{o_i \in \mathscr{M}_i, \forall i} V((o_1,...,o_n),t)<0, \\ (o_1^t,...,o_n^t) &\text{ otherwise. }\end{cases}\]
    By \eqref{eq:4i}, condition \eqref{1i} is not violated. Let $\sigma_A$ satisfy \eqref{1i} on menu profiles other than $(\mathscr{M}_1,...,\mathscr{M}_n)$ as well. Suppose $(\mathscr{M}_1,...,\mathscr{M}_n, \sigma_A)$ is not a PBE, then there exists $i$ and $\mathscr{M}_i'$ such that 
    \[\mathbb{E}_t[u_i(\sigma_A(\mathscr{M}_1,...,\mathscr{M}_i',...,\mathscr{M}_n,t))]>\mathbb{E}_t[u_i(\sigma_A(\mathscr{M}_1,...,\mathscr{M}_i,...,\mathscr{M}_n,t))].\]
    Consider $\phi_i':T\rightarrow O_i \cup \{quit\}$ such that for any $t$,
    \[ \phi_i'(t)=\begin{cases} quit, & \text{ if } \max \limits_{o_j \in \mathscr{M}_j, \forall j \neq i; o_i \in \mathscr{M}_i'} V((o_1,...,o_n),t)<0, \\ \in \argmax \limits_{o_i \in \supp \sigma_A(\mathscr{M}_1,...,\mathscr{M}_i',...,\mathscr{M}_n,t)|_{O_i}} u_i(o_i,t) &\text{ otherwise. }\end{cases}\]
    I first show that $\mathbb{E}_t[u_i(\phi_i'(t))]>\mathbb{E}_t[u_i(\phi_i^*(t))]$. By \eqref{eq:4i} and the construction of $\sigma_A$ above, for any $t$ such that $\max \limits_{o_i \in \mathscr{M}_i, \forall i} V((o_1,...,o_n),t)<0$,
    \[\sigma_A(\mathscr{M}_1,...,\mathscr{M}_n,t)=\phi_i^*(t)=quit.\]
    Similarly, for any $t$ such that $\max \limits_{o_j \in \mathscr{M}_j, \forall j \neq i; o_i \in \mathscr{M}_i'} V((o_1,...,o_n),t)<0$,
    \[\sigma_A(\mathscr{M}_1,...,\mathscr{M}_i',...,\mathscr{M}_n,t)=\phi_i'(t)=quit. \]
    Hence, \eqref{166} in the proof of Proposition \ref{p1} still holds, which directly gives $\mathbb{E}_t[u_i(\phi_i'(t))]>\mathbb{E}_t[u_i(\phi_i^*(t))]$.

    I proceed to show that $\phi_i'$ is feasible to \eqref{P4'}. If $\phi_i'(t)=quit$, then clearly the IR constraint is satisfied. As for the IC constraint,
    \begin{itemize}
        \item if $\phi_i'(t')=quit$, then 
        \[ v_i(\phi_i'(t),t|\prod_{j \neq i} \{\phi_j^*(s)\}_{s \in T} \backslash \{quit\})=0 \geq 0=v_i(\phi_i'(t'),t|\prod_{j \neq i} \{\phi_j^*(s)\}_{s \in T} \backslash \{quit\}),\]
        \item if $\phi_i'(t') \neq quit$, then
        \[
\begin{aligned}
v_i\!\left(\phi_i'(t),\,t \mid \prod_{j \neq i} \{\phi_j^*(s)\}_{s \in T}\setminus\{quit\}\right)
= 0 &> \max_{\substack{o_j \in \mathscr{M}_j\ \forall j\neq i\\ o_i \in \mathscr{M}_i'}} V\!\left((o_1,\ldots,o_n),\,t\right) \\
&\ge v_i\!\left(\phi_i'(t'),\,t \mid \prod_{j \neq i} \{\phi_j^*(s)\}_{s \in T}\setminus\{quit\}\right).
\end{aligned}
\]

    \end{itemize}
    If $\phi_i'(t') \neq quit$, then by construction, the IR constraint is satisfied. As for the IC constraint, 
    \begin{itemize}
        \item if $\phi_i'(t') = quit$, then
                \[ v_i(\phi_i'(t),t|\prod_{j \neq i} \{\phi_j^*(s)\}_{s \in T} \backslash \{quit\}) \geq 0=v_i(\phi_i'(t'),t|\prod_{j \neq i} \{\phi_j^*(s)\}_{s \in T} \backslash \{quit\}),\]
        \item if $\phi_i'(t') \neq quit$, then similar arguments in the proof of Proposition \ref{p0}, in particular, \eqref{177} holds. 
    \end{itemize}
    Hence, $\phi_i'$ is feasible to \eqref{P4'}, a contradiction to $\phi_i^*$ solves \eqref{P4'}. Therefore, $(\mathscr{M}_1,...,\mathscr{M}_n, \sigma_A)$ is a PBE. 

    I proceed to prove the delegated common agency case. For all $(\phi_1^*,...,\phi_n^*)$ such that for any $i, \phi_i^*$ solves \eqref{P4''} and satisfies \eqref{4D}, let $\mathscr{M}_i=\{\phi_i^*(s)\}_{s \in T}, \forall i$. Consider $\sigma_A$ such that for any $t$,
    \[\sigma_A(\mathscr{M}_1,...,\mathscr{M}_n,t)=(\phi_1^*(t_1),...,\phi_n^*(t_n)).\]
By \eqref{4D}, condition \eqref{1d} is not violated. Let $\sigma_A$ satisfy \eqref{1d} on menu profiles other than $(\mathscr{M}_1,...,\mathscr{M}_n)$ as well. Suppose $(\mathscr{M}_1,...,\mathscr{M}_n, \sigma_A)$ is not a PBE, then there exists $i$ and $\mathscr{M}_i'$ such that 
    \[\mathbb{E}_t[u_i(\sigma_A(\mathscr{M}_1,...,\mathscr{M}_i',...,\mathscr{M}_n,t))]>\mathbb{E}_t[u_i(\sigma_A(\mathscr{M}_1,...,\mathscr{M}_i,...,\mathscr{M}_n,t))].\]
    Consider $\phi_i':T\rightarrow O_i$ such that for any $t$,
    \[\phi_i'(t) \in \argmax \limits_{o_i \in \supp \sigma_A(\mathscr{M}_1,...,\mathscr{M}_i',...,\mathscr{M}_n,t)|_{O_i}}u_i(o_i,t).\]
    \eqref{166} in the proof of Proposition \ref{p1} still holds, which directly gives $\mathbb{E}_t[u_i(\phi_i'(t))]>\mathbb{E}_t[u_i(\phi_i^*(t))]$. Meanwhile, by $\sigma_A$ satisfying \eqref{1d}, $\phi_i'$ is feasible to \eqref{P4''}. This leads to a contradiction of $\phi_i^*$ is a solution to \eqref{P4''}. Therefore, $(\mathscr{M}_1,...,\mathscr{M}_n, \sigma_A)$ is a PBE. 
\end{proof}

\subsection{Proofs for Bundling: Proposition \ref{t2} and Corollary \ref{c1}}

\begin{proof}
    For any $b \in \mathcal{B}$, let principal 1 offer $\{(b,p)\}$, principal 2 offer $\{(b_1^*,p_1^*),...,(b_m^*,p_m^*)\}$ specified in Proposition \ref{t2}. First consider principal 2's \eqref{P4'} given principal 1 offering $\{(b,p)\}$. In particular, since principal 1 offers a singleton menu, the agent's indirect utility at principal 2 degenerates to an upper envelope of a single function, which is simply the function itself: 
    \begin{equation*}
\begin{aligned}
\max_{\phi_2: T \to O_2 \cup \{quit\}} \quad
& \mathbb{E}_t\!\left[
\varphi^{b}(b^2(t),t)-\dfrac{U\!\bigl((b,b^\complement),t_*\bigr)}{2} 
\right] \\
\text{s.t.} \quad
 U((b,b^2(t)),t)-p^2(t)-\dfrac{U\!\bigl((b,b^\complement),t_*\bigr)}{2} &\ge U((b,b^2(t')),t)-p^2(t')-\dfrac{U\!\bigl((b,b^\complement),t_*\bigr)}{2},
\qquad \forall\, t,t', \\
U((b,b^2(t)),t)-p^2(t)-\dfrac{U\!\bigl((b,b^\complement),t_*\bigr)}{2} &\ge 0,
\qquad \forall\, t.
\end{aligned}
\end{equation*}

Since $\varphi^b$ satisfies MD$^\star$ and for any $a \in \Delta(\mathcal{B}), \varphi^b(a,\cdot)$ is monotonic, by \cite{feng2025}, $\{\varphi^{b}(\tilde{b}, \cdot)\}_{\tilde{b} \in \mathcal{B}}$ has a unique upper envelope, and $\{quit\} \cup \{\text{ the upper envelope of } \{\varphi^{b}(\tilde{b}, \cdot)\}_{\tilde{b} \in \mathcal{B}} \}$ is the unique minimally optimal menu of the problem above. Moreover, at $\underline{t}$,
\[\{b^\complement\}=\argmax \limits_{\tilde{b}} U((b,\tilde{b}),\underline{t})-\dfrac{1}{f(\underline{t})}U_t((b,\tilde{b}),\underline{t}),\]
and at $\bar{t}$, for any $\tilde{b} \nsupset b^\complement, U((b,\tilde{b}), \bar{t})<U((b,b^\complement),\bar{t})$. In that case, by Proposition 3 of \cite{feng2025}, \text{ the upper envelope of } $\{\varphi^{b}(\tilde{b}, \cdot)\}_{\tilde{b} \in \mathcal{B}} \}$, denoted as $\{(b_1^*,p_1^*),...,(b_m^*,p_m^*)\}$, is a nested or a tree menu with the root being $b^\complement$, and $\min \limits_{x \in \{1,...,m\}} p_x^*=U((b,b^\complement),t_*)-\dfrac{U((b,b^\complement),t_*)}{2}=\dfrac{U((b,b^\complement),t_*)}{2}$. 

Now consider principal 1's \eqref{P4'} given principal 2 offering $\{(b_1^*,p_1^*),...,(b_m^*,p_m^*)\}$. Without loss of generality, let $p_1^*=\min \limits_{x \in \{1,...,m\}} p_x^*$. Note that principal 1's objective satisfies

\begin{equation}
\label{eq:p1objective}
\begin{aligned}
\mathbb{E}_t\!\Biggl[
&\max_{(b_i^*,p_i^*)} U\bigl((b^1(t),b_i^*),t\bigr) - p_i^*
-\frac{1-F(t)}{f(t)}\,
\frac{\partial}{\partial t}\max_{(b_i^*,p_i^*)} U\bigl((b^1(t),b_i^*),t\bigr)
\Biggr] \\
\le\;& \mathbb{E}_t\!\Biggl[
\max_{(b_i^*,p_i^*)} U\bigl((b^1(t),b_i^*),t\bigr) - p_i^*
-\frac{1-F(t)}{f(t)}\,U_t\bigl((b^1(t),b_i^*),t\bigr)
\Biggr] \\
\le\;& \mathbb{E}_t\!\Biggl[
\max_{(b_i^*,p_i^*)} U\bigl(([b_i^*]^\complement,b_i^*),t\bigr) - p_i^*
-\frac{1-F(t)}{f(t)}\,U_t\bigl(([b_i^*]^\complement,b_i^*),t\bigr)
\Biggr] \\
=\;& \mathbb{E}_t\!\Biggl[
U\bigl(([b_1^*]^\complement,b_1^*),t\bigr)
-\frac{1-F(t)}{f(t)}\,U_t\bigl(([b_1^*]^\complement,b_1^*),t\bigr)
- p_1^*
\Biggr],
\end{aligned}
\end{equation}
where the second ``$\leq$'' is by the condition \ref{p10cd2} in Proposition \ref{t2}, and the ``$=$'' is by condition \ref{p10cd3} in Proposition \ref{t2} and $p_1^*=\min \limits_{x \in \{1,...,m\}} p_x^*$. In particular, by the definition of $t_*$ and the assumption that $U\bigl(([b_1^*]^\complement,b_1^*),t\bigr)
-\frac{1-F(t)}{f(t)}\,U_t\bigl(([b_1^*]^\complement,b_1^*),t\bigr)$ is strictly increasing in $t$, one $\phi_1^*$ that can attain the equality is
\begin{equation}
\label{eq:phi1starp10}
\phi_1^*(t)=
\begin{cases}
quit, 
& \text{if } t \in [\underline{t},\, t_*), \\[6pt]
\bigl([b_1^*]^\complement,\, \dfrac{U\!\bigl(([b_1^*]^\complement,b_1^*),\, t_*\bigr)}{2} \big) ,
& \text{if } t \in [t_*,\, \bar{t}].
\end{cases}
\end{equation}

Clearly, $\phi_1^*$ given by \eqref{eq:phi1starp10} satisfies IC and IR constraints, as $\{\phi_1^*(s)\}_{s \in T} \backslash \{quit\}$ is a singleton. Note that \eqref{eq:4i} is also satisfied by Corollary \ref{cor:n-1singletonoo}. Hence, $(\{(b,p)\}, \{(b_1^*,p_1^*),...,(b_m^*,p_m^*)\})$ specified in Proposition \ref{t2} is a \eqref{P4'}-induced menu profile. 

Regarding Corollary \ref{c1}, the existence of $b_0 \supsetneq b_0^\complement$ such that $U((b_0,b_0'),\bar{t})>U(b_0,b_0^\complement),\bar{t})$ ensures that $m \geq 2$. Meanwhile, \eqref{eq:p1objective} should be rewritten as
\begin{equation*}
\begin{aligned}
\mathbb{E}_t\!\Biggl[
&\max_{(b_i^*,p_i^*)} U\bigl((b^1(t),b_i^*),t\bigr) - p_i^*
-\frac{1-F(t)}{f(t)}\,
\frac{\partial}{\partial t}\max_{(b_i^*,p_i^*)} U\bigl((b^1(t),b_i^*),t\bigr)
\Biggr] \\
\le\;& \mathbb{E}_t\!\Biggl[
\max_{(b_i^*,p_i^*)} U\bigl((b^1(t),b_i^*),t\bigr) - p_i^*
-\frac{1-F(t)}{f(t)}\,U_t\bigl((b^1(t),b_i^*),t\bigr)
\Biggr] \\
\le\;& \mathbb{E}_t\!\Biggl[
\max_{(b_i^*,p_i^*)} U\bigl(([b_i^*]^\complement,b_i^*),t\bigr) - p_i^*
-\frac{1-F(t)}{f(t)}\,U_t\bigl(([b_i^*]^\complement,b_i^*),t\bigr)
\Biggr] \\
=\;& \mathbb{E}_t\!\Biggl[
U\bigl((b_0,b_0^\complement),t\bigr)
-\frac{1-F(t)}{f(t)}\,U_t\bigl((b_0,b_0^\complement),t\bigr)
- p_1^*
\Biggr].
\end{aligned}
\end{equation*}

\end{proof}

\section{Proofs for Extensions}

\subsection{Proof of Proposition \ref{p0}}
\begin{proof}
I first prove the first part of the proposition, i.e., when $\sigma_A$ is a pure strategy that satisfies condition \ref{cd1} in Definition \ref{d0} and \eqref{piia}. For $(\phi_1^*,...,\phi_n^*)$ such that $\forall i, \phi_i^*$ solves (\ref{P1}), it is straightforward that, by the definition of $\mathscr{M}_i$, for any $i$, \[\mathbb{E}_t[u_i(\sigma_A(\{\phi_1^*(s)\}_{s \in T},...,\{\phi_n^*(s)\}_{s \in T},t),t)]=\mathbb{E}_t[u_i(\sigma_A(\mathscr{M}_1,...,\mathscr{M}_n,t),t)].\]
If $(\mathscr{M}_1,...,\mathscr{M}_n,\sigma_A)$ is not a PBE, then there exists $i$ and $\mathscr{M}_i'$ such that \[\mathbb{E}_t[u_i(\sigma_A(\mathscr{M}_1,...\mathscr{M}_i,...,\mathscr{M}_n,t),t)]<\mathbb{E}_t[u_i(\sigma_A(\mathscr{M}_1,...,\mathscr{M}_i',...,\mathscr{M}_n,t),t)].\]
Let $\phi_i': T \rightarrow O_i$ such that for any $t$
\begin{equation} 
\label{16++} \phi_i'(t):=\sigma_A(\mathscr{M}_1,...,\mathscr{M}_i',...,\mathscr{M}_n,t)\bigg|_{O_i}.
\end{equation}
$\phi_i'$ is well-defined because $\sigma_A$ is a pure strategy. It is straightforward that \begin{equation} \label{eq:subsetmiprime} \mathscr{M}_1 \times ... \times \{\phi_i'(s)\}_{s \in T} \times ... \times \mathscr{M}_n \subseteq \mathscr{M}_1 \times ... \mathscr{M}_i' \times ... \times \mathscr{M}_n.\end{equation} I first show that $\phi_i'$ satisfies the IC constraint of \eqref{P1}, i.e., $\forall t,t', v_i(\phi_i'(t),t|\mathscr{M}_{-i}) \geq v_i(\phi_i'(t'),t|\mathscr{M}_{-i})$. This is because 
\begin{equation*}
\begin{split}
    v_i(\phi_i'(t),t \mid \mathscr{M}_{-i})
    &= \max_{o \in \mathscr{M}_1 \times \dots \times \mathscr{M}_i' \times \dots \times \mathscr{M}_n} V(o,t) \\
    &\geq \max_{o \in \mathscr{M}_1 \times \dots \times \{\phi_i'(s)\}_{s \in T} \times \dots \times \mathscr{M}_n} V(o,t) \\
    &\geq v_i(\phi_i'(t'),t \mid \mathscr{M}_{-i}), 
\end{split}
\end{equation*}
where the first $``\geq''$ is by \eqref{eq:subsetmiprime}. 

Moreover, to make use of \eqref{piia}, want to show that \[\supp \sigma_A(\mathscr{M}_1,...,\mathscr{M}_i',...,\mathscr{M}_n,t) \subseteq \mathscr{M}_1 \times ... \times \{\phi_i'(s)\}_{s \in T} \times ... \times \mathscr{M}_n, \forall t.\]

It suffices to show that \[\supp \sigma_A(\mathscr{M}_1,...,\mathscr{M}_i',...,\mathscr{M}_n,t) \bigg|_{O_i} \subseteq \{\phi_i'(s)\}_{s \in T}, \forall t,\] which is guaranteed by (\ref{16++}). This, combining with $\sigma_A$ satisfying \eqref{piia}, gives
\[\sigma_A(\mathscr{M}_1,...,\{\phi_i'(s)\}_{s \in T},...,\mathscr{M}_n,t) = \sigma_A(\mathscr{M}_1,...,\mathscr{M}_i',...,\mathscr{M}_n,t), \forall t.\]
As a result,
\begin{equation*}
\begin{split}
    \mathbb{E}_t[u_i(\sigma_A(\{\phi_1^*(s)\}_{s \in T},\dots,\{\phi_n^*(s)\}_{s \in T},t),t)]
    &= \mathbb{E}_t[u_i(\sigma_A(\mathscr{M}_1,\dots,\mathscr{M}_n,t),t)] \\
    &< \mathbb{E}_t[u_i(\sigma_A(\mathscr{M}_1,\dots,\mathscr{M}_i',\dots,\mathscr{M}_n,t),t)] \\
    &= \mathbb{E}_t[u_i(\sigma_A(\mathscr{M}_1,\dots,\{\phi_i'(s)\}_{s \in T},\dots,\mathscr{M}_n,t),t)] .
\end{split}
\end{equation*}
Since $\phi_i'$ also satisfies the IC constraint of \eqref{P1}, this is a contradiction to $\phi_i^*$ solves (\ref{P1}). Hence, $(\mathscr{M}_1,...,\mathscr{M}_n, \sigma_A)$ is a PBE.

I proceed to prove the second part of the proposition, i.e., the case where $\sigma_A$ is a mixed strategy that satisfies condition \ref{cd1} in Definition \ref{d0} and \eqref{contraction}. Identical to the arguments above, if $(\mathscr{M}_1,...,\mathscr{M}_n,\sigma_A)$ is not a PBE, then there exists $i$ and $\mathscr{M}_i'$ such that $\mathbb{E}_t[u_i(\sigma_A(\mathscr{M}_1,...\mathscr{M}_i,...,\mathscr{M}_n,t),t)]<\mathbb{E}_t[u_i(\sigma_A(\mathscr{M}_1,...,\mathscr{M}_i',...,\mathscr{M}_n,t),t)]$. Let $\phi_i'': T \rightarrow \Delta(O_i)$ satisfy (\ref{16++}). I first show that $\phi_i''$ satisfies the IC constraint in \eqref{P1'}, i.e., $\forall t,t', v_i(\phi_i''(t),t|\mathscr{M}_{-i}) \geq v_i(\phi_i''(t'),t|\mathscr{M}_{-i})$. This is because 
\begin{equation*}
\begin{split}
    v_i(\phi_i''(t),t \mid \mathscr{M}_{-i})
    &= \int_{O_i} v_i(o_i,t \mid \mathscr{M}_{-i}) \, \phi_i''(t)(\mathbf{d}o_i) \\
    &= 
       \max_{o \in \mathscr{M}_1 \times \dots \times \mathscr{M}_i' \times \dots \times \mathscr{M}_n}
       V(o,t) \\
    &\geq v_i(\phi_i''(t'),t \mid \mathscr{M}_{-i}),
\end{split}
\end{equation*}
where the second ``$=$'' is by for any $o_i \in \supp \phi_i''(t), o_i \in \supp \sigma_A(\mathscr{M}_1,...,\mathscr{M}_i',...,\mathscr{M}_n,t)$. 

Note that
\small
\begin{align*}
\mathbb{E}_t[\displaystyle{\int}_{O_i}&u_i(\sigma_A(\mathscr{M}_1,...,\tilde{o}_i,...,\mathscr{M}_n,t),t)\phi_i''(t)(\mathbf{d}\tilde{o}_i)] \\
&= \mathbb{E}_t[\displaystyle{\int}_{O_i} \displaystyle{\int}_{O} u_i(o',t)\sigma_A(\mathscr{M}_1,...,\tilde{o}_i,...,\mathscr{M}_n,t)(\mathbf{d}o')\phi_i''(t)(\mathbf{d}\tilde{o}_i)] \\
&= \mathbb{E}_t[\displaystyle{\int}_{O_i} \displaystyle{\int}_{\mathscr{M}_1 \times ... \times \{\tilde{o}_i\} \times ... \times\mathscr{M}_n} u_i(o',t) \dfrac{\sigma_A(\mathscr{M}_1,...,\mathscr{M}_i',...,\mathscr{M}_n,t)(\mathbf{d}o')}{\sigma_A(\mathscr{M}_1,...,\mathscr{M}'_i,...,\mathscr{M}_n,t)[\supp \sigma_A(\mathscr{M}_1,...,\mathscr{M}_i',...,\mathscr{M}_n,t) \cap \mathscr{M}_1 \times ... \times \{\tilde{o}_i\} \times ... \times\mathscr{M}_n]} \\
& \quad \phi_i''(t)(\mathbf{d}\tilde{o}_i) ] \\
&= \mathbb{E}_t[\displaystyle{\int}_O u_i(o',t) \\
& \quad \displaystyle{\int}_{O_i} 1_{o' \in \mathscr{M}_1 \times ... \times \{\tilde{o}_i\} \times ... \mathscr{M}_n} \dfrac{\phi_i''(t)(\mathbf{d}\tilde{o}_i) }{\sigma_A(\mathscr{M}_1,...,\mathscr{M}'_i,...,\mathscr{M}_n,t)[\supp \sigma_A(\mathscr{M}_1,...,\mathscr{M}_i',...,\mathscr{M}_n,t) \cap \mathscr{M}_1 \times ... \times \{\tilde{o}_i\} \times ... \times\mathscr{M}_n]} \\
& \quad \sigma_A(\mathscr{M}_1,...,\mathscr{M}_i',...,\mathscr{M}_n,t)(\mathbf{d}o')] \\
&= \mathbb{E}_t[\displaystyle{\int}_O u_i(o',t) \displaystyle{\int}_{O_i} 1_{o' \in \mathscr{M}_1 \times ... \times \{\tilde{o}_i\} \times ... \mathscr{M}_n} \mathbf{d}\tilde{o}_i \sigma_A(\mathscr{M}_1,...,\mathscr{M}_i',...,\mathscr{M}_n,t)(\mathbf{d}o')] \\
&= \mathbb{E}_t[\sigma_A(\mathscr{M}_1,...,\mathscr{M}_i',...,\mathscr{M}_n,t)],
\end{align*} \normalsize
where the second ``$=$'' is by \eqref{contraction}, and the fourth ``$=$'' is by (\ref{16++}). As a result, \[\mathbb{E}_t[\displaystyle{\int}_{O_i}u_i(\sigma_A(\mathscr{M}_1,...,\tilde{o}_i,...,\mathscr{M}_n,t),t)\phi_i''(t)(\mathbf{d}\tilde{o}_i)] > \mathbb{E}_t[u_i(\sigma_A(\mathscr{M}_1,...,\{\phi_i^*(s)\}_{s \in T},...,\mathscr{M}_n,t),t)].\]
By (\ref{2*}), for any $i,$ \[ \mathbb{E}_t[u_i(\sigma_A(\{\phi_1^*(s)\}_{s \in T},...,\{\phi_n^*(s)\}_{s \in T},t),t)] \geq \mathbb{E}_t[u_i(\sigma_A(\mathscr{M}_1,...\{\phi_i^*(t)\},...,\mathscr{M}_n,t),t)],\] which further implies that, 

\[ \mathbb{E}_t[\displaystyle{\int}_{O_i}u_i(\sigma_A(\mathscr{M}_1,...,\tilde{o}_i,...,\mathscr{M}_n,t),t)\phi_i''(t)(\mathbf{d}\tilde{o}_i)] > \mathbb{E}_t[u_i(\sigma_A(\mathscr{M}_1,...\{\phi_i^*(t)\},...,\mathscr{M}_n,t),t)].\]

This, combining with the feasibility of $\phi_i''$, shows that $\phi_i^*$ does not solve (\ref{P1'}), a contradiction. Hence, $(\mathscr{M}_1,...,\mathscr{M}_n, \sigma_A)$ is a PBE.
\end{proof}

\subsection{Proof of Proposition \ref{p2}}
\begin{proof}

First discuss the case where $\sigma_A$ is a pure-strategy that satisfies \eqref{piia}. For any $t$, let \[\supp\sigma_A(\mathscr{M}_1, ..., \mathscr{M}_n,t) := (o_{1t}, ..., o_{nt}).\] Let $\phi_i^*(t)= o_{it}, \forall i,t$. It is straightforward that $\{\phi_i^*(s)\}_{s \in T} \subseteq \mathscr{M}_i, \forall i$. Proceed to show that for any $i, \phi_i^*$ solves problem (\ref{P2}). First show the feasibility. By condition \ref{cd1} in Definition \ref{d0},  for any $t, $
\[(o_{1t},...,o_{nt}) \in \supp \sigma_A(\mathscr{M}_1,...,\mathscr{M}_n,t) \Rightarrow (o_{1t},...,o_{nt}) \in \argmax \limits_{\tilde{o}_i \in \mathscr{M}_i, \forall i} V((\tilde{o_1},...,\tilde{o_n}),t).\] Hence, for any $i, \max \limits_{\tilde{o}_i \in \mathscr{M}_i, \forall i} V((\tilde{o_1},...,\tilde{o_n}),t)=v_i(\phi_i^*(t),t |\mathscr{M}_{-i}) \geq v_i(\phi_i^*(t'),t|\mathscr{M}_{-i}), \forall t,t'$. 

Next show the optimality. Since \[\supp \sigma_A(\mathscr{M}_1,...,\mathscr{M}_n,t) \subseteq \mathscr{M}_1 \times ... \times \{\phi_i^*(s)\}_{s \in T} \times ... \times \mathscr{M}_n,\] by \eqref{piia}, 
\begin{equation}
\label{17+}
    \mathbb{E}_t[u_i(\sigma_A(\mathscr{M}_1,...,\{\phi_i^*(s)\}_{s \in T},...,\mathscr{M}_n,t),t)]= \mathbb{E}_t[u_i(\sigma_A(\mathscr{M}_1,...,\mathscr{M}_i,...,\mathscr{M}_n,t),t)], \forall i. 
\end{equation}
Suppose there exists $i$ and $\phi_i':T \rightarrow O_i$ such that $v_i(\phi_i'(t),t|\mathscr{M}_{-i}) \geq v_i(\phi_i'(t'),t|\mathscr{M}_{-i}), \forall t,t'$ and \[\mathbb{E}_t[u_i(\sigma_A(\mathscr{M}_1,...,\{\phi_i^*(s)\}_{s \in T},...,\mathscr{M}_n,t),t)]<\mathbb{E}_t[u_i(\sigma_A(\mathscr{M}_1,...,\{\phi_i'(s)\}_{s \in T},...,\mathscr{M}_n,t),t)].\] By (\ref{17+}), \[\mathbb{E}_t[u_i(\sigma_A(\mathscr{M}_1,...,\{\phi_i'(s)\}_{s \in T},...,\mathscr{M}_n,t),t)]>\mathbb{E}_t[u_i(\sigma_A(\mathscr{M}_1,...,\mathscr{M}_i,...,\mathscr{M}_n,t),t)],\] 
a contradiction to $(\mathscr{M}_1,...,\mathscr{M}_n, \sigma_A)$ is a PBE. Hence, for any $i, \phi_i^*$ solves problem (\ref{P2}).

Alternatively, if $\sigma_A$ is a mixed strategy that satisfies \eqref{contraction}, and $(\mathscr{M}_1,...,\mathscr{M}_n,\sigma_A)$ satisfies (\ref{3*}), then for any $i$, let $\phi_i^{**}(t)=o_i^t, \forall t$, where $\{o_i^t\}_{t \in T}$ is specified by condition (\ref{3*}). It is straightforward that for any $i, \{\phi_i^{**}(t)\}_{t \in T} \subseteq \mathscr{M}_i$. I first show that for any $i,$
\begin{equation}
    \label{19++}
    \mathbb{E}_t[u_i(\sigma_A(\mathscr{M}_1,...,\mathscr{M}_n,t),t)] \leq \mathbb{E}_t[u_i(\sigma_A(\mathscr{M}_1,...,\{\phi_i^{**}(t)\}_{t \in T}, ..., \mathscr{M}_n,t),t)].
\end{equation}

By \eqref{contraction}, for any $t, i,$
\[\supp \sigma_A(\mathscr{M}_1,...,\{\phi_i^{**}(t)\}_{t \in T},...,\mathscr{M}_i,t) = \supp \sigma_A(\mathscr{M}_1,...,\mathscr{M}_n,t) \cap \mathscr{M}_1 \times ... \times \{\phi_i^{**}(t)\}_{t \in T} \times ... \times \mathscr{M}_n.\]
Suppose there exists $t$ and $\bar{o}=(\bar{o}_1,...,\bar{o}_i,...,\bar{o}_n) \in \sigma_A(\mathscr{M}_1,...,\mathscr{M}_n,t) \cap \mathscr{M}_1 \times ... \times \{\phi_i^{**}(t)\}_{t \in T} \times ... \times \mathscr{M}_n$ such that \begin{equation} \label{20++} \bar{o}_i \notin \argmax \limits_{o_i \in \supp \sigma_A(\mathscr{M}_1,...,\mathscr{M}_n,t)\bigg|_{O_i}} \displaystyle{\int}_{O_{-i}} u_i(o_i, o_{-i},t)\sigma_A(\mathscr{M}_1,...,\mathscr{M}_n,t)(\mathbf{d}o_{-i}|o_i). \end{equation} Since $\bar{o} \in \mathscr{M}_1 \times ... \times \{\phi_i^{**}(t)\}_{t \in T} \times ... \times \mathscr{M}_n$, there exists $\bar{t} \in T$ such that $\bar{o}_i=o_i^{\bar{t}}$. If $\bar{t}=t$, then (\ref{20++}) is a direct violation of the first statement of (\ref{3*}). If $\bar{t} \neq t$, then by (\ref{20++}) and $\bar{o} \in \supp \sigma_A(\mathscr{M}_1,...,\mathscr{M}_n,t)$,
\begin{equation*}
    o_i^{\bar{t}} \in \supp \sigma_A(\mathscr{M}_1,...,\mathscr{M}_n,t)\bigg|_{O_i} \backslash \argmax \limits_{o_i \in \supp \sigma_A(\mathscr{M}_1,...,\mathscr{M}_n,t)\bigg|_{O_i}} \displaystyle{\int}_{O_{-i}} u_i(o_i, o_{-i},t)\sigma_A(\mathscr{M}_1,...,\mathscr{M}_n,t)(\mathbf{d}o_{-i}|o_i),
\end{equation*}
a contradiction to the second statement of (\ref{3*}). Hence, for any $t$, \begin{equation} \label{23+} \supp \sigma_A(\mathscr{M}_1,...,\{\phi_i^{**}(t)\}_{t \in T},...,\mathscr{M}_i,t)\bigg|_{O_i} \subseteq \argmax \limits_{o_i \in \supp \sigma_A(\mathscr{M}_1,...,\mathscr{M}_n,t)\bigg|_{O_i}} \displaystyle{\int}_{O_{-i}} u_i(o_i, o_{-i},t)\sigma_A(\mathscr{M}_1,...,\mathscr{M}_n,t)(\mathbf{d}o_{-i}|o_i). \end{equation} Claim that (\ref{23+}) implies (\ref{19++}): by (\ref{23+}), 
\begin{equation*}
    \begin{split}
        \mathbb{E}_t[u_i(\sigma_A(\mathscr{M}_1,...,\{\phi_i^{**}(t)\}_{t \in T},...,\mathscr{M}_n,t),t)] &= \mathbb{E}_t[\max \limits_{o_i \in \supp \sigma_A(\mathscr{M}_1,...,\mathscr{M}_n,t)\bigg|_{O_i}} \displaystyle{\int}_{O_{-i}} u_i(o_i, o_{-i},t)\sigma_A(\mathscr{M}_1,...,\mathscr{M}_n,t)(\mathbf{d}o_{-i}|o_i)] \\ & \geq \mathbb{E}_t[\displaystyle{\int}_{O_i} \displaystyle{\int}_{O_{-i}} u_i(o_i,o_{-i},t) \sigma_A(\mathscr{M}_1,...,\mathscr{M}_n,t)(\mathbf{d}o_{-i}|o_i) \tilde{\sigma_A}(\mathscr{M}_1,...,\mathscr{M}_n,t)(\mathbf{d}o_i)] \\ & = \mathbb{E}_t[\displaystyle{\int}_O u_i(o,t)\sigma_A(\mathscr{M}_1,...,\mathscr{M}_n,t)(\mathbf{d}o)] = \mathbb{E}_t[u_i(\sigma_A(\mathscr{M}_1,...,\mathscr{M}_n,t),t)],
    \end{split}
\end{equation*}
where $\tilde{\sigma_A}(\mathscr{M}_1,...,\mathscr{M}_n,t)=\sigma_A(\mathscr{M}_1,...,\mathscr{M}_n,t) \bigg|_{O_i}$. 

I proceed to show that for any $i, \phi_i^{**}$ solves (\ref{P2}). (\ref{3*}) ensures that for any $t, \phi_i^{**}(t) \in \argmax \limits_{o \in \supp \sigma_A(\mathscr{M}_1,...,\mathscr{M}_n,t)} u_i(o,t)\bigg|_{O_i}$, so that the feasibility of $\phi_i^{**}$ directly follows from
\begin{equation*}
\begin{split}
    v_i(\phi_i^{**}(t),t \mid \mathscr{M}_{-i})
    &= \max \limits_{o \in \mathscr{M}_1 \times \ldots \times \mathscr{M}_i \times \ldots \times \mathscr{M}_n} 
       V(o,t) \geq \max \limits_{o \in \mathscr{M}_1 \times \ldots \times 
        \{\phi_i^{**}(t)\}_{t \in T} \times \ldots \times \mathscr{M}_n} 
        V(o,t) \\
    &\geq v_i(\phi_i^{**}(t'), t \mid \mathscr{M}_{-i}), 
\end{split}
\end{equation*}
where the first $``\geq''$ is by $\mathscr{M}_1 \times ... \times \{\phi_i^{**}(t)\}_{t \in T} \times ... \times \mathscr{M}_n \subseteq \mathscr{M}_1 \times ... \mathscr{M}_i \times ... \times \mathscr{M}_n$. As for the optimality, suppose there exists $i$ and feasible $\phi_i'$ such that \[\mathbb{E}_t[u_i(\sigma_A(\mathscr{M}_1,...,\{\phi_i^{**}(t)\}_{t \in T},...,\mathscr{M}_n,t),t)]< \mathbb{E}_t[u_i(\sigma_A(\mathscr{M}_1,...,\{\phi_i'(s)\}_{s \in T},...,\mathscr{M}_n,t),t)],\]
then by (\ref{19++}), \begin{equation*} \mathbb{E}_t[u_i(\sigma_A(\mathscr{M}_1,...,\mathscr{M}_n,t),t)]< \mathbb{E}_t[u_i(\sigma_A(\mathscr{M}_1,...,\{\phi_i'(s)\}_{s \in T},...,\mathscr{M}_n,t),t)],\end{equation*} 
a contradiction to $(\mathscr{M}_1,...,\mathscr{M}_n,\sigma_A)$ is a PBE. As a result, for any $i, \phi_i^{**}$ solves (\ref{P2}).
    
\end{proof}

\subsection{Proof of Proposition \ref{prop:extendeddm1}}

     \begin{proof}
     It is straightforward that, by the definition of $\mathscr{M}_i$, for any $i$, \[\mathbb{E}_t\!\left[
u_i\!\left(
\sigma_A(
\{\phi_1^*(s,o_{-1})\}_{s\in T,\, o_{-1}\in O_{-1}},\ldots,
\{\phi_i(s,o_{-i})\}_{s\in T,\, o_{-i}\in O_{-i}},\ldots,
\{\phi_n^*(s,o_{-n})\}_{s\in T,\, o_{-n}\in O_{-n}},
t
),
t
\right)
\right]=\mathbb{E}_t[u_i(\sigma_A(\mathscr{M}_1,...,\mathscr{M}_n,t),t)].\]
         Suppose there exists $i$ and $\mathscr{M}_i'$ such that 
         \begin{equation*} \mathbb{E}_t[u_i(\sigma_A(\mathscr{M}_1,...,\mathscr{M}_n,t),t)]<\mathbb{E}_t[u_i(\sigma_A(\mathscr{M}_1,...,\mathscr{M}_i',...,\mathscr{M}_n,t),t)].
         \end{equation*}
         Consider $\phi_i': T \times O_{-i} \rightarrow O_i$ such that for any $o_{-i},t$, \[ \phi_i'(t,o_{-i}):=\sigma_A(\mathscr{M}_1,...,\mathscr{M}_i',...,\mathscr{M}_n,t)\big|_{O_i}.\]
         $\phi_i'$ is well-defined because $\sigma_A$ is a pure-strategy. Clearly, $\{\phi_i'(t,o_{-i})\}_{t \in T, o_{-i} \in O_{-i}} \subseteq \mathscr{M}_i'$, which, by similar arguments made in the proof of Proposition \ref{p0}, gives the feasibility of $\phi_i'$. Also similar to the arguments made in the proof of Proposition \ref{p0}, by \eqref{piia},
         \begin{equation*}
\begin{split}
    &\mathbb{E}_t\!\left[
u_i\!\left(
\sigma_A(
\{\phi_1^*(s,o_{-1})\}_{s\in T,\, o_{-1}\in O_{-1}},\ldots,
\{\phi_i(s,o_{-i})\}_{s\in T,\, o_{-i}\in O_{-i}},\ldots,
\{\phi_n^*(s,o_{-n})\}_{s\in T,\, o_{-n}\in O_{-n}},
t
),
t
\right)
\right] \\
    &= \mathbb{E}_t[u_i(\sigma_A(\mathscr{M}_1,\dots,\mathscr{M}_n,t),t)] \\
    &< \mathbb{E}_t[u_i(\sigma_A(\mathscr{M}_1,\dots,\mathscr{M}_i',\dots,\mathscr{M}_n,t),t)] \\
    &= \mathbb{E}_t[u_i(\sigma_A(\mathscr{M}_1,\dots,\{\phi_i'(s, o_{-i})\}_{s \in T, o_{-i} \in O_{-i}},\dots,\mathscr{M}_n,t),t)],
\end{split}
\end{equation*}
a contradiction to $\phi_i^*$ solves the problem. Hence, $(\mathscr{M}_1,...,\mathscr{M}_n, \sigma_A)$ is a PBE. 
     \end{proof}

\subsection{Proof of Proposition \ref{prop:extendeddm2}}

\begin{proof}
    For any $t$, let $\sigma_A(\mathscr{M}_1,...,\mathscr{M}_n,t):=(o_1^t,...,o_n^t)$, which is well-defined as $\sigma_A$ is a pure-strategy. For any $i$, consider $\phi_i^*: T \times O_{-i} \rightarrow O_i$ such that for any $t$ and $o_{-i}^t$,
    \begin{equation} \label{eq:extendedonpath} \phi_i^*(t,o_{-i}^t):=o_i^t. \end{equation}
    In addition, since $\card(\mathscr{M}_i) \leq \card (O_{-i})$, for a fixed $t_0 \in T$,
    \[ \card (\mathscr{M}_i \backslash \{o_i^{t_0}\}) \leq \card (O_{-i} \backslash \{o_{-i}^{t_0}\}),\]
    which further implies that there exists an injection $\iota_i: \mathscr{M}_i \backslash \{o_i^t\}_{t \in T} \rightarrow O_{-i} \backslash \{o_{-i}^{t_0}\}$. For any $o_i \in \mathscr{M}_i \backslash \{o_i^t\}_{t \in T}$, let 
    \begin{equation} \label{eq:extendeeoffpath} \phi_i^*(t_0,\iota_i(o_i)):=o_i.\end{equation}
    For any other $(t,o_{-i})$ that is not specified by \eqref{eq:extendedonpath} or \eqref{eq:extendeeoffpath}, let $\phi_i^*(t, o_{-i})$ be an arbitrary element of $\mathscr{M}_i$. In that case, $\{\phi_i^*(s,o_{-i})\}_{s \in T, o_{-i} \in O_{-i}}=\mathscr{M}_{-i}$. 

    I proceed to show that for any $i, \phi_i^*$ constructed above is feasible to the problem in Proposition \ref{prop:extendeddm2}. For any $t$, since $\sigma_A$ is a pure-strategy, $\supp \sigma_A(\mathscr{M}_1,...,\mathscr{M}_n,t)\big|_{O_{-i}}=\{o_{-i}^t\}$. In addition, 
    \[ \phi_i^*(t, o_{-i}^t)=o_i^t \in \argmax \limits_{o \in \mathscr{M}_1 \times ... \times \mathscr{M}_n} V(o,t) \big|_{O_i},\]
    which directly gives for any $t', o_{-i}',$
    \[ v_i(\phi_i^*(t,o_{-i}^t),t|\mathscr{M}_{-i}) \geq v_i(\phi_i^*(t', o_{-i}'),t|\mathscr{M}_{-i}).\]
    Suppose the constructed $\phi_i^*$ does not solve the problem in Proposition \ref{prop:extendeddm2}, i.e., there exists $\phi_i'$ that is feasible and
   \[ \mathbb{E}_t\!\left[
u_i\!\left(
\sigma_A(\mathscr{M}_1,\ldots,\{\phi^*_i(s,o_{-i})\}_{s\in T,\, o_{-i}\in O_{-i}},\ldots,\mathscr{M}_n,t),
t
\right)
\right] < \mathbb{E}_t\!\left[
u_i\!\left(
\sigma_A(\mathscr{M}_1,\ldots,\{\phi'_i(s,o_{-i})\}_{s\in T,\, o_{-i}\in O_{-i}},\ldots,\mathscr{M}_n,t),
t
\right)
\right]. \]
Since $\{\phi_i^*(s, o_{-i})\}_{s \in T, o_{-i} \in O_{-i}}=\mathscr{M}_i, \{\phi_i'(s, o_{-i})\}_{s \in T, o_{-i} \in O_{-i}}$ is a profitable deviation for principal $i$, a contradiction to $(\mathscr{M}_1,...,\mathscr{M}_n,\sigma_A)$ is a PBE. 
\end{proof}

\section{Supplementary Examples}

\subsection{Examples for Remark \ref{rmk:upnr}}
\label{sec:examplesrmk1}

I first show an example that there exists $(\phi_1^*,...,\phi_n^*)$ such that for any $i, \phi_i^*$ solves \eqref{P3} and \eqref{2+} is satisfied but $(\phi_1^{**},...,\phi_n^{**})$, where for any $i,t,$
    \[ \phi_i^{**}(t)=o_i^t,\]
    where $o_i^t$ is specified in \eqref{2+}, does not satisfy for any $i, \phi_i^{**}$ solves \eqref{P3}.

    \begin{example}
\label{ex:upr-construction-counterexample}
Consider a setting with two principals $\{1,2\}$. The agent has two types $T=\{t_1,t_2\}$.
Let
\[
O_1=\{a,c\},\qquad O_2=\{b,b'\}.
\]
The agent's utility is given by the following tables:

\begin{minipage}{0.45\textwidth}
\centering
\[
\begin{array}{c|cc}
 & b & b' \\ \hline
a & 5 & 5 \\
c & 0 & 10 \\
 & t_1
\end{array}
\]
\end{minipage}
\qquad
\begin{minipage}{0.45\textwidth}
\centering
\[
\begin{array}{c|cc}
 & b & b' \\ \hline
a & 0 & 0 \\
c & 10 & 10 \\
 & t_2
\end{array}
\]
\end{minipage}

Principal 1's utility is
\[
u_1(a,t_1)=1,\quad u_1(c,t_1)=-100,\qquad
u_1(a,t_2)=1,\quad u_1(c,t_2)=5,
\]
and principal 2's utility is
\[
u_2(b,t_1)=1,\quad u_2(b',t_1)=0,\qquad
u_2(b,t_2)=1,\quad u_2(b',t_2)=1.
\]

Consider the profile
\[
\phi_1^*(t_1)=a,\qquad \phi_1^*(t_2)=a,
\]
\[
\phi_2^*(t_1)=b,\qquad \phi_2^*(t_2)=b'.
\]
It is straightforward to check that for any $i \in \{1,2\}, \phi_i^*$ solves \eqref{P3}.

Moreover, \eqref{2+} is satisfied. In particular, one can choose
\[
(o_1^{t_1},o_2^{t_1})=(a,b),\qquad
(o_1^{t_2},o_2^{t_2})=(a,b),
\]
which are utility-maximizing (within the menu profile) for the agent at the corresponding types and preserve both principals' utilities relative to $(\phi_1^*,\phi_2^*)$.

Now define $(\phi_1^{**},\phi_2^{**})$ by
\[
\phi_i^{**}(t)=o_i^t.
\]
Then
\[
\phi_1^{**}(t_1)=a,\qquad \phi_1^{**}(t_2)=a,
\]
\[
\phi_2^{**}(t_1)=b,\qquad \phi_2^{**}(t_2)=b.
\]
Hence the induced rival menu for principal 1 becomes $\{b\}$ instead of $\{b,b'\}$.

Given the rival menu $\{b\}$, principal 1 can deviate to
\[
\psi_1(t_1)=a,\qquad \psi_1(t_2)=c,
\]
which is feasible for \eqref{P3} since
\[
v_1(a,t_1\mid\{b\})=5>0=v_1(c,t_1\mid\{b\}),
\]
and
\[
v_1(c,t_2\mid\{b\})=10>0=v_1(a,t_2\mid\{b\}).
\]
Its expected payoff is
\[
\frac{1}{2}u_1(a,t_1)+\frac{1}{2}u_1(c,t_2)
=\frac{1}{2}(1)+\frac{1}{2}(5)=3,
\]
whereas $(\phi_1^{**},\phi_2^{**})$ yields principal 1 expected payoff
\[
\frac{1}{2}u_1(a,t_1)+\frac{1}{2}u_1(a,t_2)
=\frac{1}{2}(1)+\frac{1}{2}(1)=1.
\]
Therefore, although $(\phi_1^*,\phi_2^*)$ solves \eqref{P3} and satisfies \eqref{2+}, the profile obtained by the construction $\phi_i^{**}(t)=o_i^t$ does not necessarily solve \eqref{P3}.
\end{example}

I proceed to show an example that directly contradicts the claim in Remark \ref{rmk:upnr}.

\begin{example}
\label{ex:upr-not-upnr}
Consider a setting with two principals $\{1,2\}$. The agent has three types
\[
T=\{t_1,t_2,t_3\},
\]
each occurring with probability $\frac{1}{3}$. Let
\[
O_1=\{a_1,a_2,a_3\},\qquad O_2=\{b_1,b_2,b_3\}.
\]

The agent's utility is given by the following tables:

\begin{minipage}{0.3\textwidth}
\centering
\[
\begin{array}{c|ccc}
 & b_1 & b_2 & b_3 \\ \hline
a_1 & 1 & 2 & 0 \\
a_2 & 1 & 1 & 0 \\
a_3 & 2 & 1 & 0 \\
 & t_1
\end{array}
\]
\end{minipage}
\qquad
\begin{minipage}{0.3\textwidth}
\centering
\[
\begin{array}{c|ccc}
 & b_1 & b_2 & b_3 \\ \hline
a_1 & 1 & 1 & 2 \\
a_2 & 2 & 1 & 0 \\
a_3 & 0 & 1 & 0 \\
 & t_2
\end{array}
\]
\end{minipage}
\qquad
\begin{minipage}{0.3\textwidth}
\centering
\[
\begin{array}{c|ccc}
 & b_1 & b_2 & b_3 \\ \hline
a_1 & 1 & 2 & 0 \\
a_2 & 2 & 2 & 2 \\
a_3 & 2 & 2 & 0 \\
 & t_3
\end{array}
\]
\end{minipage}

Principal 1's utility, which is only $(a_i,t_i)$-dependent, is
\[
\begin{array}{c|ccc}
 & a_1 & a_2 & a_3 \\ \hline
t_1 & 2 & 1 & 0 \\
t_2 & 0 & 2 & 0 \\
t_3 & 1 & 0 & 0
\end{array}
\]
and principal 2's utility is
\[
\begin{array}{c|ccc}
 & b_1 & b_2 & b_3 \\ \hline
t_1 & 0 & 1 & 0 \\
t_2 & 2 & 2 & 0 \\
t_3 & 1 & 0 & 0
\end{array}.
\]

Consider
\[
\phi_1^*(t_1)=a_2,\qquad \phi_1^*(t_2)=a_2,\qquad \phi_1^*(t_3)=a_3,
\]
\[
\phi_2^*(t_1)=b_2,\qquad \phi_2^*(t_2)=b_2,\qquad \phi_2^*(t_3)=b_3.
\]

A direct check shows that $(\phi_1^*,\phi_2^*)$ solves \eqref{P3}. Moreover, \eqref{2+} is satisfied:
\begin{itemize}
    \item at $t_1$, choose $(o_1^{t_1},o_2^{t_1})=(a_2,b_2)$;
    \item at $t_2$, choose $(o_1^{t_2},o_2^{t_2})=(a_2,b_2)$;
    \item at $t_3$, choose $(o_1^{t_3},o_2^{t_3})=(a_2,b_2)$.
\end{itemize}
Indeed, for each type $t$, the chosen pair lies in
\[
\argmax_{o_i \in \{\phi_i^*(s)\}_{s \in T},\, i=1,2} V((o_1,o_2),t),
\]
and for each principal $i$,
\[
u_i(o_i^t,t)=u_i(\phi_i^*(t),t).
\]
Hence, there exists a profile where the menus within it mutually solve \eqref{P3} and satisfying \eqref{2+}.

I now claim that there does not exist any profile $(\psi_1,\psi_2)$ such that both principals solve \eqref{P3} and
\[
(\psi_1(t),\psi_2(t))
\in
\argmax_{o_i \in \{\psi_i(s)\}_{s \in T},\, i=1,2}
V((o_1,o_2),t)
\qquad \forall t.
\]

The only profiles where the menus within them mutually solve \eqref{P3} are
\[
(\psi_1,\psi_2)=\bigl((a_1,a_2,a_1),(b_2,b_1,b_1)\bigr)
\]
and
\[
(\psi_1,\psi_2)=\bigl((a_2,a_2,a_3),(b_2,b_2,b_3)\bigr).
\]
Neither of these satisfies the aforementioned condition:
\begin{itemize}
    \item for $\bigl((a_1,a_2,a_1),(b_2,b_1,b_1)\bigr)$, at type $t_3$ the chosen pair
    $(a_1,b_1)$ yields utility $1$, while the maximum over
    $\{a_1,a_2\}\times\{b_1,b_2\}$ is $2$;
    \item for $\bigl((a_2,a_2,a_3),(b_2,b_2,b_3)\bigr)$, at type $t_3$ the chosen pair
    $(a_3,b_3)$ yields utility $0$, while the maximum over
    $\{a_2,a_3\}\times\{b_2,b_3\}$ is $2$.
\end{itemize}
\end{example}

\subsection{Example for Pareto Improvements of \eqref{P3}-induced PBE}
\label{sec:paretoimprovements}

\begin{example}[A counterexample to Pareto optimality of \eqref{P3}-induced PBE under \eqref{eq:noindifference}]
\label{ex:pareto-counterexample}
Consider a setting with two principals $\{1,2\}$. The agent has two types $T=\{t_1,t_2\}$, each occurring with probability $\frac{1}{2}$. Let
\[
O_1=\{a,b,c\},\qquad O_2=\{A,B,C\}.
\]

The agent's utility is given by the following tables:

\begin{minipage}{0.45\textwidth}
\centering
\[
\begin{array}{c|ccc}
 & A & B & C \\ \hline
a & 5 & 2 & 0 \\
b & 1 & 7 & 6 \\
c & 4 & 8 & 3 \\
 & t_1
\end{array}
\]
\end{minipage}
\qquad
\begin{minipage}{0.45\textwidth}
\centering
\[
\begin{array}{c|ccc}
 & A & B & C \\ \hline
a & 5 & 8 & 7 \\
b & 2 & 3 & 4 \\
c & 6 & 1 & 0 \\
 & t_2
\end{array}
\]
\end{minipage}

For each type, all nine utility levels are distinct. Hence \eqref{eq:noindifference} holds.

Principal 1's utility is
\[
u_1(a,t_1)=4,\quad u_1(b,t_1)=0,\quad u_1(c,t_1)=3,
\]
\[
u_1(a,t_2)=1,\quad u_1(b,t_2)=0,\quad u_1(c,t_2)=4,
\]
and principal 2's utility is
\[
u_2(A,t_1)=4,\quad u_2(B,t_1)=2,\quad u_2(C,t_1)=1,
\]
\[
u_2(A,t_2)=0,\quad u_2(B,t_2)=1,\quad u_2(C,t_2)=3.
\]

Consider the singleton menu profile
\[
(\{c\},\{C\}).
\]
This profile is \eqref{P3}-induced: given $\{C\}$, principal 1's problem is solved by the constant schedule $t\mapsto c$; given $\{c\}$, principal 2's problem is solved by the constant schedule $t\mapsto C$; and with singleton menus, \eqref{2+} is automatic.

The corresponding expected payoff vector is
\[
\left(
\frac{u_1(c,t_1)+u_1(c,t_2)}{2},
\frac{u_2(C,t_1)+u_2(C,t_2)}{2}
\right)
=
\left(
\frac{3+4}{2},
\frac{1+3}{2}
\right)
=
(3.5,2).
\]

Now consider the menu profile
\[
(\{a,b,c\},\{A\}).
\]
Given principal 2's singleton menu $\{A\}$, the agent chooses
\[
a \text{ at type } t_1,\qquad c \text{ at type } t_2,
\]
since
\[
V((a,A),t_1)=5>\max\{1,4\},\qquad
V((c,A),t_2)=6>\max\{5,2\}.
\]
Thus the induced payoff vector is
\[
\left(
\frac{u_1(a,t_1)+u_1(c,t_2)}{2},
\frac{u_2(A,t_1)+u_2(A,t_2)}{2}
\right)
=
\left(
\frac{4+4}{2},
\frac{4+0}{2}
\right)
=
(4,2).
\]

It is straightforward to verify that $(\{a,b,c\},\{A\})$ is a PBE: given $\{A\}$, principal 1 already attains her maximal feasible payoff at each type, and given $\{a,b,c\}$, principal 2 cannot obtain expected payoff exceeding $2$.

Therefore, $(\{a,b,c\},\{A\})$ Pareto dominates $(\{c\},\{C\})$, since
\[
(4,2)\ge (3.5,2)
\]
with strict inequality for principal 1.

Hence, even under \eqref{eq:noindifference}, a PBE that is Pareto optimal among all \eqref{P3}-induced PBE need not be Pareto optimal.
\end{example}

\subsection{Example for Outside Options}
\label{appendixegoo}
\begin{example}
\label{e2+}
Consider an example with the exactly same setup as Example \ref{e1+} but with different payoff tables. It is straightforward that for any $i \in \{1,2\}, u_i$ only depends on $o_i$. 

\begin{minipage}{0.45\textwidth}
\centering
\[
\begin{array}{c|cc}
 & b & b' \\ \hline
a & (5,5,5) & (5,0,10) \\
a' & (0,5,10) & (0,0,0) \\
 & t_1  &
\end{array}
\]
\end{minipage}
\quad
\begin{minipage}{0.45\textwidth}
\centering
\[
\begin{array}{c|cc}
 & b & b' \\ \hline
a & (5,5,-2) & (5,0,-2) \\
a' & (0,5,-2) & (0,0,0) \\
 & t_2 &
\end{array}
\]
\end{minipage}

I first consider intrinsic common agency. Claim that $\phi_1^*(t_1)=a, \phi_1^*(t_2)=a';\phi_2^*(t_1)=b, \phi_2^*(t_2)=b'$ solves \eqref{P4'} for any distribution on $T$. I first check the feasibility. IC and \eqref{i} for $t_1$ is guaranteed by
\begin{equation*}
    \begin{split}
        v_1(a,t_1|\{b,b'\})=10 \geq 10=v_1(a',t_1|\{b,b'\}) \geq 0, \\
        v_2(b,t_1|\{a,a'\})=10 \geq 10=v_2(b',t_1|\{a,a'\}) \geq 0,
    \end{split}
\end{equation*}
and IC and \eqref{i} for $t_2$ is guaranteed by
\begin{equation*}
    \begin{split}
        v_1(a',t_2|\{b,b'\}) \geq (=) 0 \geq -2=v_1(a,t_2|\{b,b'\}), \\
        v_2(b',t_2|\{a,a'\}) \geq (=) 0 \geq -2=v_2(b,t_2|\{a,a'\}).
    \end{split}
\end{equation*}

The optimality of $(\phi_1^*,\phi_2^*)$ can be argued through the following: under $t_1$, $\{a\}$ and $\{b\}$ strictly dominates $\{a'\}$ and $\{b'\}$ respectively, and under $t_2$, allocating $\{a'\}$ and $\{b'\}$ are the unique deterministic mechanism that satisfies \eqref{i} for principal 1 and 2 respectively. 

However, $(\phi_1^*, \phi_2^*)$ does not satisfy \eqref{eq:4i}, for exactly the same reason demonstrated in Example \ref{e1+}. 

I proceed to show that, for almost all distributions on $T$, there does not exist $\sigma_A$ such that $(\{a,a'\}, \{b,b'\}, \sigma_A)$ is a PBE. Given principal 2 plays $\{b,b'\}$, for any $\sigma_A$ that satisfies condition \eqref{1i}, 
\begin{equation*}
\begin{split}
    u_1(\sigma_A(\{a\}, \{b,b'\}, t_1)) &= 5; \quad 
    u_1(\sigma_A(\{a\}, \{b,b'\}, t_2)) = 0, \\
    u_1(\sigma_A(\{a'\}, \{b,b'\}, t_1)) &= 0; \quad 
    u_1(\sigma_A(\{a'\}, \{b,b'\}, t_2)) = 0, \\
    u_1(\sigma_A(\{a,a'\}, \{b,b'\}, t_1)) &= 5p_A; \quad 
    u_1(\sigma_A(\{a,a'\}, \{b,b'\}, t_2)) = 0.
\end{split}
\end{equation*}
where $\sigma_A(\{a,a'\}, \{b,b'\},t_1)=p_A(a,b')+(1-p_A)(a',b)$. Similarly, given principal 1 plays $\{a,a'\}$, 
\begin{equation*}
\begin{split}
    u_2(\sigma_A(\{a,a'\}, \{b\}, t_1)) &= 5; \quad
    u_2(\sigma_A(\{a,a'\}, \{b\}, t_2)) = 0, \\
    u_2(\sigma_A(\{a,a'\}, \{b'\}, t_1)) &= 0; \quad
    u_2(\sigma_A(\{a,a'\}, \{b'\}, t_2)) = 0, \\
    u_2(\sigma_A(\{a,a'\}, \{b,b'\}, t_1)) &= 5 - 5p_A; \quad
    u_2(\sigma_A(\{a,a'\}, \{b,b'\}, t_2)) = 0.
\end{split}
\end{equation*}
Let $Pr(t=t_1)=p$, then if there exists $\sigma_A$ such that $(\{a,a'\}, \{b,b'\}, \sigma_A)$ is a PBE, then there exists $p_A$ such that
\begin{equation*}
    \begin{split}
        p\cdot 5p_A \geq 5p, \\
        p \cdot 5(1-p_A) \geq 5p,
    \end{split}
\end{equation*}
which implies that $p_A$ exists if and only if $p=0$. In that case, for $p > 0$, there does not exist any $p_A$ that makes the inequalities hold, thus there does not exist $\sigma_A$ such that $(\{a,a'\},\{b,b'\}, \sigma_A)$ is a PBE. 

As for delegated common agency, let $a'$ and $b'$ be the agent's outside option at principal 1 and 2 respectively.\footnote{This corresponds to the natural interpretation of ``outside options'', as $V((a',b'),t_i)=0, \forall i \in \{1,2\}$.} Then, all the arguments follow.

\end{example}

\end{spacing}

\end{document}